\documentclass[acmsmall,screen]{acmart}
\pdfoutput=1
\setcopyright{rightsretained}
\acmYear{2023}
\acmVolume{7}
\acmJournal{PACMPL}
\acmNumber{PLDI}
\acmArticle{112}
\acmMonth{6}
\acmPrice{}
\acmDOI{10.1145/3591226}
\copyrightyear{2023}
\acmSubmissionID{pldi23main-p50-p}
\received{2022-11-10}
\received[accepted]{2023-03-31}

\usepackage{mathpartir,relsize,enumitem,subcaption,wrapfig,tikz,pgfplots,listings,bbm}
\newtheorem{definition}{Definition}[section]
\newtheorem{corollary}[definition]{Corollary}
\newtheorem{lemma}[definition]{Lemma}

\newtheorem{thm}[definition]{Theorem}

\definecolor{lilac}{RGB}{128,49,167} 
\newcommand\assert[1]{{\color{lilac}#1}}
\newcommand\assertc[1]{\assert{\left\{#1\right\}}}
\usetikzlibrary{patterns,patterns.meta}

\newcommand\mathextralarger[1]{\mathlarger{\mathlarger{\mathlarger{#1}}}}

\ifdefined\cameraready
\newcommand\referto[1]{\citet{li2023lilac}}
\newcommand\camerareadyonly[1]{#1}
\else
\newcommand\referto[1]{Appendix~\ref{#1}}
\newcommand\camerareadyonly[1]{}
\fi

\newcommand\applcolor[1]{{\color{teal}#1}}
\newcommand\applkw[1]{\texttt{\applcolor{#1}}}

\newcommand\vdashrv{\vdash_{{\textrm{rv}}}}
\newcommand\vdashdet{\vdash_{{\textrm{det}}}}
\newcommand\vdashprog{\vdash_{{\textrm{prog}}}}
\newcommand\vdashappl{\vdash_{{\textrm{APPL}}}}

\newcommand{\calM}{\mathcal M}
\newcommand\bfone{\mathbf 1}
\newcommand\pmbone{{\pmb{\pmb1}}}
\newcommand\Mfin{\calM_\mathrm{finite}}
\newcommand\Mdis{\calM_\mathrm{disintegrable}}
\newcommand\plte{\sqsubseteq}
\newcommand{\N}{\mathbb N}
\newcommand{\R}{\mathbb R}
\newcommand{\hilbertcube}{[0,1]^{\N}}
\newcommand{\sembr}[1]{\llbracket #1\rrbracket}
\newcommand{\angled}[1]{\langle #1\rangle}

\newcommand{\calA}{\mathcal A}
\newcommand{\calB}{\mathcal B}

\newcommand{\calE}{\mathcal E}
\newcommand{\calF}{\mathcal F}
\newcommand{\calG}{\mathcal G}
\newcommand{\giry}{\mathcal G}
\newcommand{\calP}{\mathcal P}
\newcommand{\calQ}{\mathcal Q}
\newcommand{\calR}{\mathcal R}
\newcommand{\indep}{\perp \!\!\! \perp}

\newcommand{\ber}{\operatorname{Ber}}
\newcommand{\ind}{\mathbf 1}
\DeclareMathOperator*{\Ex}{{\mathbb{E}}}
\newcommand{\inv}{^{-1}}

\newcommand{\pow}[2]{\operatorname{pow}\left(#1, #2\right)}

\newcommand{\Bigpow}[2]{\operatorname{pow}\Big(#1, #2\Big)}

\newcommand\monadic[1]{\left(~\begin{aligned}#1\end{aligned}~\right)}
\newcommand\rv{\operatorname{RV}}
\newcommand\goodparagraph[1]{\hspace{0em}\\[-0.5em]\indent\textbf{\textit{#1.~}}}

\makeatletter
\newcommand*\pdot{\mathpalette\pdot@{.8}}
\newcommand*\pdot@[2]{\mathbin{\vcenter{\hbox{\scalebox{#2}{$\m@th#1\bullet$}}}}}
\makeatother

\newcommand\wand{-\!\!*~}

\newcommand\weakpre{\mathsf{wp}}
\newcommand\nec{\operatorname{\square}}
\newcommand\own{\operatorname{\mathsf{own}}}

\newcommand\pdotp{\pdot\!'}

\newcommand\pto{\rightharpoonup}
\newcommand\mto{\stackrel{\mathsmaller{\mathrm{m}\,}}\to}
\newcommand\vdashiff{\dashv\vdash}

\newcommand{\indexty}{\applkw{index}}
\newcommand{\boolty}{\applkw{bool}}
\newcommand{\realty}{\applkw{real}}
\newcommand\ofty{\mathrm{:}}
\newcommand\asequal{\mathbin{\stackrel{\mathsf{as}}{\mathop{=}}}}
\newcommand\forallrv{\forall_\mathsf{rv}}
\newcommand\existsrv{\exists_\mathsf{rv}}
\newcommand\arbitrary{\mathrm{arbitrary}}

\newcommand\applunif{\applkw{unif }\texttt{[0,1]}}
\newcommand\unif{\operatorname{Unif}}

\newcommand\appleq{\gets}
\newcommand\appllet{}
\newcommand\applin{\hspace{0.1em};}
\newcommand\letin[3]{\appllet #1\appleq #2 \applin~#3}
\newcommand\vtrue{\applkw{T}}
\newcommand\vfalse{\applkw{F}}
\newcommand\ite[3]{\mathrm{if}~#1~\mathrm{then}~#2~\mathrm{else}~#3}
\newcommand\applite[3]{\applkw{if}~#1~\applkw{then}~#2~\applkw{else}~#3}
\newcommand\applif{\applkw{if}}
\newcommand\applthen{\applkw{then}}
\newcommand\applelse{\applkw{else}}
\newcommand\new{\applkw{new}}
\newcommand\flip{\applkw{flip}}

\newcommand\forloop[3]{\mathrm{for}(#1,#2,#3)}
\newcommand\applfor[5]{\applkw{for}(#1,#2,#3~#4.~#5)}
\newcommand\applforopen[5]{\applkw{for}(#1,#2,#3~#4.~#5}
\newcommand\appltuple[1]{(#1)}
\newcommand\applfst{\applkw{fst}~}
\newcommand\applsnd{\applkw{snd}~}
\newcommand\applexp{\,\string^\,}
\newcommand\applpar[1]{(#1)}
\newcommand\applindex[2]{#1[#2]}
\newcommand\applarray[1]{[#1]}
\newcommand\applforclose{)}
\newcommand\applgiry{\operatorname{\applkw{G}}}
\newcommand\applret{\applkw{ret}~}
\newcommand\ret{\mathrm{ret}~}

\newcommand\triple[4]{\{#1\}~#2~\{#3.\,#4\}}
\newcommand\lrtriple[4]{\left\{#1\right\}~#2~\left\{#3.\,#4\right\}}

\DeclareMathOperator*{\hugestar}{{\mathop{\scalebox{2.5}{\raisebox{-0.2ex}{$\ast$}}}}}
\DeclareMathOperator*{\argmax}{arg\,max}
\DeclareMathOperator*{\D}{{\pmb{\mathlarger{\mathsf{C}}}}}
\newcommand\Down{\D^{\own}}
\newcommand\DIndepName{C-Indep}
\newcommand\DIndep{\textsc{\DIndepName}}

\newcommand\rectbaseline{0.15\baselineskip}
\newcommand\rectwidth{0.35\baselineskip}
\newcommand\rectheight{0.7\baselineskip}

\newcommand\leftrect{
  \tikz[baseline=\rectbaseline]{
    \draw[fill=blue,opacity=0.25](0,0)rectangle(\rectwidth,\rectheight){};
    \draw[opacity=0.25](\rectwidth,0)rectangle(\rectheight,\rectheight){};
  }
}
\newcommand\rightrect{
  \tikz[baseline=\rectbaseline]{
    \draw[opacity=0.25](0,0)rectangle(\rectwidth,\rectheight){};
    \draw[fill=orange,opacity=0.25](\rectwidth,0)rectangle(\rectheight,\rectheight){};
  }
}
\newcommand\fullsqlr{
  \tikz[baseline=\rectbaseline]{
    \draw[fill=blue,opacity=0.25](0,0)rectangle(\rectwidth,\rectheight){};
    \draw[fill=orange,opacity=0.25](\rectwidth,0)rectangle(\rectheight,\rectheight){};
  }
}
\newcommand\fullsqlrbig{
  \tikz[baseline=1.5ex]{
    \draw[fill=blue,opacity=0.25](0,0)rectangle(0.25,0.5){};
    \draw[fill=orange,opacity=0.25](0.25,0)rectangle(0.5,0.5){};
  }
}

\newcommand\toprect{
  \tikz[baseline=\rectbaseline]{
    \draw[pattern=north west lines,opacity=0.2] (0,0) rectangle (\rectheight,\rectwidth) {};
    \draw[pattern=dots] (0,\rectheight) rectangle (\rectheight,\rectwidth) {};
  }
}
\newcommand\botrect{
  \tikz[baseline=\rectbaseline]{
    \draw[pattern=north west lines] (0,0) rectangle (\rectheight,\rectwidth) {};
    \draw[pattern=dots,opacity=0.2] (0,\rectheight) rectangle (\rectheight,\rectwidth) {};
  }
}
\newcommand\fullsqtb{
  \tikz[baseline=\rectbaseline]{
    \draw[pattern=north west lines] (0,0) rectangle (\rectheight,\rectwidth) {};
    \draw[pattern=dots] (0,\rectheight) rectangle (\rectheight,\rectwidth) {};
  }
}
\newcommand\fullsqlrtb{
  \tikz[baseline=\rectbaseline]{
    \draw[pattern=north west lines] (0,0) rectangle (\rectheight,\rectwidth) {};
    \draw[pattern=dots] (0,\rectheight) rectangle (\rectheight,\rectwidth) {};
    \draw[fill=blue,opacity=0.25](0,0)rectangle(\rectwidth,\rectheight){};
    \draw[fill=orange,opacity=0.25](\rectwidth,0)rectangle(\rectheight,\rectheight){};
  }
}
\newcommand\topleftsq{
  \tikz[baseline=\rectbaseline]{
    \draw[pattern=dots] (0,\rectwidth) rectangle (\rectwidth,\rectheight) {};
    \draw[fill=blue,opacity=0.25] (0,\rectwidth) rectangle (\rectwidth,\rectheight) {};
    \draw[opacity=0.25] (0,0) rectangle (\rectwidth,\rectwidth) {};
    \draw[opacity=0.25] (\rectwidth,0) rectangle (\rectheight,\rectwidth) {};
    \draw[opacity=0.25] (\rectwidth,\rectwidth) rectangle (\rectheight,\rectheight) {};
  }
}

\title{Lilac: A Modal Separation Logic for Conditional Probability}
\author{John M. Li}
\affiliation{%
  \institution{Northeastern University}
  \city{Boston, MA}
  \country{USA}}
\email{li.john@northeastern.edu}

\author{Amal Ahmed}
\affiliation{%
  \institution{Northeastern University}
  \city{Boston, MA}
  \country{USA}}
\email{a.ahmed@northeastern.edu}

\author{Steven Holtzen}
\affiliation{%
  \institution{Northeastern University}
  \city{Boston, MA}
  \country{USA}}
\email{s.holtzen@northeastern.edu}

\begin{abstract}
  We present Lilac, a separation logic for reasoning about probabilistic
  programs where separating conjunction captures probabilistic independence.
  Inspired by an analogy with mutable state where sampling corresponds to
  dynamic allocation, we show how probability spaces over a fixed, ambient
  sample space appear to be the natural analogue of heap fragments, and present
  a new combining operation on them such that probability spaces behave like 
  heaps and measurability of random variables behaves like ownership.
  This combining operation forms the basis for our model of separation,
  and produces a logic with many pleasant properties. In particular, Lilac has a
  frame rule identical to the ordinary one, and naturally
  accommodates advanced features like continuous random variables and reasoning
  about quantitative properties of programs. Then we propose a new modality
  based on disintegration theory for reasoning about conditional probability.
  We show how the resulting modal logic validates examples from prior work,
  and give a formal verification of an intricate weighted sampling algorithm
  whose correctness depends crucially on conditional independence structure.%
  \camerareadyonly{\footnote{For the extended version of this paper, with proofs, see \citet{li2023lilac}.}}
\end{abstract}

\begin{document}


\maketitle

\section{Introduction}
\label{sec:intro}
Software systems involving probability are pervasive. Such systems naturally
appear in diverse domains such as network reliability
analysis~\citep{smolka2019scalable,gehr2018bayonet}, reliability for
cyberphysical systems~\citep{lee2016introduction,holtzen2021model}, distributed software
systems~\citep{tassarotti2019separation}, and many others.  As these systems are
increasingly deployed in high-consequence domains, there is a growing need for
formal frameworks capable of reasoning about and verifying probabilistic
correctness properties.  We are especially interested in formal frameworks that
support \emph{compositional} reasoning: putting probabilistic systems together
correctly is a tricky business, as a probabilistic component often makes subtle
assumptions about the distribution of its inputs.  A formal framework for
reasoning about probabilistic systems should facilitate the sound composition of
formal verifications of individual components.

In the traditional non-probabilistic setting, \emph{program logics} have become
standard kit for compositionally reasoning about heap-manipulating programs at
scale~\citep{distefano2019scaling}.  In particular, \emph{separation logic}
enables modular reasoning about heap-manipulating
programs~\citep{ishtiaq2001bi,o2009separation,reynolds2009introduction,reynolds2002separation}.
The key to this
modularity is the \emph{frame rule},
\begin{align}
  \frac{\triple Pex{Q(x)}}{\triple{F*P}ex{F*Q(x)}},
\end{align} 
which states that a program $e$ satisfying precondition $P$ and
postcondition $Q$ doesn't interfere with any parts of the heap $F$ (``frames'')
disjoint from parts of the heap described by $P$~\citep{o2012primer}.  This
facilitates local reasoning, and is the distinctive advantage of using the
substructural separation logic over ordinary predicate logic in reasoning about
pointers. An equivalent specification without separation logic leads to an
unwieldy proliferation of assertions about inequality of locations and
pointer-graph reachability~\citep{reynolds2002separation}. 

What is an effective separation logic for probabilistic programs?
In the probabilistic setting, the fundamental source of modularity is
\emph{probabilistic independence}. Intuitively, two sources of randomness
are independent if knowledge of one does not give any knowledge of the other.
In direct analogy to disjointness of heaps,
independence structure permeates probabilistic programs:
sampling produces a random variable independent of all previously-sampled ones,
and commonly-used subroutines (e.g. randomly initializing an array) generate
multiple mutually-independent outputs.  Just as in the traditional setting,
attempting to write specifications for these procedures in ordinary predicate
logic leads to a proliferation of independence assertions.  A logic for
compositional probabilistic reasoning should compositionally support
independence, in the same way that ordinary separation logic compositionally
supports reasoning about disjoint heaps.

In this paper we present Lilac, a separation logic whose separating conjunction
means probabilistic independence.
Lilac enjoys a frame rule that is identical
to the frame rule of ordinary separation logic; as a consequence, the same
modular reasoning principles used on heap-manipulating programs apply directly
to the probabilistic setting.
Moreover, we prove that Lilac's separating conjunction completely captures
probabilistic independence: all probabilistic independence relationships are
validated by its semantic model (Lemma \ref{lem:star-is-independence}).
Both of these points are improvements over prior
work~\citep{barthe2019probabilistic,bao2022separation},
and are consequences
of our first core contribution: a new \emph{combining operation on probability
spaces}, analogous to \emph{disjoint union of heap fragments} in ordinary separation
logic, that serves as the interpretation of separating conjunction.

Our second core contribution is
a \emph{modal} treatment of conditional probability.
It is common in many probabilistic systems for a property
to only hold \emph{conditional} on some random variable: for instance, two random
variables might only be independent conditional on a third, a property called
\emph{conditional independence}.
Conditional reasoning is a second powerful and
prevalent source of modularity:
correctness arguments for probabilistic programs often hinge upon a clever
choice of what to condition on, exploiting key conditional independence
relationships to complete the proof.
Historically, conditional independence has been very difficult to capture in a substructural
program logic: it has either gone unsupported~\citep{barthe2019probabilistic,barthe2018assertion}
or required a host of new logical connectives and significant changes to the underlying semantic
model~\citep{bao2021bunched}.
Lilac captures conditioning via the addition of a single modal operator.
Adding support for this conditioning modality doesn't require any changes to
our underlying semantic model beyond restricting ourselves to a class of
suitably-well-behaved probability spaces.
By importing standard theorems of probability theory, we validate a set of derived rules
about the conditioning modality that we argue captures
the informal flavor of conditional reasoning.

In sum, our contributions are as follows:
\begin{itemize}[leftmargin=*]
\item We present Core Lilac, a separation logic 
  whose separating conjunction captures
  independence (Lemma~\ref{lem:star-is-independence}), alongside 
  proof rules for reasoning
  about a simple probabilistic programming language capable 
  of expressing all of the examples we will consider
  (Section~\ref{sec:core-lilac}).
  Core Lilac's semantic model is based on a novel combining operation
  on probability spaces that resembles disjoint union of heap fragments
  (Lemma~\ref{thm:spaces-form-a-krm}).
\item We extend Core Lilac with a modal operator $\D$
  to express conditional reasoning (Section~\ref{sec:modality}).
  This makes Lilac the first logic that supports
  conditioning and continuous random variables in combination with a substructural
  treatment of independence.
\item
  We validate the effectiveness of Lilac as a useful tool for program
  verification by proving correctness properties of existing examples from the
  literature as well as a new challenging example.  Establishing independence
  structure is key for proving certain cryptographic protocols correct.
  We give Lilac proofs for the one-time pad, private information retrieval, and oblivious
  transfer protocols studied by \citet{barthe2019probabilistic}
   (see \referto{sec:barthe-examples}).
  Lilac can also establish conditional independence properties:
  we show this by validating the conditional independence properties of 
  all programs considered by \citet{bao2021bunched} (Sections~\ref{sec:tour-of-lilac} and \ref{sec:condsamples}). 
Finally, we consider a challenging new example that goes beyond the scope of existing 
separation logics.
We validate an intricate reservoir
  sampling algorithm that uses continuous random variables and whose correctness
  argument depends crucially on conditional independence structure
  (Section~\ref{sec:weighted-sampling}).
\end{itemize}

\subsection{A Tour of Lilac} \label{sec:tour-of-lilac}

To concretize the discussion, we now present a few simple examples in order to
illustrate how Lilac's separating conjunction encodes independence and how the
conditioning modality can be used to establish simple conditional independence
relationships.

First, consider a program \ref{prog:unif2} that samples two reals $X,Y$
uniformly from the interval $[0,1]$:
\begin{equation}
  {\small
  \tag{\textsc{unif2}}
\begin{aligned}
  \appllet X\appleq\applunif \applin~~
  \appllet Y\appleq\applunif \applin~~
  \applret (X,Y)
\end{aligned}
\label{prog:unif2}
  }%
\end{equation}
We will write all probabilistic programs in monadic style, in a manner similar to Haskell's
do-notation; the keyword $\applret$lifts the value $(X,Y)$ into a
pure monadic computation.
This program satisfies two main properties of interest. First, the outputs $X$ and $Y$ are independent
and distributed as $\unif[0,1]$. Second, as freshly generated
random variables, both $X$ and $Y$ are independent of all other variables.
Both properties are asserted by the following Hoare triple:
\[
  \triple{ \top } {\ref{prog:unif2}} {(X,Y)}{X\sim\unif[0,1] ~~*~~ Y\sim\unif[0,1]}
\]
We write $\triple PMX{Q(X)}$ for a program $M$ that satisfies precondition $P$
and produces a random variable $X$ satisfying postcondition $Q(X)$.
In this case, the precondition is the trivial $\top$.
The proposition $X\sim\unif[0,1]$ is a Lilac assertion: it asserts that random variable $X$ is distributed as
$\unif[0,1]$.
This is in direct analogy to the proposition $\ell\mapsto v$ from ordinary separation logic.
Ordinarily,
\begin{wrapfigure}{r}{0.4\linewidth}
  \centering
  \quad
  \begin{minipage}{0.4\linewidth}
   {\footnotesize
  \begin{align*}
   &\assert{\{\top\}} \\
   &\hspace{1em}\appllet X\appleq\applunif \applin\\
   &\assert{\{X\sim \unif[0,1]\}} \\
   &\hspace{1em}\appllet Y\appleq\applunif \applin\\
   &\assert{\{X\sim \unif[0,1]~~*~~Y\sim\unif[0,1]\}} \\
   &\hspace{1em}\applret (X,Y) \\
   &\assert{\{(X,Y).~X\sim \unif[0,1]~~*~~Y\sim\unif[0,1]\}}
 \end{align*}%
   }
\end{minipage}
\caption{Lilac-annotated \ref{prog:unif2}.}
\label{fig:annot-1}
\end{wrapfigure}
the separating conjunction $(\ell_1\mapsto
v_1)*(\ell_2\mapsto v_2)$ asserts that $\ell_1$ and $\ell_2$ refer to disjoint
heap chunks containing values $v_1$ and $v_2$ respectively.
Analogously, the postcondition $(X\sim\unif[0,1])*(Y\sim\unif[0,1])$ asserts
that $X$ is independent of $Y$ (henceforth written $X \indep Y$) and that both
are distributed as $\unif[0,1]$.
Finally, the frame rule implies $X$ and $Y$ are independent of all other random variables,
expressing the fact that $X$ and $Y$ are freshly generated.

To establish this postcondition, Lilac provides proof rules that enable
the usual forward-symbolic-execution-style reasoning~\citep{reynolds2009introduction};
we will present these in Section~\ref{sec:derived}.
The rule for generating a random variable using $\applunif$ is:
\begin{align*}
  \triple\top{\applunif}X{X\sim\unif~[0,1]},
\end{align*}
in direct analogy to the rule for allocating a new reference in ordinary
separation logic. The rules for monadic operators are standard; see rules \textsc{H-Let}
and \textsc{H-Ret} in Figure~\ref{fig:proof-rules}. This allows derivations in the Lilac program
logic to be abbreviated
as assertion-annotated programs in the usual way. For
example, the postcondition for \ref{prog:unif2} can be established by the
annotation in Figure~\ref{fig:annot-1}.

Our next example illustrates how Lilac's notion of separation can be used
together with probability-specific reasoning principles to compute an expectation.
Consider a program \ref{prog:halve} that takes a random variable $X$ as input,
generates a uniform random variable $Y$, and computes their product:
\begin{equation}
  \tag{\textsc{halve}}
\begin{aligned}
  \ref{prog:halve}(X)
  \hspace{1em}:=\hspace{1em}
  \big(\appllet Y\appleq \applunif \applin~ \applret XY\big)
\end{aligned}
\label{prog:halve}
\end{equation}
Because $Y$ is freshly generated, $X$ and $Y$ must be independent;
combined with the fact that $Y$ is uniform, we have $\Ex[XY] = \Ex[X]\Ex[Y] =
\Ex[X]/2$.\footnote{%
$\Ex[X]$ denotes the expectation of the random variable $X$.}
(Hence the name: $\ref{prog:halve}(X)$'s output is half of $X$ in expectation.)
This argument makes use of two key facts:
expectation distributes over the product of independent random variables,
and for $U \sim \unif[0,1]$ we have $\Ex[U] = 1/2$.

In Lilac, the claim that $\ref{prog:halve}(X)$ has expectation
$\Ex[X]/2$ can be expressed by:
\begin{align}
  \{\own X\}~ \ref{prog:halve}(X) ~\{Z.~\Ex[Z]= \Ex[X]/2 \}
  \label{eq:spec-halve}
\end{align}
The postcondition states that the random variable $Z$ produced by \ref{prog:halve}
has expectation $\Ex[X]/2$.
The precondition $\own X$ is a new kind of Lilac assertion: it asserts ``probabilistic ownership'' of a random variable $X$
but nothing about its distribution. 
Probabilistic ownership is in direct analogy to the assertion
$\ell \mapsto -$ of ordinary separation logic, which asserts ownership of
a location $\ell$ in the heap but nothing about the value stored at $\ell$.
Probabilistic ownership allows us to describe independence relationships involving $X$
without knowledge of its distribution.
In particular, the proposition $\own X ~*~ \own Y$ asserts that $X\indep Y$,
just as the proposition $(\ell_1\mapsto -) ~~*~~ (\ell_2\mapsto -)$ asserts
that $\ell_1$ and $\ell_2$ are disjoint locations in ordinary separation logic.
In Section~\ref{sec:tour-of-model} we will describe precisely what ownership means in this probabilistic context.

\newcommand\halveSymbexecStart1
\newcommand\halveFramePost3
\newcommand\halveSymbexecEnd5
\newcommand\halveAsequal5
\newcommand\halveEx6
\newcommand\halveSubst7
\newcommand\halveCalc8
\begin{wrapfigure}{r}{0.5\linewidth}
  \centering
  \quad\\[-1em]
  \begin{minipage}{0.5\linewidth}
   {\footnotesize
\begin{align*}
  &1~\,\assert{\{\own X\}} \\
  &2~\,\hspace{1em}\appllet Y\appleq\applunif \applin\\
  &3~\,\assert{\{\own X~~*~~Y\sim\unif[0,1]\}} \\
  &4~\,\hspace{1em}\applret XY \\
  &5~\,\assert{\{Z.~\own X~~*~~Y\sim\unif[0,1]~~*~~Z\asequal XY \}} \\
  &6~\,\assert{\Big\{Z.~\Big(\Ex[Y] = 1/2 ~~\land~~ \Ex[XY] = \Ex[X]\Ex[Y] \Big)~~*~~Z\asequal XY \Big\}} \\
  &7~\,\assert{\{Z.~\Ex[Y] = 1/2 ~~\land~~ \Ex[Z] = \Ex[X]\Ex[Y]\}} \\
  &8~\,\assert{\{Z.~\Ex[Z]= \Ex[X]/2 \}}
\end{align*}%
   }
\end{minipage}
\caption{Lilac-annotated \ref{prog:halve}.}
\label{fig:annot-2}
\end{wrapfigure}
The annotation in Figure~\ref{fig:annot-2} proves that \ref{prog:halve} meets
the specification in Equation~\ref{eq:spec-halve}. The proof proceeds in two
phases. In the first phase, we apply standard proof rules for $\applunif$ and
$\applret$to obtain the assertion on Line~\halveSymbexecEnd.\footnote{%
Section~\ref{sec:derived} describes these rules in detail.} As is standard for
proofs in separation logic, we implicitly apply the frame rule when needed. In
this case, the frame rule guarantees that $\own X$ is preserved across the
invocation of $\applunif$, giving the separating conjunction on
Line~\halveFramePost.  This formalizes the intuition that $Y$ is freshly
generated, and therefore independent of $X$.  Line~\halveAsequal\ introduces the
proposition $Z\asequal XY$, which asserts that the output $Z$ is almost-surely
equal to the product $XY$. Two random variables $X$ and $Y$ are
almost-surely equal if they are equal with probability $1$ (i.e., $\Pr(X=Y) =
1$); this is the natural notion of equality for random variables.
See Figure~\ref{fig:lilac-prob-semantics} for the semantics of $(\asequal)$.
Note that using separating conjunction to combine this proposition with the
others does not introduce any spurious independence relationships; unlike $\own
X$ or $Y\sim\unif[0,1]$ the proposition $Z\asequal XY$ does not assert ownership
of any random variable.\footnote{Section~\ref{sec:core-lilac-semantics} will
make this precise.} This is in contrast to prior work such as
\citet{barthe2019probabilistic}, where almost-sure equality requires ownership
of the variables being compared and, as a result, proofs involving equalities
often require a cumbersome mixing of $(\land)$ and $(*)$.

Lines~\halveEx-\halveCalc{} finish the proof using proof
rules expressing nontrivial probability-specific reasoning.
  Line~\halveEx\ applies the proof rules:
  \begin{align}
    \own X ~~*~~ \own Y &\hspace{1em}\vdash\hspace{1em} \Ex[XY]=\Ex[X]\Ex[Y] \tag{\textsc{Indep-Prod}}\label{eqn:indep-prod}\\
    Y \sim \unif[0,1] &\hspace{1em}\vdash\hspace{1em} \Ex[Y] = 1/2  \tag{\textsc{Ex-Unif}}\label{eqn:ex-unif}
  \end{align}
  to obtain a conjunction of equalities about the expectations of $Y$ and $XY$.
  Unlike the proof rules we have considered so far, which are structural in nature and express
  standard logical reasoning principles, these rules express probability-specific facts:
  \ref{eqn:indep-prod} expresses the fact that
  expectation distributes over products of independent
  random variables, and \ref{eqn:ex-unif} expresses the expectation of
  the standard uniform distribution.
  As probability-specific proof rules, these do not follow from the ordinary laws of separation logic;
  instead, they are validated by the probability-specific model that we will
  describe in Section~\ref{sec:tour-of-model}.
  Line~\halveSubst\ uses the almost-sure equality $Z\asequal XY$ to replace all occurrences of
  $XY$ with $Z$.
  The rest of the proof follows by calculation.

\goodparagraph{Lilac's conditioning modality}
Conditioning is at the heart of probabilistic reasoning, and is central to most
probabilistic arguments. A core goal of Lilac is to facilitate conditional
probability arguments that would be familiar to probability theorists.
Informally, conditioning ``[turns a] random event or variable into a
deterministic one, while preserving the random nature of other events and
variables''~\citep{tao_2015}. We capture this intuition
with the \emph{conditioning modality} ``$\D_{x\gets X}P(x)$'', which asserts that
$P(x)$ holds \emph{conditional} on all possible values $x$
the random variable $X$ can take on. A key feature of this modality is that
standard separation logic assertions have intuitive conditional readings under it;
thus $\D$ lifts statements about unconditional probability to their natural conditional counterparts.
For instance, we have seen that $\own X* \own Y$
asserts that $X \indep Y$, and $\Ex[X] = v$ asserts that $X$ has
expectation $v$.  Accordingly, $\D_{z\gets Z}~~(\own X* \own Y)$ expresses
\emph{conditional independence} of $X$ and $Y$ given the random variable $Z$,
denoted $X \indep Y \mid Z$, and $\D_{y\gets Y} ~~(\Ex[X]
= f(y))$ asserts that $X$ has \emph{conditional expectation} $f(Y)$ given $Y$, i.e. $\Ex[X \mid Y] = f(Y)$.

Conditional arguments are driven by a careful distinction between random and
deterministic quantities.  We make these distinctions notationally and
semantically explicit in Lilac. We write random expressions as capital letters
$X$; these variables stand for random variables manipulated by a probabilistic
program, and are compared for equality using $(\asequal)$.  \emph{Deterministic}
(i.e., non-probabilistic) variables are written in lower-case letters~$x$, and
are compared for equality using ordinary $(=)$. 

Lilac supports familiar ways of transitioning between deterministic and 
probabilistic quantities, and typing rules (given in
Section~\ref{sec:core-lilac-syntax}) govern precisely where random and
deterministic expressions can appear. For instance, $\Ex[X] = v$ asserts 
that a random
expression $X$ relates to a deterministic quantity $v$. The conditioning
modality permits arguments
that mediate between the random and the deterministic.
When a probability-theorist says ``and now we proceed by conditioning $X$'' in a 
pen-and-paper proof, we translate this to entering the conditioning modality.
The proposition
$\D_{x\gets X}P(x)$ binds a new deterministic variable $x$ for use in $P(x)$.
Intuitively, the deterministic $x$ represents an arbitrary but non-probabilistic
value that $X$ has been fixed to inside $P$ via conditioning.  This allows
replacing $X$ with $x$ inside $P$. For example, the proposition
$
  \D_{x\gets X} (\Ex[X] = x)
$
is valid: under $\D_{x\gets X}$, we can replace $X$ with $x$ to get $\Ex[X] = \Ex[x]$,
and $\Ex[x] = x$ because $x$ is deterministic.

This ability to turn random quantities $X$ into deterministic variables $x$
is especially useful because deterministic variables are, generally speaking,  better behaved than
their random counterparts. 
In particular, deterministic variables support
case analysis: if $b$ is a deterministic variable
of type $\applkw{bool}$ then
  $P[\vtrue/b] \land P[\vfalse/b] \vdash P$.
This allows Lilac to express
a kind of conditional argument by case analysis that is pervasive in probabilistic reasoning.
For example, consider the following program:
\begin{equation}
  {\small
  \tag{\textsc{CommonCause}}
\begin{aligned}
  &\appllet Z \appleq \flip~1/2 \applin \quad
  \appllet X \appleq \flip~1/2 \applin  \quad
  \appllet Y \appleq \flip~1/2 \applin \\
  &\appllet A \appleq \applret(X~\applkw{||}~Z) \applin \quad
  \appllet B \appleq \applret(Y~\applkw{||}~Z) \applin \\
  &\applret (Z, X, Y, A, B)
\end{aligned}
\label{prog:commoncause}
  }
\end{equation}
This program is taken from
Figure~6(a) of \citet{bao2021bunched} (and translated into a monadic style),
where it is used as an example of conditional independence structure.
The program computes $A$ and $B$ from mutually-independent boolean random variables $Z,X,Y$
generated by $\flip~1/2$; at exit, we have that $A \indep B \mid Z$.
In \citet{bao2021bunched}, this is established via a logic of
``doubly-bunched'' implications -- an extension of separation logic with an additional
family of substructural connectives for expressing conditional independence.
We will show how Lilac can state and prove this conditional independence without having
to introduce new connectives beyond $\D$.

The conditional independence property $A \indep B \mid Z$ is captured by the following triple:
\[
  {\small
  \{\top\}~~\ref{prog:commoncause}~~\Big\{ (Z,X,Y,A,B).~\D_{z\gets Z} (\own A * \own B) \Big\}
  }
\]
This postcondition can be established as follows.\footnote{
Here we prefer to describe the proof as a mixture of prose and Lilac assertions.
\referto{sec:annotated-commoncause} gives a fully annotated program.}
First, applying rules for $\flip$ and $\applkw{||}$ gives:
\[
  {\small
    Z \sim \ber 1/2 \hspace{0.5em}*\hspace{0.5em}  X\sim\ber 1/2 \hspace{0.5em}*\hspace{0.5em} Y\sim\ber 1/2
              \hspace{0.5em}*\hspace{0.5em} A\asequal (X\lor Z) \hspace{0.5em}*\hspace{0.5em} B\asequal (Y\lor Z)
  }
\]
At this point, an informal proof would continue by case analysis on $Z$ as follows.
Because $Z$ is a Boolean value, it can only take on one of two values: $\vtrue$ or $\vfalse$.
First establish $A \indep B \mid (Z=\vtrue)$.
  If $Z=\vtrue$ then $A=B=\vtrue$, and therefore $A \indep B$ because
  by definition $\vtrue \indep \vtrue$.
Now establish $A \indep B \mid (Z=\vfalse)$.
  If $Z=\vfalse$ then $A = X$ and $B = Y$. By construction we have $X \indep Y \indep Z$, 
  since they are all independent $\flip$s.
  This mutual independence implies that $X$ and $Y$ remain independent after
conditioning on $Z=\vfalse$. Hence $A \indep B \mid (Z=\vfalse)$.
This completes the case analysis, so $A\indep B\mid Z$ as desired.

In Lilac, this conditional argument is expressed by introducing the operator $\D_{z\gets Z}$
and performing case analysis on the deterministic $z$.
First, $\D_{z\gets Z}$ is introduced as follows:
\[
  {\small
    \D_{z\gets Z} \hspace{0.7em} \Big(X\sim\ber 1/2 \hspace{0.5em}*\hspace{0.5em} Y\sim\ber 1/2
              \hspace{0.5em}*\hspace{0.5em} A\asequal (X\lor Z) \hspace{0.5em}*\hspace{0.5em} B\asequal (Y\lor Z)\Big)
  }
\]
This step formalizes the idea that anything independent of $Z$ continues to be
independent after conditioning. It is justified by an application of the
\DIndep~ rule, which we will describe in Section~\ref{sec:modality}.

Conditioning on $Z$ allows us to replace 
occurrences of the random variable $Z$ with the newly-introduced deterministic $z$:
\begin{align}
  {\small
    \D_{z\gets Z} \hspace{0.7em} \Big(X\sim\ber 1/2 \hspace{0.5em}*\hspace{0.5em} Y\sim\ber 1/2
              \hspace{0.5em}*\hspace{0.5em} A\asequal (X\lor z) \hspace{0.5em}*\hspace{0.5em} B\asequal (Y\lor z)\Big)
  }
  \label{eq:intermed-entail}
\end{align}
We are done if we can show that (\ref{eq:intermed-entail}) entails $\D_{z\gets
Z}(\own A * \own B)$.  At this point we use a key property of $\D$: as a modal
operator, it respects entailment, so we must establish the
$\D$-free entailment:\footnote{ That is, if $P\vdash Q$ then
$\D_{x\gets X}P\vdash \D_{x\gets X}Q$.}
\[ {\small{X\sim\ber 1/2 \hspace{0.3em}*\hspace{0.3em} Y\sim\ber 1/2
              \hspace{0.3em}*\hspace{0.3em} A\asequal (X\lor z) \hspace{0.3em}*\hspace{0.3em} B\asequal (Y\lor z)
              \hspace{0.5em}\vdash\hspace{0.5em} \own A * \own B}.}
\]
Now the rest of the proof follows directly by cases on $z$:
if $z=\vtrue$ then after simplifying we have:
\[ {\small{A\asequal \vtrue \hspace{0.3em}*\hspace{0.3em} B\asequal \vtrue
              \hspace{0.5em}\vdash\hspace{0.5em} \own A * \own B},}
\]
which follows from the rule $X\asequal \vtrue \vdash \own X$ specialized to $X=A$ and $X=B$.
On the other hand, if $z=\vfalse$ then we are left (again, after simplifications) with
\[ {\small{A \sim \ber 1/2 \hspace{0.3em}*\hspace{0.3em} B\sim\ber 1/2
              \hspace{0.5em}\vdash\hspace{0.5em} \own A * \own B}},
\]
which follows from the rule $X\sim\ber 1/2 \vdash \own X$.

\subsection{A Tour of Lilac's Semantic Model}
\label{sec:tour-of-model}

So far we have sketched a system of proof rules for reasoning about probability that would
appear quite intuitive to a probability theorist. However, reasonable-looking
proof rules are only half of the story. To ensure that the rules are sound and
that Lilac propositions have sensible interpretations in terms of existing
probability-theoretic objects, we construct a model that validates the reasoning
principles described in the previous section and grounds them in probability-theory.

Lilac's semantic model is
designed by analogy with mutable state.
The key idea is that \emph{probability spaces} ---
mathematical objects that model random phenomena --- behave like heaps.
To make this concrete, consider the following program:
\begin{equation}
  {\small
  \tag{\textsc{flip2}}
\begin{aligned}
  \appllet X\appleq \flip~1/2 \applin ~~
  \appllet Y\appleq \flip~1/2 \applin ~~
  \applret (X,Y)
\end{aligned}
\label{prog:flip2}
  }
\end{equation}
According to the informal semantics for probabilistic programs used in Section~\ref{sec:tour-of-lilac},
this program ``generates'' two uniformly distributed boolean random variables
$X$ and $Y$.
To make this precise, we need some standard definitions.
A \emph{probability space} is a tuple $(\Omega,\calF,\mu)$.
The set $\Omega$ is called the \emph{sample space}.
The set $\calF$, called the \emph{$\sigma$-algebra}, is a collection of subsets of $\Omega$
that (1) contains the empty set and $\Omega$, and (2)
satisfies closure under countable union and complements.
Elements $E\in\mathcal F$ are \emph{events};
events models an observable property of the phenomenon.
The map $\mu : \calF \rightarrow [0,1]$ is a \emph{probability measure} assigning each 
event its probability.
Given a finite set $A$,
an \emph{$A$-valued random variable} is a map $X:\Omega\to A$ that is \emph{$\mathcal F$-measurable},
which means that the set $\{\omega \mid X(\omega) = a\}$ is an event of $\mathcal F$
for all $a\in A$.
Intuitively, a random variable $X$ represents an object of type $A$ that depends
on the random phenomenon modelled by $\Omega$; the measurability condition ensures that
$X$ only depends on observable properties.

To visualize these objects, we will temporarily fix $\Omega$ to be the unit square $[0,1]\times[0,1]$
and $\mu$ to be the function that assigns to each subset of $\Omega$ its area, if possible.
This allows us to draw a probability space $(\Omega,\mathcal F,\mu)$ as a partitioning of the unit square
given by the $\sigma$-algebra $\mathcal F$. For example,
the square $\fullsqlrbig$
depicts a probability space with
$\sigma$-algebra $\mathcal F_{\fullsqlr} :=\{\emptyset,\leftrect,\rightrect,\Omega\}$
and probability measure:
\[
  \mu(\emptyset) = 0 \qquad \mu(\leftrect) = 1/2 \qquad \mu(\rightrect) = 1/2 \qquad \mu(\Omega) = 1.
\]
This probability space models a random phenomenon with two events (the blue rectangle
and the orange rectangle), each with equal probability.

For example, one can define a random variable on the space $\fullsqlr$
as follows.
Let $A$ be the set of boolean values $\{\vtrue,\vfalse\}$
and $X : [0,1]\times[0,1]\to\{\vtrue,\vfalse\}$
be a function that sends blue points to $\vtrue$
and orange points to $\vfalse$.
Concretely, $X(\omega_1, \omega_2) = \vtrue$ if $\omega_1 < 1/2$ and $\vfalse$ otherwise.
This is a random variable: the measurability condition is satisfied because
$X^{-1}(\vtrue)=\leftrect\in \mathcal F_{\fullsqlr}$ and $X^{-1}(\vfalse)=\rightrect\in\mathcal F_{\fullsqlr}$.

We now have the machinery necessary to interpret the \ref{prog:flip2} example.
The idea is to think of \ref{prog:flip2} as carrying along a probability space
as it executes, analogous to how programs with mutable state carry along a heap.
We can visualize execution of \ref{prog:flip2} as follows:
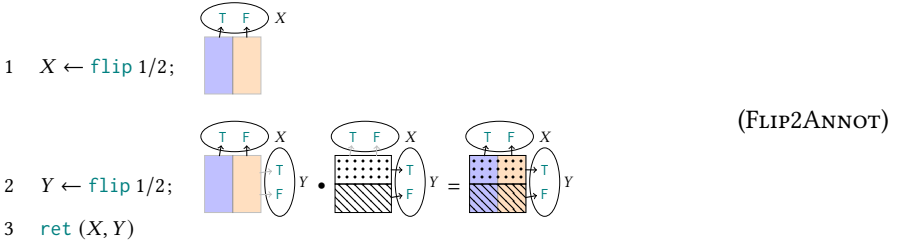
\begin{align}
  {\footnotesize
\begin{tabular}{lll}
  1&$\appllet X \appleq \flip~1/2 \applin$
  & \scalebox{0.5}{${\tikz[baseline=4.5ex]{
    \draw (0.8, 2) ellipse (0.9cm and 0.4cm) node[xshift=9] (xT) {$\mathextralarger\vfalse$} node[xshift=-9] (xF) {$\mathextralarger\vtrue$} node[xshift=35] {$\mathextralarger X$};
    \draw[fill=blue,opacity=0.25] (0,0) rectangle (0.75,1.5) node[pos=0.5,yshift=18] (l) {};
    \draw[fill=orange,opacity=0.25] (0.75,0) rectangle (1.5,1.5) node[pos=0.5,yshift=18] (r) {};
    \draw[->, thick] (l) --  (xF);
    \draw[->, thick] (r) --  (xT);
    }}$ }
    \\
  \\
  2&$\appllet Y \appleq \flip~1/2 \applin$
  &
  \scalebox{0.5}{${\tikz[baseline=4.5ex]{
    \draw (0.8, 2) ellipse (0.9cm and 0.4cm) node[xshift=9] (xT) {$\mathextralarger\vfalse$} node[xshift=-9] (xF) {$\mathextralarger\vtrue$} node[xshift=35] {${\mathextralarger X}$};
    \draw (2.0, 0.8) ellipse (0.4cm and 0.9cm) node[yshift=9] (yT) {$\mathextralarger\vtrue$} node[yshift=-9] (yF) {$\mathextralarger\vfalse$} node[xshift=18] {$\mathextralarger Y$};
    \draw[fill=blue,opacity=0.25] (0,0) rectangle (0.75,1.5) node[pos=0.5,yshift=18] (l) {};
    \draw[fill=orange,opacity=0.25] (0.75,0) rectangle (1.5,1.5) node[pos=0.5,yshift=18] (r) {};
    \node[] (0,0) {} rectangle (1.5,0.75) node[xshift=-4, yshift=-8] (u) {} {};
    \node[] (0,1.5) {} rectangle (1.5,0.75) node[xshift=-4, yshift=8] (d) {} {};
    \draw[->, thick] (l) --  (xF);
    \draw[->, thick] (r) --  (xT);
    \draw[->, thick, lightgray] (u) --  (yF);
    \draw[->, thick, lightgray] (d) --  (yT);
    }}$}
    $\pdot$
  \scalebox{0.5}{${\tikz[baseline=4.5ex]{
    \draw (0.8, 2) ellipse (0.9cm and 0.4cm) node[xshift=9] (xT) {$\mathextralarger \vfalse$} node[xshift=-9] (xF) {$\mathextralarger \vtrue$} node[xshift=35] {${\mathextralarger X}$};
    \draw (2.0, 0.8) ellipse (0.4cm and 0.9cm) node[yshift=9] (yT) {$\mathextralarger \vtrue$} node[yshift=-9] (yF) {$\mathextralarger \vfalse$} node[xshift=18] {${\mathextralarger Y}$};
    \node[] (0,0) {} rectangle (0.75,1.5) node[pos=0.5,yshift=18] (l) {};
    \node[] (0.75,0) {} rectangle (1.5,1.5) node[pos=0.5,yshift=18, xshift=10] (r) {};
    \draw[pattern=north west lines] (0,0) rectangle (1.5,0.75) node[pos=0.5,xshift=18] (u) {};
    \draw[pattern=dots] (0,1.5) rectangle (1.5,0.75) node[pos=0.5,xshift=18] (d) {};
    \draw[->, thick, lightgray] (l) --  (xF);
    \draw[->, thick, lightgray] (r) --  (xT);
    \draw[->, thick] (u) --  (yF);
    \draw[->, thick] (d) --  (yT);
    }}$}
  ~$=$~
  \scalebox{0.5}{${
  \tikz[baseline=4.5ex]{
    \draw (0.8, 2) ellipse (0.9cm and 0.4cm) node[xshift=9] (xT) {$\mathextralarger\vfalse$} node[xshift=-9] (xF) {$\mathextralarger\vtrue$} node[xshift=35] {${\mathextralarger X}$};
    \draw (2.0, 0.8) ellipse (0.4cm and 0.9cm) node[yshift=9] (yT) {$\mathextralarger\vtrue$} node[yshift=-9] (yF) {$\mathextralarger\vfalse$} node[xshift=18] {${\mathextralarger Y}$};
    \draw[fill=blue,opacity=0.25] (0,0) rectangle (0.75,1.5) node[pos=0.5,yshift=18] (l) {};
    \draw[fill=orange,opacity=0.25] (0.75,0) rectangle (1.5,1.5) node[pos=0.5,yshift=18] (r) {};
    \draw[pattern=north west lines] (0,0) rectangle (1.5,0.75) node[pos=0.5,xshift=18] (u) {};
    \draw[pattern=dots] (0,1.5) rectangle (1.5,0.75) node[pos=0.5,xshift=18] (d) {};
    \draw[->, thick] (l) --  (xF);
    \draw[->, thick] (r) --  (xT);
    \draw[->, thick] (u) --  (yF);
    \draw[->, thick] (d) --  (yT);
    }}$}
   \\
  3&$\applret (X,Y)$ & ~
\end{tabular}}
\tag{\textsc{Flip2Annot}}
\label{eq:lilac-semantics}
\end{align}
Before Line 1, the probability space has exactly two events,
$\emptyset$ and $\Omega$, with respective probabilities $0$ and $1$.
This is the trivial probability space, analogous to an empty heap.
After the first line is executed two things change:
\begin{itemize}[leftmargin=*]
  \item Two new events are allocated, forming the blue-orange probability space $\fullsqlr$
     we considered earlier.
     This is analogous to how $\new$ allocates a fresh memory cell on the heap: probability spaces correspond to heap fragments. 
  \item The $\flip$ operation yields the random variable $X$ we considered earlier that maps blue points to $\vtrue$;
    it is depicted by the arrows.
    This is analogous to how $\new$ returns the
    location of the newly allocated heap cell:
    random variables correspond to locations.
\end{itemize}
Line 2 allocates a second probability space $\fullsqtb$, whose events are the dotted
region $\toprect$ and dashed region $\botrect$, and a new random variable $Y$ that associates dotted
points to $\vtrue$ and dashed points to $\vfalse$; concretely, $Y(\omega_1, \omega_2) = \vtrue$ if 
$\omega_2 > 1/2$ and $\vfalse$ otherwise.
There are many other possible alternatives to $\fullsqtb$
and $Y$; we have simply chosen the ones that are easiest to visualize.
This is analogous to how, in heap-manipulating languages, there are many
possible locations in the heap that $\new$ can choose to allocate in.
Ordinarily, the location chosen by $\new$ is also \emph{fresh}, so that the entire heap
after executing $\new$ is a \emph{disjoint union} of the old heap and the newly
allocated cell. Analogously, $\flip$ allocates a probability space
\emph{probabilistically independent} from the old one, so that the entire space
$\fullsqlrtb$
after executing the second $\flip$ is an \emph{independent combination} of the
old space $\fullsqlr$ and the newly allocated one $\fullsqtb$.
The operator $(\pdot)$ is 
our new combining operation on probability
spaces: it has the same algebraic properties as disjoint union of heap fragments
does in ordinary separation logic (Theorem \ref{thm:spaces-form-a-krm}), and is the heart of
Lilac's notion of separation.
The events of the combined space $\fullsqlrtb$ are generated by the events
of $\fullsqlr$ and $\fullsqtb$;
we give a formal definition in Section \ref{sec:core-lilac}.
The insight that disjoint union of heaps corresponds to independent combinations of 
probability spaces underlies Lilac's interpretation of standard separation logic connectives:
in particular, a probability space $\calP$ satisfies assertion $P_1*P_2$
if there exists a ``splitting'' of $\calP$ into 
an independent combination $\calP_1\pdot\calP_2$ such that $\calP_1$ satisfies $P_1$
and $\calP_2$ satisfies $P_2$.

\ref{eq:lilac-semantics} also illustrates the second key insight that forms the
basis for our model of separation logic: ownership is measurability.  In
Section~\ref{sec:tour-of-lilac} we described the proposition $\own X$ as
asserting ``probabilistic ownership'' of $X$ but no knowledge of its
distribution.  Now we make this
precise: the proposition  $\own X$ holds in probability space $(\Omega, \calF,
\mu)$ when $X$ is $\mathcal F$-measurable.  In
\ref{eq:lilac-semantics}, $\own X$ holds in the space $\fullsqlr$ because $X$ is
$\calF_{\fullsqlr}$-measurable.  On the other hand, $\own Y$ does \emph{not}
hold in $\fullsqlr$, because the event $Y^{-1}(\vfalse) = \botrect$ is not
contained in $\calF_{\fullsqlr}$.  Similarly, $\own X$ does not hold in
$\fullsqtb$ because $X^{-1}(\vtrue) = \leftrect\notin\calF_{\fullsqtb}$.  The
grayed-out arrows depict non-measurability in
\ref{eq:lilac-semantics}.

\section{Core Lilac} \label{sec:core-lilac}
Now that we have given informal descriptions of Lilac's proof rules and semantic model,
we now begin the formal development of Core Lilac, a subset of Lilac without
the conditioning modality.
First we will introduce the syntax and
semantics of Core Lilac propositions. Then we present our
combining operation $(\pdot)$ on probability spaces
described in Section~\ref{sec:tour-of-model}, and show that it behaves like
disjoint union of heaps. Next, we give the semantics of Lilac propositions,
and fix a small PPL capable of expressing the examples presented in the
previous section, called ``APPL''.
Finally, we connect  Lilac to APPL by giving proof 
rules for reasoning about APPL programs.
\subsection{Syntax and Typing of Core Lilac Propositions} 
\label{sec:core-lilac-syntax}
  \begin{figure}
    \centering
 {\footnotesize{
\begin{align*}
  P,Q ::=~ &\top \mid \bot \mid P \land Q \mid P \lor Q \mid P\to Q \mid 
           P * Q \mid P\wand Q \mid \nec P \mid \\
           &\forall x\ofty S.P \mid \exists x\ofty S. P \mid
            \forallrv X\ofty A.P \mid \existsrv X\ofty A.P \mid 
           E \sim \mu \mid \own E \mid \Ex[E] = e \mid \weakpre(M, X\ofty A. Q)
\end{align*}
}}
\caption{Core Lilac syntax. Metavariables $S,T$ range over sets and $A,B$ over
  measurable spaces. Metavariables $E, e, M$, and $\mu$ range over 
  arbitrary measurable maps of a certain type described in Figure~\ref{fig:lilac-typing}.} 
  \label{fig:lilac-syntax}
  \end{figure}

  \begin{figure}
    \centering
    {\footnotesize
 \begin{align*}
  S&\in\mathbf{Set} & \Gamma &::= \cdot \mid \Gamma,x:S & \sembr{x_1:S_1,\dots,x_n:S_n} &= S_1\times\cdots\times S_n\\
  A&\in\mathbf{Meas} & \Delta &::= \cdot \mid \Delta,X:A &\sembr{X_1:A_1,\dots,X_n:A_n} &= A_1\times\cdots\times A_n 
 \end{align*}
\begin{mathpar}
  \inferrule*[right=T-RandE]{E \in \sembr\Gamma\to\sembr\Delta\mto A}
            {\Gamma;\Delta\vdashrv E:A}
  \and
  \inferrule*[right=T-DetE]{e \in \sembr\Gamma\to A}
            {\Gamma\vdashdet e:A}
  \and
  \inferrule*[right=T-Distr]{\mu \in \sembr\Gamma\to \giry A}
            {\Gamma\vdashdet \mu:A}
  \and
  \inferrule*[right=T-Prog]{M\in \sembr\Gamma\to \sembr\Delta\mto \giry A}{\Gamma;\Delta\vdashprog M:A}

  \inferrule*[right=T-Det$\forall$]{\Gamma,x\ofty S;\Delta\vdash P}{\Gamma;\Delta\vdash \forall x\ofty S. P}
  \and
  \inferrule*[right=T-Rand$\forall$]{\Gamma;\Delta,X\ofty A\vdash P}{\Gamma;\Delta\vdash \forallrv X\ofty A. P}
  \\
  \inferrule*[right=T-Sim]{\Gamma;\Delta\vdashrv E:A \\ \Gamma\vdashdet \mu: A}
            {\Gamma;\Delta\vdash E\sim\mu}
  \and
  \inferrule*[right=T-$\Ex$]{\Gamma;\Delta\vdashrv E:\R \quad \Gamma\vdashdet e:\R}
            {\Gamma;\Delta\vdash \Ex[E] = e}
  \and
  \inferrule*[right=T-$\asequal$]{\Gamma;\Delta\vdashrv E_1:A \quad \Gamma;\Delta\vdashrv E_2:A}
            {\Gamma;\Delta\vdash E_1 \asequal E_2}
  \and
  \inferrule*[right=T-$\weakpre$]{\Gamma;\Delta\vdashprog M:A \\
             \Gamma;\Delta,X\ofty A\vdash Q}
            {\Gamma;\Delta\vdash \weakpre(M,X\ofty A.Q)}
 
\end{mathpar}
    }
\caption{Core Lilac typing rules. Contexts $\Gamma$ contain the types of deterministic
variables and contexts $\Delta$ contain the types of random variables.
Metavariables $E$ range over random expressions and $e$ over deterministic expressions.
We write $\mathcal GA$ for the set of distributions on $A$.
  } \label{fig:lilac-typing}
\end{figure}

The syntax of Core Lilac propositions is given in Figure~\ref{fig:lilac-syntax}.
It includes the standard intuitionistic and substructural connectives,
plus the probability-specific ones $\own E$, $E\sim\mu$, $\Ex[E]=e$, and $E_1\asequal E_2$
introduced in Section~\ref{sec:tour-of-lilac}.
Figure~\ref{fig:lilac-typing} gives selected typing rules for Core
Lilac.\footnote{The full typing rules are in \referto{app:lilac-typing}.} 
The typing judgment has shape $\Gamma; \Delta \vdash P$, where $\Gamma$ is 
a context containing the types of deterministic variables and $\Delta$ is a context 
containing the types of random variables.
It is defined in terms of auxiliary judgments
$\Gamma;\Delta\vdashrv E : A$ for typing random expressions $E$, which may mention both
deterministic and random variables,
and $\Gamma\vdashdet e : A$ for typing deterministic expressions $e$,
which can only mention deterministic variables,
and $\Gamma;\Delta\vdashprog M : A$ for typing programs.
Since there are two kinds of variables,
there are also two kinds of quantifiers, with typing rules
\textsc{T-Det$\forall$} and \textsc{T-Rand$\forall$}.

For clarity of presentation the Core Lilac syntax in
Figure~\ref{fig:lilac-syntax} permits reference to arbitrary measurable 
spaces in its types and arbitrary measurable functions in its terms,
in a manner similar to \citet{shan2017exact} and
\citet{staton2020probabilistic}. 
The rule \textsc{T-RandE} characterizes 
the kinds of functions allowed as random expressions:
a random expression $E$ has type $A$ in context $(\Gamma; \Delta)$
if it is a $\sembr{\Gamma}$-indexed family of measurable maps
$\sembr{\Delta}\mto A$. 
For instance, the following random expressions are all well-typed 
via \textsc{T-RandE} because $(+)$ and $\pow--$ are both
measurable functions $\R\times\R\to\R$:

{\small
\vspace{-0.5em}
\begin{mathpar}
\inferrule*[]{\Gamma;\Delta \vdashrv E_1 : \R \\
\Gamma;\Delta \vdashrv E_2 : \R
}{\Gamma; \Delta \vdashrv E_1 + E_2 : \R}
\and
\inferrule*[]{\Gamma;\Delta \vdashrv E_1 : \R \\
\Gamma;\Delta \vdashrv E_2 : \R
}{\Gamma; \Delta \vdashrv \pow{E_1}{E_2} : \R}
\end{mathpar}}%
Similarly, \textsc{T-DetE} says a deterministic expression $e$ is well-typed
at $A$ in $\Gamma$ if it is a function $\sembr{\Gamma}\to A$, and
\textsc{T-Distr} says $\mu$ is well-typed at $A$ in context $\Gamma$ if it 
is a $\sembr{\Gamma}$-indexed family of distributions;
we write $\calG A$ for the set of distributions on $A$.
Finally, \textsc{T-Prog} says a program $M$ is well-typed in
$\Gamma;\Delta$ if it is a $\sembr{\Gamma}$-indexed family of \emph{Markov
kernels} --- maps $A \mto \giry B$ often used to give semantics to
probabilistic programs~\citep{staton2020probabilistic}.

To reason about the behavior of programs we add a weakest precondition modality
in the style of dynamic logic~\cite{harel2001dynamic,jung2018iris}.  Intuitively,
$\weakpre(M, X \ofty A.Q)$ asserts that $M$ produces a random variable $X$
satisfying postcondition $Q$. For example, the fact that 
\ref{prog:flip2} from Section~\ref{sec:tour-of-model} produces two independent $\ber~1/2$ is stated
  $\weakpre(\sembr{\ref{prog:flip2}}, (X,Y).~ X \sim
\ber 1/2 ~~*~~ Y \sim \ber 1/2)$.

\subsection{Combining Independent Probability Spaces}
Before we present the semantic interpretation of Core Lilac, we first formally
describe our novel combining operation on probability spaces.
As described in
Section~\ref{sec:tour-of-model},
this combining operation behaves like disjoint union of heap fragments,
and underlies Lilac's model of separation. Formally, this is captured by the
notion of a Kripke resource monoid~\cite{galmiche2005semantics}:
\begin{definition}
A \emph{Kripke resource monoid}
is a tuple $(\calM,\plte,\pdot,\bfone)$ where \begin{enumerate}
\item $(\calM,\plte)$ is a poset,
\item $(\pdot)$ is a partial function $\calM \times \calM \pto \calM$,
\item $(\calM,\pdot,\bfone)$ is a partial commutative monoid,
\item $(\pdot)$ respects $(\plte)$: if $x\plte x'$ and $y\plte y'$
and $x'\pdot y'$ defined, then $x\pdot y$ defined and $x\pdot y\plte x' \pdot y'$.\footnote{This is a specialization of Definition 5.5 from
\citet{galmiche2005semantics}.}
\end{enumerate}
\label{def:krm}
\end{definition}
Intuitively, a KRM models the notion of a resource. The set $\mathcal M$
models the space of possible resources; the ordering $(\sqsubseteq)$ models how a given resource
can evolve over time. The operation $(\pdot)$ models how resources can be combined;
it is partial because not all resources are compatible with each other (e.g., 
overlapping heap fragments).
The choice of $(\pdot)$, which captures the desired notion of separation, determines the
interpretations of the standard separation logic connectives.
Our choice, as foreshadowed in Section~\ref{sec:tour-of-model}, combines 
two independent probability spaces:
\begin{definition}[Independent Combination] \label{def:independent-combination}
  Let $(\Omega, \calE,\mu)$ and $(\Omega, \calF,\nu)$ be probability spaces over the 
  same ambient sample space $\Omega$. A probability space $(\Omega,\calG,\rho)$ is an
  \emph{independent combination} of $(\Omega,\calE,\mu)$ and $(\Omega, \calF,\nu)$
  if (1) $\calG$ is the smallest $\sigma$-algebra containing $\calE$ and $\calF$, and (2)
   $\rho$ witnesses the independence of $\mu$ and $\nu$ in the sense that
  for all $E\in\calE$ and $F\in\calF$ it holds that
  $\rho(E\cap F)=\mu(E)\nu(F)$ .
\end{definition}

Recall the program \ref{prog:flip2} from Section~\ref{sec:tour-of-model}.
In this example the probability space $\fullsqlrtb$ is obtained as an independent
combination of $\fullsqlr$
and $\fullsqtb$.
To show this, the areas of the regions in $\fullsqlrtb$
must be products of intersections of regions in $\fullsqlr$ and $\fullsqtb$.
Consider the events $X^{-1}(\vtrue) = \leftrect$ and $Y^{-1}(\vtrue) = \toprect$.
Both of these events have area $1/2$, and
their intersection $\topleftsq$ -- the upper-left quandrant of the
unit square -- has area $1/4=(1/2)(1/2)$
as desired; clearly this holds for all quadrants.

To form a resource monoid, the combining operation $(\pdot)$ must be a partial function.
Definition~\ref{def:independent-combination}
relates probability spaces to their independent combinations.
However, it requires a witness $\rho$ of independence; if there
are multiple possible choices for $\rho$, then this relation does not define a partial function.
Thankfully -- and somewhat surprisingly%
-- it is possible to establish the uniqueness of $\rho$ and therefore of independent combinations:
 \begin{lemma}[independent combinations are unique] \label{lem:independent-combinations-are-unique}
  Suppose $(\Omega, \calG,\rho)$ and $(\Omega, \calG',\rho')$ are independent combinations of
  $(\Omega, \calE,\mu)$ and $(\Omega, \calF,\nu)$.  Then $\calG = \calG'$ and $\rho=\rho'$.
\end{lemma}
The proof is concise and relies on an application of the well-known \emph{Dynkin
$\pi$-$\lambda$ theorem}~\citep{kallenberg1997foundations}; see
\referto{app:indep-combination} for details.
Given Lemma~\ref{lem:independent-combinations-are-unique}, we can safely write
$(\calE,\mu)\pdot(\calF,\nu)=(\calG,\rho)$ whenever $(\calG,\rho)$
is an independent combination of $(\calE,\mu)$ and $({\calF},\nu)$, 
making $(\pdot)$ a partial function on probability spaces.
In addition to being a partial function, Definition~\ref{def:krm} also
requires that $(\pdot)$ forms a partial commutative monoid and respects a certain ordering 
relation. This indeed holds of our model: we take the ordering relation to be
inclusion of probability spaces, analogous to inclusion of heaps used in ordinary (affine) separation logic:
\begin{thm} \label{thm:spaces-form-a-krm}
  Let $\calM$ be the set of probability spaces over a fixed sample space $\Omega$.
  Let $(\pdot)$ be the partial function mapping two probability spaces to their independent
  combination if it exists.
  Let $(\plte)$ be the ordering such that $(\calF,\mu)\plte(\calG,\nu)$ iff
  $\calF\subseteq\calG$ and $\mu=\nu|_\calF$.\footnotemark
  The tuple $(\calM,\plte,\pdot,\mathbf{1})$ is a Kripke resource monoid,
  where $\mathbf{1}$ is the trivial probability space $(\calF_\pmbone,\mu_\pmbone)$
  with $\calF_\pmbone=\{\emptyset,\Omega\}$ and $\mu_\pmbone(\Omega) = 1$.
\footnotetext{
  For a distribution $\mu : \calG \rightarrow [0,1]$ and sub-$\sigma$-algebra
  $\calF \subseteq \calG$, we write $\mu|_\calF$ for the restriction of
  $\mu$ to $\calF$.
}
\end{thm}
\begin{proof} The main proof obligation is to establish associativity of $(\pdot)$.
  This follows from an application of the $\pi$-$\lambda$ theorem.
  The proof is intricate;
  for details, see \referto{app:spaces-form-a-krm}.
\end{proof}

There is a curious contrast between $(\bullet)$ and the standard definition of independence of $\sigma$-algebras in 
probability theory. The
standard notion of independence of two sub-$\sigma$-algebras $\calE, \calF
\subseteq \calG$ with respect to an ambient probability space $(\Omega, \calG,
\rho)$ states that $\rho$ factorizes along $\calE$ and $\calF$: i.e., it says that for all $E \in
\calE, F \in \calF$, it holds that $\rho(E \cap F) = \rho(E)\rho(F)$.  This
definition presupposes the existence of an ambient measure $\rho$ by which the
independence of $\calE$ and $\calF$ can be judged. In contrast, our independent
combination does not require $\rho$.
Instead, Lemma~\ref{lem:independent-combinations-are-unique} guarantees
that if any such $\rho$ exists, it is unique.
This observation turns the standard definition into
a partial function on probability spaces with the structure of a partial commutative monoid;
once one has combined $\mu$ and $\nu$ to obtain $\rho$ in this way, our definition 
coincides with the standard one.

\subsection{Semantics of Lilac Propositions} \label{sec:core-lilac-semantics}
Having established that independent combination of probability spaces
forms a Kripke resource monoid, we are now ready to present Lilac's semantic model
with it as the foundation.
Figure~\ref{fig:lilac-basic-semantics} gives 
Lilac's interpretations of standard separation logic connectives,
along with interpretations of the two kinds of quantifiers.
We describe these familiar rules first.
Figure~\ref{fig:lilac-prob-semantics} gives the probability-specific 
rules, which we will discuss after.
Both definitions are parameterized by
an ambient sample space $\Omega$ equipped with a $\sigma$-algebra $\Sigma_\Omega$:
all probability spaces $\calP$ are assumed to contain sub-$\sigma$-algebras of $\Sigma_\Omega$,
and $\rv A$ denotes the set of random variables over $\Omega$.
Because $\Omega$ is fixed throughout, we write probability spaces simply as $(\calF,\mu)$.

Figure~\ref{fig:lilac-basic-semantics} defines
the meaning of propositions.
Each proposition
$\Gamma;\Delta\vdash P$
is interpreted as a set of \emph{configurations} of the form $(\gamma,D,\calP)$ by the relation $\gamma, D, \calP \models P$.
In ordinary separation logic, configurations are of the form $(s,h)$
where $h$ is a heap and $s$ a substitution associating values to variables.
Here, the probability space $\calP$ plays the role of the heap.
The pair $(\gamma,D)$ plays the role of the substitution;
because Lilac has two kinds of variables --- random and deterministic --- it
also has two kinds of substitutions: $\gamma\in\sembr\Gamma$ maps each deterministic variable to a value,
and $D\in \rv\sembr\Delta$ maps each random variable to a mathematical random
variable.
The last four lines of Figure~\ref{fig:lilac-basic-semantics}
give familiar-looking interpretations of quantifiers.
All other lines are standard for separation logics,
and follow from the fact that $(\pdot)$ forms a Kripke resource monoid.

\begin{figure}[]
\footnotesize{
\begin{tabular}{lcl}
  $\gamma,D,\calP\vDash \top$       &   always  \\
  $\gamma,D,\calP\vDash \bot$       &   never  \\
  $\gamma,D,\calP\vDash P\land Q$   &   iff  &   $\gamma,D,\calP\vDash P$ and $\gamma,D,\calP\vDash Q$\\
  $\gamma,D,\calP\vDash P\lor Q$    &   iff  &   $\gamma,D,\calP\vDash P$ or $\gamma,D,\calP\vDash Q$\\
  $\gamma,D,\calP\vDash P\to Q$     &   iff  &   $\gamma,D,\calP'\vDash P$ implies $\gamma,D,\calP'\vDash Q$ for all $\calP'\sqsupseteq \calP$\\
  $\gamma,D,\calP\vDash P*Q$        &   iff  &   $\gamma,D,\calP_P\vDash P$ and $\gamma,D,\calP_Q\vDash Q$ for some
  $\calP_P\pdot\calP_Q\plte \calP$\\
  $\gamma,D,\calP\vDash P\wand Q$   &   iff  &   $\gamma,D,\calP_P\vDash P$ implies
  $\gamma,D,\calP_P\pdot\calP\vDash Q$ for all $\calP_P$ with $\calP_P\pdot\calP$ defined\\
  $\gamma,D,\calP\vDash \nec P$     &   iff  &   $\gamma,D, \bfone\vDash P$ \\
  $\gamma,D,\calP\vDash \forall x\ofty S.P$  &   iff  &   $(\gamma,x),D,\calP\vDash P$ for all $x\in S$ \\
  $\gamma,D,\calP\vDash \exists x\ofty S.P$  &   iff  &   $(\gamma,x),D,\calP\vDash P$ for some $x\in S$ \\
  $\gamma,D,\calP\vDash \forallrv X\ofty A.P$  &   iff  &   $\gamma,(D,X),\calP\vDash P$ for all
    $X:\rv A$ \\
  $\gamma,D,\calP\vDash \existsrv X\ofty A.P$  &   iff  &   $\gamma,(D,X),\calP\vDash P$ for some
    $X:\rv A$
\end{tabular}
}
\caption{Semantics of basic Lilac connectives.} 
\label{fig:lilac-basic-semantics}
\end{figure}
\begin{figure}[ht!]
\footnotesize{
\begin{tabular}{lll}
  $\gamma,D,(\calF,\mu)\vDash \own E$             & iff &  $E(\gamma)\circ D$ is $\calF$-measurable\\
  $\gamma,D,(\calF,\mu)\vDash E\sim \mu'$        & iff &  $E(\gamma)\circ D$ is $\calF$-measurable and
                                                              $\mu'(\gamma)=\monadic{
                                                                 &\omega\gets \mu;\\
                                                                 &\ret (E(\gamma)(D(\omega)))
                                                              }$\\
 $\gamma,D,(\calF,\mu)\vDash \Ex[E] = e$    & iff &  $E(\gamma)\circ D$ is $\calF$-measurable and
                                                      $\Ex_{\omega\sim\mu}[E(\gamma)(D(\omega))] = e(\gamma)$ \\
  $\gamma,D,(\calF,\mu)\vDash E_1\asequal E_2$    & iff &  $F\in\calF$ and $\mu(F)=1$ and $F\cup(X_1,X_2)^{-1}(A)\in\calF$ for all $A\in\mathrm{cod}(X_1)\otimes\mathrm{cod}(X_2)$ \\
                                                  &     &  where $F = \{\omega \mid X_1(\omega) = X_2(\omega)\}$
                                                           and $X_i = E_i(\gamma)\circ D$ for $i\in\{1,2\}$ \\

  $\gamma,D,\calP\vDash \weakpre(M,X\ofty A.Q)$ 
                                     & iff &  for all $\calP_\mathrm{frame}$ and $\mu$
                                              with $\calP_\mathrm{frame}\pdot\calP
                                              \plte(\Sigma_\Omega,\mu)$\\
                                     &     &   and all $D_\textrm{ext}:\rv{\sembr{\Delta_\textrm{ext}}}$\\
                                     &     &   there exists $X:\rv A$ and $\calP'$ and $\mu'$ with 
                                                 $\calP_\mathrm{frame}\pdot \calP'
                                                  \plte(\Sigma_\Omega,\mu')$\\
                                     &     &   such that
                                                  $\monadic{
                                                     &\omega\gets \mu;\\
                                                     &v\gets M(\gamma)(D(\omega));\\
                                                     &\ret (D_\textrm{ext}(\omega), D(\omega), v)
                                                   } = \monadic{
                                                     &\omega\gets \mu';\\
                                                     &\ret (D_\textrm{ext}(\omega), D(\omega), X(\omega))
                                                   }$\\
                                     &     &   and $\gamma,(D,X),\calP'\vDash Q$
  ~
\end{tabular}
}
\caption{Semantics of probability-specific Lilac connectives.} 
\label{fig:lilac-prob-semantics}
\end{figure}

Figure~\ref{fig:lilac-prob-semantics} describes the probability-specific 
Lilac connectives. 
We start with the first line in the figure, which defines the meaning of ownership. 
Following the intuition from Section~\ref{sec:tour-of-model}, ownership 
corresponds to measurability of a random variable with respect to a 
particular $\sigma$-algebra. This intuition is made formal here:
the proposition
$\own E$ holds with respect to a configuration $(\gamma,D,(\calF,\mu))$ if 
the random variable denoted by the random expression $E$ is $\calF$-measurable.\footnote{Formally, 
for a probability space $(\Omega, \calF, \mu)$,
a random variable $X : (\Omega, \calF) \rightarrow (A, \mathcal{A})$ is 
\emph{$\calF$-measurable} if for every $E \in \mathcal{A}$ 
it holds that $X^{-1}(E) \in \calF$.
}
The expression $E$ can have free variables that are either 
random or deterministic. The random variable $E(\gamma) \circ D$ is constructed 
by performing
the relevant substitutions:
first all deterministic values are substituted into $E$,
and then the resulting measurable map $E(\gamma)$ is composed
with $D$ to produce a random variable. With this connective in hand, we can 
formally relate separating conjunction in Lilac to the familiar 
probabilistic notion of independence of random variables:
\begin{lemma}[separating conjunction is mutual independence] \label{lem:star-is-independence}
  Fix a configuration $(\gamma,D,\calP)$.
  Abbreviate $X_i(\gamma)\circ D$ as $X_i'$. Then,
   $\indep_i X'_i$ holds with respect to $\calP$
  if and only if $(\gamma,D,\calP) \vDash \hugestar_i \own X_i$.
\end{lemma}
For a proof, see \referto{app:star-is-independence}.
Next, proposition $E\sim \mu'$ holds with respect to $(\gamma,D,(\calF,\mu))$ if
$E$ owns $\calF$ and additionally follows distribution $\mu'(\gamma)$, which is the distribution 
obtained by substituting the values in $\gamma$ for deterministic variables in the distribution
expression $\mu'$. Throughout this figure, we use Haskell-style notation to
construct distributions using the Giry monad~\citep{giry1982categorical}; here we use this notation
on the right hand side of the equation for $\mu'(\gamma)$ to construct the
distribution produced by first sampling a value $\omega$ from the ambient probability
measure $\mu$ and then running $E$ on it. Intuitively, this captures the notion
that $\mu'$ is the push-forward of $\mu$ through $E$. The interpretation
of $\Ex[E]=e$ has a similar structure.

The proposition $E_1\asequal E_2$ holds with respect to $(\gamma,D,(\mathcal F,\mu))$
if $E_1$ and $E_2$ are almost-surely equal: formally, we require the event $F$ that
the random variables $E_1(\gamma)\circ D$ and $E_2(\gamma)\circ D$ agree to have probability $1$.
We additionally require $\calF$
to contain all supersets of $F$ that may be expressed as events involving $E_1$
and $E_2$; this is necessary to support rewriting along equalities $E_1\asequal
E_2$ as illustrated by the examples in
Section~\ref{sec:tour-of-lilac}.\footnote{We would like to thank Jialu Bao for
pointing this out.}
Note that we do not require $E_1$ or $E_2$ to be
$\calF$-measurable: this makes $E_1\asequal E_2$
a \emph{duplicable proposition},\footnote{%
As in \citet{jung2018iris}, we say a proposition $P$ is duplicable if $P\vdash P * P$.}
and allows it to be combined with other propositions using separating conjunction without asserting
spurious independence relationships.
For details on the properties of almost-sure equality, see \referto{app:asequal-properties}.

The most intricate part of Figure~\ref{fig:lilac-prob-semantics} is the interpretation
of our weakest-precondition modality $\weakpre$. 
Intuitively,
configurations of the form $(\gamma,D,\calP)$ represent fragments of
a machine state, much like how a configuration $(s,h)$ in ordinary separation logic
represents a fragment of the full heap.
The idea is that $\weakpre(M, X\ofty A.Q)$ should hold in configuration $(\gamma, D, \calP)$
if
(1) running $M$ with any state
containing fragment $\calP$ produces a new state containing a new
fragment $\calP'$ and a new random variable $X$; (2) the new fragment $\calP'$ satisfies
postcondition $Q$; (3) any fragments $\calP_\textrm{frame}$ independent of $\calP$ are
preserved by $M$, which is necessary to establish a frame rule.
To enforce (1), we quantify over all probability spaces $(\Sigma_\Omega, \mu)$ 
containing $\calP$ and require that running $M$ in $\mu$
produce a new probability space $(\Sigma_\Omega,\mu')$ containing a new
fragment $\calP'$ and new random variable $X$ whose distribution is equal
to the distribution produced by $M$. To enforce (2),
we require that the new configuration $(\gamma,(D,X),\calP')$ satisfy $Q$.
To enforce (3), we quantify over all possible
``frames'' $\calP_\textrm{frame}$, and require that the new space $\mu'$ contain
the exact same frame unchanged.
Finally, in order to prove a fundamental substitution lemma,
we quantify over arbitrary extensions $D_\textrm{ext}$ to the random
substitution $D$; for details on this technical point see \referto{app:substitution-lemma}.

\subsection{Syntax and Semantics of APPL}
\begin{figure}[]
  \centering
  {\footnotesize
\begin{tabular}{rcl}
  $A,B$ & $::=$ & $A\times B\mid \boolty \mid \realty \mid A^n \mid \indexty \mid \applgiry A$ \\
  $M,N,O$ & $::=$ & $X\mid \applret M\mid \letin XMN \mid \appltuple{M,N} \mid \applfst~M \mid \applsnd~M \mid$\\
      &       & $\vtrue \mid \vfalse \mid \applite MNO \mid \flip~p \mid r \mid M \oplus N \mid M \prec N \mid \applunif \mid$ \\
      &       & $\applarray{M,\dots,M} \mid \applindex MN \mid \applfor n{M_\textrm{init}}iX{M_\textrm{step}}$ \\
\end{tabular}
  }
\caption{APPL syntax. Metavariables $p$ range over probabilities, $r$ over real numbers, and $n$
  over natural numbers; $\oplus$ and $\prec$ range over standard arithmetic and comparison
  operators.} \label{fig:appl-syntax}
\end{figure}
Now we establish a program logic that leverages Core Lilac. We fix a small probabilistic programing language 
called APPL capable of expressing the examples in Section~\ref{sec:tour-of-model}.
The syntax of APPL is given in Figure~\ref{fig:appl-syntax}.
It is a simply-typed first-order calculus with a sampling operation, immutable arrays, and bounded loops.
It has a simple monadic type-system as in \citet{staton2020probabilistic}. 
The important monadic typing rules are:

{\footnotesize
\vspace{-1.5em}
\begin{mathpar}
  \inferrule*[right=T-Unif]{~}{\Delta\vdashappl \applunif:\applgiry\realty}
    \and
  \inferrule*[right=T-Ret]{{\Delta\vdashappl M : A}}{\Delta\vdashappl \applret M:\applgiry A}
   \and
  \inferrule*[right=T-Bind]{\Delta\vdashappl M:\applgiry A \\\\ \Delta,X:A\vdashappl N:\applgiry B}{\Delta\vdashappl \letin XMN:\applgiry B}
\end{mathpar}
}%
Monadic computations have type $\applgiry A$; 
the $\applgiry$ stands for the standard Giry monad~\citep{giry1982categorical}.
The \textsc{T-Unif} rule states that $\applunif$ is a probabilistic computation
producing a real number.

The semantics for APPL are standard and follow \citet{staton2020probabilistic}.
Types $A$ are interpreted as measurable spaces $\sembr{A}$
and typing contexts $\Delta = \{x_1 : A_1, \dots, n_n : A_N\}$
as products $\sembr{\Delta} = \sembr{A_1} \times \cdots \times \sembr{A_n}$. 
Programs $\Delta \vdashappl M : \applgiry{} A$ are interpreted as measurable maps 
$\sembr{M} : \sembr{\Delta} \mto \calG \sembr{A}$.
The full semantics can be found in \referto{app:appl-semantics}. 

\subsection{Reasoning About APPL Programs} \label{sec:derived}

\begin{figure}
  {\footnotesize
  \begin{mathpar}
    \inferrule*[right=H-Consequence]{P\vdash P' \\ Q'\vdash Q \\ \triple{P'}MX{Q'}}{\triple PMXQ}
    \and
    \inferrule*[right=H-Frame $(X\notin F)$]{\triple PMXQ}{\triple{F*P}MX{F*Q}}
    \and
    \inferrule*[right=H-Ret]{~}{\triple{Q\big[\sembr M/X\big]}{\applret M}XQ}
    \and
    \inferrule*[right=H-Let]
      {\triple PMXQ \\ \forallrv X.~\triple QNYR}
      {\triple P{\letin XMN}YR}
    \and
    \inferrule*[right=H-Uniform]{~}{\triple\top{\applunif}X{X\sim\unif~[0,1]})}
    \and
    \inferrule*[right=H-Flip]{~}{\triple\top{\flip~p}X{X\sim\ber p}}
    \and
    \inferrule*[right=H-For]
      {\forall i\ofty\N.~\forallrv X\ofty A.~\triple{I(i,X)}M{X'}{I(i+1,X')}}
      {\triple{I(1,e)}{\applfor neiXM}{X\ofty A}{I(n+1,X)}}
    \and
    {\inferrule*[right=H-If]
      {\triple P M X {Q(X)}
       \\\\
       \forallrv X.~\triple {Q(X)} N Y {R(\ite EXY)}
       }
      {\triple{P}{\applite EMN}{Z}{R(Z)}}}
  \end{mathpar}
  }
  \caption{Selected proof rules for reasoning about APPL programs.}
  \label{fig:proof-rules}
\end{figure}

We now show how the semantic model described in the previous section validates
standard proof rules for reasoning about APPL programs.
Using the connectives described in Section~\ref{sec:core-lilac-syntax},
we define the meaning of Hoare triples $\triple PMXQ$ in terms of $\weakpre$, following \citet{jung2018iris}.\footnote{%
  Concretely, $\triple PMXQ  ~:=~ \nec (P\wand \weakpre\left(\sembr M,X.\,Q\right))$;
  see \citet{jung2018iris} for a detailed explanation.
  }
Then, we use the model described in Section~\ref{sec:core-lilac-semantics}
to validate the proof rules in Figure~\ref{fig:proof-rules};
these rules justify the annotated programs given in Section~\ref{sec:tour-of-lilac}.

The structural rules \textsc{H-Consequence} and \textsc{H-Frame} are completely
standard, as are \textsc{H-Ret} and \textsc{H-Let}.
The rules \textsc{H-Uniform} and \textsc{H-Flip}
specify APPL's sampling operations; they formalize the
intuition that sampling is like allocation. The rule \textsc{H-For}
is a standard proof principle for reasoning about APPL's for-loops:
it states that one can conclude postcondition $I(n+1,X)$
after running a for-loop if an invariant
$I(i,X)$ -- a proposition indexed by the loop iteration $i$
and value of the accumulator variable $X$ -- holds on entry of the initial accumulator $e$
and is maintained by every loop iteration.
The rule \textsc{H-If} is used to reason about if-then-else.
Unlike in the traditional setting, a probabilistic program can be thought of as taking
both branches of an if-then-else, since it is possible that a Boolean random variable
is both true and false with nonzero probability.
The \textsc{H-If} rule reflects this: it states that, to establish $R(Z)$,
one can first run the $\applthen$-branch to obtain $X$,
and then run the $\applelse$-branch to obtain $Y$,
and then show that $R$ holds of the \emph{random variable} $(\ite EXY)$
that combines the outcomes of the two branches.

Now we turn our attention to validating these rules with respect to a 
suitable model.
Thus far we have been rather abstract about the ambient sample space $\Omega$
underlying Lilac's semantic model.  At this point we make a concrete choice
in order to validate the proof rules in Figure~\ref{fig:proof-rules}.  The
soundness of \textsc{H-Uniform} and \textsc{H-Flip} require constructing a new
probability space independent of an existing one.
To ensure that it is always possible to construct such a fresh probability space,
we fix a particular choice of $\Omega$ and restrict our Kripke resource monoid to a class of
probability spaces on $\Omega$ with so-called ``finite footprint''; this
guarantees that there is always enough ``room'' in $\Omega$ for new probability spaces to be allocated.

Specifically, we fix $\Omega$ to be the Hilbert cube $[0,1]^\N$, the collection of infinite streams
of real numbers in the interval $[0,1]$; these infinite streams can be thought of as
infinitely-replenishable randomness sources for use throughout a probabilistic program's execution
\cite{culpepper2017contextual,zhang2022reasoning}.
A probability space has finite footprint if it only uses finitely-many
dimensions of the Hilbert cube:

\begin{definition}
  A $\sigma$-algebra $\calF$ on $[0,1]^\N$
  has \emph{finite footprint} if there is some finite $n$ such that every $F\in \calF$ is of the form
  $F'\times\hilbertcube$ for some $F'\subseteq[0,1]^n$.
\end{definition}

Then we restrict our Kripke resource monoid on probability spaces to
only those probability spaces with finite footprint.
With this choice of $\Omega$ and a restriction to suitably-well-behaved probability
spaces in hand, we can validate the above proof rules:

\begin{thm}
  The proof rules in Figure~\ref{fig:proof-rules} are sound.
\end{thm}
\begin{proof}
  The structural rules, \textsc{H-Ret}, and \textsc{H-Let} follow straightforwardly from unwinding the
  definitions of Hoare triples and the interpretations of the logical connectives.
  The rule \textsc{H-For} follows by induction on the number of loop iterations.
  As foreshadowed, the rules \textsc{H-Uniform} and \textsc{H-Flip} require constructing a
  new probability space independent of an existing one; because the existing space
  only exhausts some finite $n$ dimensions of the Hilbert cube, we are free to allocate
  the new probability space in dimensions $n+1$ and above.
  For details see \referto{app:proof-rules}.
\end{proof}

\section{The conditioning modality} \label{sec:modality}

\begin{figure}
{\footnotesize
  \begin{mathpar}
  \inferrule*[lab=C-Entail]{P~\vdash~ Q}{\D_{x\gets E} P~\vdash~ \D_{x\gets E} Q}
  \and
  \inferrule*[lab=\DIndepName]{}{(\own E) * P ~\vdash~ \D_{x\gets E} P}
  \and
  \inferrule*[lab=C-Subst]
    {}
    {\own X \vdash \D_{x\gets X} \big(X \asequal x \big)}
  \and
  \inferrule*[lab=C-Total-Expectation] 
    {}
    {\D_{x\gets X} \Big(\Ex[E] = e\Big) \hspace{.5em}\land\hspace{.5em} \Ex[e[X/x]] = v
      \hspace{.5em}\mathlarger{\vdash}\hspace{.5em} \Ex[E] = v}
    \label{rule:total-expectation}
  \end{mathpar}
}%
  \caption{Selected properties of Lilac's conditioning modality $\D$.}
  \label{fig:C-props}
\end{figure}

So far we have presented Core Lilac, which defines probabilistic interpretations
of the standard separation logic connectives, along with atomic
propositions for making probability-specific assertions. 
Now we describe our
second main contribution: Lilac's modal operator for reasoning
about conditioning. We extend Core Lilac with the proposition
$\D_{x\ofty A\gets E} P$
which states that $P$ holds conditional on the event $E=x$ for all deterministic $x$.
Its typing rule is:

{\small
\vspace{-0.5em}
\[
  \inferrule*[right=T-$\D$]{\Gamma;\Delta\vdashrv E:A \\ \Gamma,x\ofty A;\Delta\vdash P}
    {\Gamma;\Delta\vdash \D_{x\ofty A\gets E} P}
\]}%

Figure~\ref{fig:C-props} lists useful laws about $\D$ (proofs are given in
\referto{app:C-laws}).  The rule \textsc{C-Entail} says $\D$ respects
entailment; this allows ordinary logical reasoning to be carried out under $\D$,
automatically lifting statements and proofs about unconditional probability to
the conditional setting.  The rule \textsc{C-Subst} captures the intuition that
$X$ can be safely replaced by $x$ under $\D_{x\gets X}$.  The remaining rules
express standard facts about conditioning.  The rule \DIndep~ states that if $P$
holds independent of some random expression $E$, then $P$ also holds conditional
on $E=x$ for any $x$; this acts as a form of introduction rule for $\D$.  The
rule \textsc{C-Total-Expectation} states the Law of Total Expectation, a theorem
of probability theory that relates an unconditional expectation to an
expectation over conditional expectations.  As a rule, it says that, to compute
the expectation of a random expression $E$, one can proceed in two stages:
first, compute the conditional expectation of $E$ given $X=x$, yielding some
deterministic expression $e$ in terms of the conditioned $x$; then, compute the
desired unconditional expectation by putting the random $X$ back into $e$ and
taking the expectation of the resulting expression $e[X/x]$. 
Section~\ref{sec:weighted-sampling} will give an example illustrating this
rule's use.

\subsection{Semantics of the Conditioning Modality}
The rules stated in Figure~\ref{fig:C-props} give a powerful and intuitive
framework for reasoning about conditioning that would be familiar to an
experienced probability theorist. 
Our goal in this section is to identify a model that validates
these rules.
Intuitively, a model for entering the conditioning modality involves
reasoning under a new conditioned space:
$(\gamma, D, \calP) \vDash \D_{x\ofty A\gets X}P$ 
holds if for all $x\in A$ there exists some \emph{conditioned space} $\calP_{X=x}$
such that $\gamma, D, \calP|_{X=x} \vDash P$.
We would like to define $\calP|_{X=x}$ using the standard definition of conditional probability:
let $\calP = (\Omega, \calF, \mu)$ and 
define $\calP|_{X=x} = (\Omega, \calF, \mu|_{X=x})$ where 
$\mu|_{X=x}(E) = \mu(E \cap \{X=x\})/\mu(\{X=x\})$.
This definition for the conditioned space $\calP|_{X=x}$ is 
useful for discrete random variables $X$, where it is 
practical to disregard conditioned spaces over null events where $\mu(\{X=x\}) = 0$.
However, if $X$ is a continuous random variable, then by definition $\mu(\{X=x\}) = 0$ 
for all $x$, so these null events cannot be ignored.

In probability theory, \emph{disintegrations} were developed in order to resolve 
this issue and give a natural notion of conditioned spaces for 
continuous random variables~\citep{chang1997conditioning}.
A \emph{disintegration} for a probability space $\calP$ 
with respect to a random variable $X$ is defined as a collection 
of all conditioned spaces $\{\calP|_{X=x}\}_{x\in A}$ satisfying certain measurability
and concentration properties~\citep{chang1997conditioning}. The existence of a 
disintegration for a probability space and random variable is a very strong condition,
and not all probability spaces $\calP$ will have a well-defined 
disintegration for all random variables $X$. The study of disintegrations 
has formally characterized some of the conditions under which 
there exist well-defined notions of disintegration~\citep{chang1997conditioning}.
We leverage 
this knowledge here to design a model for $\D$.

Our strategy will be to identify a suitable class of probability spaces that is
both large enough to accommodate all of our design criteria and examples, and
well-behaved enough to admit all reasonable disintegrations. 
Our starting point in
this search is the \emph{Hilbert cube}, the countable product of
unit intervals $[0,1]^\N$. The Hilbert cube is disintegrable with respect to a large
class of random variables (those whose codomain has a well-behaved $\sigma$-algebra):
\begin{lemma} \label{lem:hilbert-cube-disintegrable}
  Let $X : [0,1]^\N \rightarrow (A,\mathcal{A})$ be a random variable and $\calP$ be a
  probability space on the Hilbert cube. If $\mathcal{A}$ is countably-generated and contains all
  singletons, then there exists a disintegration of $\calP$ with respect to $X$.
\end{lemma}
\begin{proof}
  The Hilbert cube is a complete separable metric space~\cite{srivastava2008course}
  so any probability measure on it is finite Borel; the result follows from Theorem~1.4 of
  \citet{chang1997conditioning}.
\end{proof}

The class of spaces $(A,\mathcal{A})$ required by Lemma~\ref{lem:hilbert-cube-disintegrable} includes
many familiar examples, such as $\R^n$, $\N$, and all finite spaces with the usual powerset
$\sigma$-algebra. 
Since the Hilbert cube is the sample space $\Omega$ underlying Lilac's semantic model,
this result allows us to disintegrate configurations $(\gamma,D,\calP)$
whenever $\calP$ is a probability space whose $\sigma$-algebra is
exactly the Borel $\sigma$-algebra on the Hilbert cube, and whose measure $\mu$ is correspondingly a
Borel measure. However, our configurations are not quite of this form: $\mu$
may be a probability measure on a sub-$\sigma$-algebra on the Hilbert cube,
and such measures unfortunately cannot in general be extended to a Borel measure \cite{ershov1975extension}.
This motivates the next step in our search for suitably-well-behaved probability spaces:

\begin{thm} 
  \label{lem:extends-to-borel-measure-disintegrable}
  Let $\mathcal{M}_\textrm{Borel}$ be the set of probability spaces on the
  Hilbert cube of the form $(\Omega,\calF,\mu)$, where $\mu$ can be extended to
  a Borel measure.  The restriction of the KRM given by
  Theorem~\ref{thm:spaces-form-a-krm} to $\calM_\textrm{Borel}$ is still a KRM.
\end{thm}

A proof of this theorem is in \referto{app:disintegration-proofs}.
The upshot of Theorem~\ref{lem:extends-to-borel-measure-disintegrable} is that
configurations of the form $(\gamma,D,\calP)$ where
$\calP\in\calM_\textrm{Borel}$ can be extended to the Hilbert cube, where they
are disintegrable with respect to suitably-well-behaved random variables
following Lemma~\ref{lem:hilbert-cube-disintegrable}.  The final step in our
search is motivated by the desire to validate rule \DIndep~ in
Figure~\ref{fig:C-props}. The soundness of \DIndep~ requires the ability
to show that a union of negligible sets (a set with measure 0) remains
negligible.  In general this is not the case, so we need to further specialize
our model.  We force these unions to be countable -- from which the result
follows straightforwardly from the axioms of probability -- by restricting
ourselves to probability spaces with \emph{countably-generated
$\sigma$-algebras}. Putting this all together yields the final Kripke resource
monoid underlying Lilac's semantic model:

\begin{thm}
  \label{thm:final-krm}
  Let $\Mdis$ be the set of countably-generated probability spaces $\calP$ that have finite footprint
  and can be extended to a Borel measure on the entire Hilbert cube.
  The restriction of the KRM given by Theorem~\ref{thm:spaces-form-a-krm} to $\Mdis$ is still a KRM.
\end{thm}
For a proof see \referto{app:final-krm}. 
Using Theorem~\ref{thm:final-krm}, we can finally give an interpretation to $\D$:

\begin{lemma}\label{lem:disintegration-def}
  The following interpretation of $\D_{x\ofty A\gets E}P$ is validates
  the rules in Figure~\ref{fig:C-props}:
  \\[0.5em]
  \begin{tabular}{lcl}
    \hspace{2.5em}$\displaystyle\gamma,D,(\calF,\mu)\vDash \D_{x\ofty A\gets E} P$
    & iff
    &{
      \begin{tabular}{l}
        for all $(\Sigma_\Omega,\mu')\sqsupseteq (\calF,\mu)$ \\
        and all disintegrations of $\mu'$ along $E(\gamma)\circ D$ into $\{\nu_x\}_{x\in A}$, \\
        and almost all $x\in A$, it holds that $(\gamma,x),D,(\calF,\nu_x|_\calP)\vDash P$. \\
      \end{tabular}}
  \end{tabular}
\end{lemma}%
For a detailed proof, see \referto{app:C-laws}.

\section{Further examples of applying Lilac} \label{sec:more-examples}
An essential component of evaluating any new program logic is applying it to
validate interesting correctness properties of programs.  Our goal in
this section is to further establish (1) that Lilac can validate examples that
existing probabilistic separation logic approaches can
handle~\citep{barthe2019probabilistic,bao2021bunched}; and (2) give an example
that goes beyond these existing approaches. 

\subsection{Proving a Weighted Sampling Algorithm Correct} \label{sec:weighted-sampling}

To exercise Lilac's support for conditional reasoning, continuous random variables,
and substructural handling of independence, we now prove a sophisticated 
constant-space \emph{weighted sampling algorithm} correct using Lilac.  Suppose
you are given a collection of items $\{x_1, \dots, x_n\}$ each with associated
weight $w_i \in \mathbb{R}^+$.  The task is to draw a sample from the collection
$\{x_i\}$ in a manner where each item is drawn with probability proportional to
its weight. This problem is an instance of \emph{reservoir
sampling}~\citep{efraimidis2006weighted}.

\begin{wrapfigure}{r}{0.5\linewidth}
  ~~\begin{minipage}{0.5\linewidth}
\begin{equation*}
\text{\footnotesize $\begin{aligned}
  &1~\,\appllet W \appleq \applret\applarray{w_1,\dots,w_n} \applin \\
  &2~\,\appllet M \appleq \applret (-\infty) \applin \appllet K \appleq \applret(0)\applin \\
  &{\color{gray}\texttt{ // for i from 1 to n with accumulator (M,K),}}\\
  &3~\,\applforopen n{\appltuple{M,K}}i{\appltuple{M,K}}{\\
    &4~\,\hspace{1em}\appllet S \appleq\applunif \applin \\
    &5~\,\hspace{1em}\appllet U \appleq \applret(S\applexp\applpar{1/\applindex Wi})\applin \\
    &6~\,\hspace{1em}\applif~U>M\\
    &7~\,\hspace{1em}\applthen~\applret \appltuple{U,i}\\
    &8~\,\hspace{1em}\applelse~\applret \appltuple{M,K}}
  \applforclose
\end{aligned}$}
\end{equation*}%
\end{minipage}
\caption{Constant-space reservoir sampling.}
\label{fig:reservoir}
\end{wrapfigure}
A naive solution might first normalize the weights so that they sum to $1$
and then sample from the resulting probability distribution.
Such an approach is inappropriate for application in large-scale systems: it requires storing
all previously encountered weights and scanning over them before a single sample can be
drawn, and so does not scale to a streaming setting where new weights are acquired
one at a time (for instance, as each user visits a website). To fix this,
\citet{efraimidis2006weighted} proposed the 
\emph{constant-space solution} in Figure~\ref{fig:reservoir}.

The core idea is to generate a value $S$ uniformly at random from $[0,1]$ on \emph{every}
iteration (Line~3), perturb $S$ according to the next weight $w_i$ in the stream
(Line~4), and store only the \emph{greatest} perturbed sample (Lines~5--8).
It is a surprising fact that this program is equivalent to the naive one.
To prove it, we will establish the postcondition
   $\forall k. \Pr(K=k) = w_k/\sum_j w_j$.
First, mechanically applying the rules given in Section~\ref{sec:derived}
allows us to conclude the following at exit (for details, 
which involve a loop invariant, see \referto{sec:full-verification}):

\vspace{-1em}
{\small 
\begin{align}
  \existsrv S_1 \dots S_n.~\hugestar_i S_i \sim \mathrm{Unif}~[0,1] ~~*~~ 
  K \asequal \argmax_i S_i^{1/w_i}
  \label{eq:post1}
\end{align}}%
Here $\{S_i\}_i$ are i.i.d.\ random variables with $S_i$ denoting the value sampled by Line~4 on the $i$th
iteration, and $K$ denotes the final result.
The rest of the proof is devoted to showing that (\ref{eq:post1}) entails the desired postcondition.
Given arbitrary $k$, note that
  $\Pr(K=k) = \Pr\big(S_k^{1/w_k} > S_j^{1/w_j} \text{ for all } j \ne k\big)$,
since $K$ is defined to be the $\argmax$ of $j$ over all $S_j^{1/w_j}$.
To make computing this probability tractable, we condition on $S_k$:
fixing $S_k$ to a deterministic $s_k$,

\vspace{-1em}
  {\small 
\begin{align}
  \Pr(K=k \mid S_k = s_k) &= \Pr\big(s_k^{1/w_k} > S_j^{1/w_j} \text{ for all } j \ne k\big) \label{eq:eqset1:start} \\
    &= \Pr\big(s_k^{w_j/w_k} > S_j \text{ for all } j \ne k\big) & \text{Exponentiating} \\
    &= \prod_{j \ne k} \Pr\big(s_k^{w_j/w_k} > S_j\big) & \text{By conditional independence} \label{eq:indep-ex}
\end{align}
  }%
From Equation~\ref{eq:indep-ex} we proceed by calculation. 
If $U \sim \unif[0,1]$, then $\Pr(u > U) = u$; this lets 
us conclude that 
$(\ref{eq:indep-ex}) = \prod_{j \ne k} s_k^{w_j/w_k} = \pow{s_k}{\frac{\sum_{j\ne k} w_j}{w_k}}$.

Formally, this calculation occurs under $\D_{s_k\gets S_k}$, which is introduced via \DIndep.
The expression $\Pr(E)$ abbreviates $\Ex[\ind[E]]$, the expectation of the indicator random variable $\ind[E]$.\footnote{%
If $E$ is an event then the random variable $\ind[E]$ is $1$ if $E$ holds and $0$ otherwise.}
The critical step occurs in Equation~\ref{eq:indep-ex}:
since $\indep_{j\ne k} S_j \mid S_k$, we can apply:

\vspace{-0.75em}
{\small
\begin{align}
  \tag{\textsc{Indep-Prod}}
  {
  \hugestar_i \own E_i ~~~\vdash~~~ \Pr\Big( \bigcap_i E_i \Big) = \prod_i \Pr(E_i),
  } \label{eq:indep-prod}
\end{align}}%
an immediate consequence of Lemma~\ref{lem:star-is-independence}.

Finally, to complete the proof we connect the conditional
$\Pr(K=k \mid S_k = s_k)$ to the unconditional $\Pr(K=k)$
using the following instantiation of \textsc{C-Total-Expectation}:

\vspace{-0.75em}
{\footnotesize
\begin{align*}
  \D_{s_k \gets S_k} \Big(\Ex[\underbrace{\ind[K=k]}_{E}] 
  = \underbrace{\Bigpow{s_k}{\frac{\sum_{j \ne k}{w_j}}{w_k}}}_{e} \Big)
  \land 
  \Big(
  \Ex\Big[\underbrace{\mathrm{pow}\Big(S_k, \frac{\sum_{j \ne k} w_j}{w_k} \Big)}_{e[S_k/s_k]}\Big]
  =
  \underbrace{\frac{w_k}{\sum_j w_j}}_{v}
  \Big)
  ~~\mathlarger{\vdash}~~
  \Ex[\underbrace{\ind[K=k]}_{E}] = \underbrace{\frac{w_k}{\sum_j w_j}}_{v}
\end{align*}
}%
In the left-hand side of this entailment,
the first conjunct follows from the above
and the second conjunct follows from a calculation.
For a detailed presentation of this proof, see \referto{sec:full-verification}.

To sum up, we have shown how Lilac can be used to verify a constant-space weighted sampling
algorithm whose correctness argument requires reasoning about conditional independence
of continuous random variables and imports several important results from probability theory,
including the law of total expectation and key properties of the uniform distribution.
Hopefully, the above example illustrates how Lilac's
substructural handling of independence, modal treatment of conditioning,
and semantic model grounded in familiar constructs from probability theory
allow for easy and natural formalizations of standard informal proofs.

\subsection{An Example of Conditional Independence via Control Flow} \label{sec:condsamples}

For this example, we borrow
the \ref{prog:condsamples} program from Figure~6(b) of
\citet{bao2021bunched} (translated into a functional style):
\begin{equation}
  {\small
  \tag{\textsc{CondSamples}}
\begin{aligned}
  &\appllet Z \appleq \flip~1/2 \applin \\
  &\applite{Z}
      {\left(\begin{aligned}
         &\appllet X_1 \appleq \flip~p \applin \\
         &\appllet Y_1 \appleq \flip~p \applin \\
         &\applret (Z,X_1,Y_1) \\
      \end{aligned}\right)}
      {\left(\begin{aligned}
        &\appllet X_2 \appleq \flip~q \applin \\
        &\appllet Y_2 \appleq \flip~q \applin \\
        &\applret (Z,X_2,Y_2)
      \end{aligned}\right)}
\end{aligned}
\label{prog:condsamples}
  }
\end{equation}
This program produces a tuple $(Z,X,Y)$ with $X$ and $Y$ conditionally independent given $Z$.
The random variables $X$ and $Y$ are sampled from different distributions
depending on the outcome $Z$ of a fair coin flip: if $Z=\vtrue$ then $X$ and $Y$
are Bernoulli random variables with parameter $p$, and if $Z=\vfalse$ then $X$
and $Y$ are Bernoulli random variables with parameter $q$.
The proof of conditional independence, as in the \ref{prog:commoncause} example,
goes by case analysis on $Z$.

Conditional independence of $X$ and $Y$ given $Z$ is expressed by the following triple:
\[
  \lrtriple\top{\ref{prog:condsamples}}{(Z,X,Y)}{\D_{z\gets Z} (\own X * \own Y)}
\]
As usual, the proof begins by mechanically applying proof rules. This yields:
\begin{align*}
  &Z \sim \ber 1/2 ~~*~~ \existsrv~X_1~Y_1~X_2~Y_2.~~
  \left(\begin{aligned}
    &X_1 \sim \ber p ~~*~~ Y_1 \sim\ber p ~~*~~ X_2 \sim \ber q ~~*~~ Y_2 \sim \ber q ~~*~~ \\
    &X \asequal (\ite Z{X_1}{X_2}) ~~*~~ Y \asequal (\ite Z{Y_1}{Y_2})
  \end{aligned}\right)
\end{align*}
This mechanically-derived postcondition makes use of existential quantification
over random variables, written $\existsrv$, in order to talk about the random
variables produced by the $\applthen$ and $\applelse$ branches.
The subformula $X_1\sim\ber p ~~*~~ Y_1\sim\ber p$ is the postcondition derived
for the $\applthen$ branch, and the subformula
$X_2\sim\ber p ~~*~~ Y_2\sim\ber p$ is the postcondition derived
for the $\applelse$ branch.
The almost-sure equalities $X\asequal (\ite Z{X_1}{X_2})$
and $Y\asequal (\ite Z{Y_1}{Y_2})$
combine the variables produced by the individual branches into the variables $X$ and $Y$
produced by the whole if-then-else.

We now proceed as in the \ref{prog:commoncause} example.
First, we condition on $Z$ and replace all occurrences of $Z$ with the newly introduced deterministic variable $z$,
giving
\begin{align*}
  &\D_{z\gets Z} \hspace{0.5em} \underbrace{\left(\existsrv~X_1~Y_1~X_2~Y_2.~~
  \begin{aligned}
    &X_1 \sim \ber p ~~*~~ Y_1 \sim\ber p ~~*~~ X_2 \sim \ber q ~~*~~ Y_2 \sim \ber q ~~*~~ \\
    &X \asequal (\ite z{X_1}{X_2}) ~~*~~ Y \asequal (\ite z{Y_1}{Y_2})
  \end{aligned}\right)}_{P(z)}.
\end{align*}
The goal is to show $\D_{z\gets Z} P(z) \vdash \D_{z\gets Z}(\own X*\own Y)$.
Because $\D$ respects entailment, it suffices to show 
$P(z)\vdash \own X * \own Y$.
This follows by a case analysis on $z$.
If $z=\vtrue$ then $P(z)$ can be simplified to $X\sim\ber p ~~*~~ Y\sim\ber p$,
and if $z=\vfalse$ then $P(z)$ can be simplified to $X\sim \ber q ~~*~~ Y\sim \ber q$.
In both cases the simplified form entails $\own X ~~*~~ \own Y$ as desired.
See \referto{sec:condsamples-annotated} for a fully annotated program.
For more examples of applying Lilac, see \referto{sec:barthe-examples}.

\section{Discussion and future work} \label{sec:discussion}
In this section we explore various possible extensions to Lilac and
expound on the more subtle consequences of some of the design decisions we made
while validating certain proof rules.

\goodparagraph{Properties of the conditioning modality}
Here we investigate further some formal properties of the conditioning
modality. Specifically, we compare $\D$ to modal necessity $\nec$~\citep{kripke1972naming}.
The standard properties of $\nec$ are:\\[0.5em]
\begin{tabular}{lcl}
  \hspace{1em}
  (a) If $\vdash P$ then $\vdash\nec P$ (necessitation).
  & &
  (d) $\nec P \lor \nec Q \vdash \nec(P\lor Q)$.
  \\
  \hspace{1em}
  (b) $\nec (P\to Q)\vdash \nec P\to\nec Q$ (distribution).
  & &
  (e) $\nec P\vdash P$ (axiom M).
  \\
  \hspace{1em}
  (c) $\nec (P\land Q)\vdashiff \nec P\land \nec Q$.
  & &
  (f) $\nec P\vdash \nec\nec P$ (axiom 4).
\end{tabular}\\[0.5em]
The modality $\D_{x\gets X}$ satisfies (a)-(d); for proofs see
\referto{app:C-modal-laws}.
The similarity between $\D$ and modal
necessity is somewhat expected, due to the similarity in the logical structure of their
interpretations: $\D_{x\gets X}P$ requires
$P$ to hold in almost-all conditional probability spaces $\calP|_{X=x}$,
similar to how the usual interpretation of $\nec P$ in modal logic requires
that $P$ hold in all reachable worlds.
We are not sure whether
Axiom 4 holds. Axiom M however has a counterexample -- this is to be expected, as Axiom M in
standard modal logic says that
what is necessary is the case, whereas we do not expect something that holds conditional on $X=x$ to
hold unconditionally, even if it holds conditional on $X=x$ for all $x$.

\goodparagraph{Embedding Lilac into Iris}
In the future 
we would like to embed Lilac in Iris in order to use Iris's
support for reasoning about feature-rich languages and its
interface for carrying out interactive separation logic
proofs~\citep{jung2018iris,krebbers2017interactive}.
This requires expressing Lilac's KRM as a \emph{camera}~\cite{jung2018iris} --- an
object similar to a KRM that additionally supports step-indexed reasoning.
One difference between cameras and KRMs is that,
for cameras, $(\plte)$ is implicitly defined to be the
relation $x\plte y \Leftrightarrow \exists z. x\pdot z = y$.
Lilac's KRM includes ordering relations $\calP\sqsubseteq \calR$
that are not of the form $\calP\pdot\calQ=\calR$ for any $\calQ$,
so embedding Lilac into Iris would require bridging this gap between
KRMs and cameras.
Morever, making use of Iris's support for step-indexed reasoning could require
developing a suitable step-indexed generalization of the KRM in
Theorem~\ref{thm:spaces-form-a-krm}
so that one can talk about probability spaces ``up to $k$ steps.''
We leave these problems for future work.

\goodparagraph{Formal structure of Lilac models}
We made many design decisions while
constructing a model validating Lilac's proof rules.
Following \citet{biering2007bi}, it would be 
interesting future work to pursue a principled characterization of
the space of valid probabilistic models of separation logic; this would
potentially facilitate future extensions
to more sophisticated features
such as higher-order functions, polymorphism, mutable state, and concurrency.

\section{Related work} \label{sec:related-work}
Probabilistic program verification has a long history going back to 
\citet{kozen1983probabilistic}. In this section we sketch the broad themes 
that are most related to program logics for probabilistic 
programs. First, we discuss approaches that make use of separation logic.
Then, we discuss alternative approaches
based on expectations, logical relations, and denotational semantics.

\goodparagraph{Program Logics for Probability} \label{sec:related-work:psl}
The most closely related work is the probabilistic separation logic (PSL) introduced
by \citet{barthe2019probabilistic}, which gives the
first separation logic where separating conjunction explicitly models
probabilistic independence. Follow-on work extends PSL to support
negative
dependence~\citep{bao2022separation} and to settings beyond probabilistic
computation~\citep{zhou2021quantum}. PSL interprets separating
conjunction as a combining operation on distributions over random stores with
disjoint domains, over-approximating the semantic notion of probabilistic
independence with a semi-syntactic criterion on stores.
As a consequence, PSL's notion of independence is linked 
to the occurrences of free variables in logical formulas; statements 
such as $\own (X + Y) * \own (X - Y)$ are inexpressible in PSL
due to the occurrence of the random variables $X$ and $Y$ on both sides 
of $*$.
PSL's frame rule imposes a number of extra side-conditions
capturing data-flow properties of the program.
This is in part due to PSL's notion of separation,
and in part because PSL programs are written using mutable variables
whereas we have preferred to work with a purely functional language.
These side-conditions are nontrivial to check and make applying the frame rule
cumbersome.
Lilac's frame rule is standard for separation logic,
Lilac's interpretation of separating conjunction coincides with probabilistic
independence (Lemma~\ref{lem:star-is-independence}),
and its semantic model is defined in terms of standard objects of probability theory (i.e.,
probability spaces and random variables).
Moreover, Lilac has support for continuous random variables and a modality for reasoning about
conditioning; all of these features in combination seem difficult to add to PSL
without significant changes to its semantic model.
For a concrete comparison,
we validated three of the five examples from \citet{barthe2019probabilistic}:
one-time pad, private information retrieval, and oblivious transfer;
we do not believe the remaining examples exercise Lilac in ways that go
beyond the ones we verified. Validating these examples required no changes to Lilac's semantic model;
it suffices to extend APPL with support for bitvectors
and to import facts about uniformity and independence via a handful of derived rules.
For details, see \referto{sec:barthe-examples}.

\citet{bao2021bunched} extends PSL to handle conditional independence by
extending the standard logic of bunched implications underlying separation logic
with a family of specially-designed connectives
in a new logic called \emph{doubly-bunched implications} (DIBI). The corresponding
model required for proving soundness of DIBI deviates significantly from the usual model
of separation logic. Lilac handles conditional
independence via the conditioning modality, and this extension does not require
any changes to the standard model beyond the restriction to well-behaved probability
spaces (Theorem~\ref{thm:final-krm}).
As
a consequence, Lilac behaves very similarly to existing separation logics while
still having facilities for handling conditional independence.
For a concrete comparison, the \ref{prog:commoncause}
example presented in Section~\ref{sec:tour-of-lilac}
gives a Lilac proof of conditional independence for one of the examples from \citet{bao2021bunched};
Section~\ref{sec:condsamples} gives a description of the other example.

A separate line of logics seeks to verify probabilistic 
programs without a substructural notion of independence.  An example of
this is Ellora~\cite{barthe2018assertion}, where
independence is encoded as an assertion about factorization of probabilities.
This is similar to how in program logics without separating conjunction,
aliasing can be ruled out by asserting pairwise-disjointness of
heap locations.
Ellora is equipped with the ability to abstract over these definitions
via special-purpose logics such as a \emph{law and independence logic}
for reasoning about mutual independence
relationships, but these embedded logics are rather limited:
\citet{barthe2019probabilistic} note that the resulting independence logic
cannot handle conditional control flow and that it is more ergonomic to handle
independence substructurally.  This limitation
was a primary motivation for developing PSL.

Another strategy for designing a separation logic for probabilistic programs is
embodied by Polaris~\citep{tassarotti2019separation}, an extension of Iris for
verifying concurrent randomized algorithms.  The goal of Polaris is very
different from Lilac's, and so it makes different design choices.  The notion of
separation in Polaris is the standard one, enforcing
ownership of disjoint heap fragments. To reason about probability, Polaris
enriches base Iris with the ability to make coupling-style arguments.
Polaris has no substructural treatment of independence or
method for stating facts involving conditioning,
and does not support continuous random variables.
Yet another way to generalize separation logic to the probabilistic setting 
is given by \citet{batz2019quantitative}, who introduced quantitative separation 
logic (QSL). QSL generalizes the meaning of assertions: rather than interpreting 
assertions as predicates on configurations, i.e. functions from configurations to 
Boolean values, QSL interprets predicates as functions from configurations to 
expectations. In QSL, separating conjunction does not model independence as
in Lilac or PSL.

\goodparagraph{Expectation-based approaches}
Classically the dominant approach to verifying randomized algorithms has been
expectation-based techniques such as PPDL~\citep{kozen1983probabilistic} and
pGCL~\citep{morgan1996probabilistic}. These approaches reason about 
expected quantities of probabilistic programs via a weakest-pre-expectation 
operator that propagates information about expected values backwards 
through the program. These methods have been widely-used in practice, verifying 
properties such as probabilistic bounds and running-times of 
randomized algorithms~\citep{gretz2014operational,olmedo2016reasoning,kaminski2016inferring}.
However, expectation-based approaches verify a single property about expectations 
at a time; verifying multiple interwoven properties of expectations can require 
multiple separate passes, leading to cumbersome and non-modular proofs.
These limitations in expectation-based approaches were an important motivation 
for the development of probabilistic program logics like Ellora~\citep{barthe2018assertion}.

\goodparagraph{Logical Relations for Probabilistic Programs}
A separate method for reasoning about probabilistic programs recasts 
reasoning problems as problems of program equivalence.
Logical relations are a proof-technique for characterizing program equivalence,
and recent work has generalized this strategy to the probabilistic
setting.
\citet{bizjak2015step} characterize equivalence for a
language with recursive types, polymorphism, and first-order mutable references;
\citet{culpepper2017contextual} and \citet{wand2018contextual}
treat continuous random variables and scoring;
\citet{zhang2022reasoning} study nested queries.
In each case, equivalence is characterized using a step-indexed biorthogonal logical relation
constructed over an operational semantics.
Though logical relations are well-suited for proving the validity of 
program rewrite rules, they are less
well-suited for proving intricate post-conditions that can be stated 
in a program logic.

\goodparagraph{Probabilistic Denotational Semantics}
An entirely separate method for verifying probabilistic programs
performs all reasoning in a suitably-well-behaved denotational model.
For example, \citet{staton2017commutative} validates
intuitive laws such as commutativity of let-bindings
by interpreting programs in an appropriate category.
Recently, there have been significant developments towards designing 
convenient general-purpose
models~\cite{heunen2017convenient,staton2016semantics,fritz2020synthetic,stein2021structural},
and a possible avenue for future work is to replace Lilac's Giry-monad-based
semantics with these richer domains in order to support more language
features (e.g., higher-order functions).

\section{Conclusion}
Lilac is a probabilistic separation logic with support for
continuous random variables and conditional reasoning
whose interpretation of separating conjunction coincides with the ordinary
notion of probabilistic independence.
The core contributions of Lilac are (1) a novel notion
of separation based on independent combination of probability spaces; and (2) a
modal treatment of conditional probability, which includes a set of proof rules for
reasoning about conditioning that would be intuitive to an
experienced probability theorist.
To demonstrate Lilac, we derived proof rules for reasoning about 
a simple probabilistic programming language
and showed how they can be used in combination with Lilac's other features
to prove a sophisticated weighted sampling algorithm correct.
Notably, the derived proof rules mirror those of ordinary separation logic:
rules for sampling resemble the usual rules for allocation,
and our frame rule is completely standard.
Ultimately, we envision Lilac becoming a standard tool 
in the toolkit for verifying probabilistic programs. For future work, we 
are curious if Lilac can be extended to the quantum programming setting 
in a style similar to \citet{zhou2021quantum}, or if it can 
handle the exotic forms of negative dependence studied in \citet{bao2022separation}.

\section*{Acknowledgments}
We thank our shepherd, Joseph Tassarotti, and
the anonymous reviewers for their careful feedback and suggestions. 
This work was supported by the National Science Foundation under Grant No.
\#CCF-2220408.

\bibliography{psl}

\ifdefined\cameraready\else

\pagebreak
\appendix

\section{Syntax and semantics of APPL}

\subsection{Syntax}

{\small{
\begin{center}
\begin{tabular}{rcl}
  $p$  & $\in$ & $[0,1]$ \\
  $r$  & $\in$ & $\R$ \\
  $n$  & $\in$ & $\N$ \\
  $\oplus$ & $\in$ & $\mathrm{Arith} := \{+,-,\times,/,\string^\}$ \\
  $\prec$ & $\in$ & $\mathrm{Cmp} := \{<,\le,=\}$ \\
  $A,B$ & $::=$ & $A\times B\mid \boolty \mid \realty \mid A^n \mid \indexty \mid \applgiry A$ \\
  $M,N,O$ & $::=$ & $X\mid \applret M\mid \letin XMN \mid$ \\
      &       & $\appltuple{M,N} \mid \applfst~M \mid \applsnd~M \mid$ \\
      &       & $\vtrue \mid \vfalse \mid \applite MNO \mid \flip~p \mid $ \\
      &       & $r \mid M \oplus N \mid M \prec N \mid \applunif \mid$ \\
      &       & $\applarray{M,\dots,M} \mid \applindex MN \mid \applfor n{M_\textrm{init}}iX{M_\textrm{step}}$ \\
\end{tabular}
\end{center}
}}

\subsection{Typing}
\label{app:appl-typing}

{\small{
\begin{center}
\begin{mathpar}
  \inferrule{}{\Gamma,X:A\vdashappl X:A}
  \and
  \inferrule{\Gamma\vdashappl M : A}{\Gamma\vdashappl \applret M:\applgiry A}
  \and
  \inferrule{\Gamma\vdashappl M:\applgiry A \\ \Gamma,X:A\vdashappl N:\applgiry B}{\Gamma\vdashappl \letin XMN:\applgiry B}
  \and
  \inferrule{\Gamma\vdashappl M:A\\\Gamma\vdashappl N:B}{\Gamma\vdashappl \appltuple{M,N}:A\times B}
  \and
  \inferrule{\Gamma\vdashappl M:A\times B}{\Gamma\vdashappl \applfst~M:A}
  \and
  \inferrule{\Gamma\vdashappl M:A\times B}{\Gamma\vdashappl \applsnd~M:B}
  \and
  \inferrule{}{\Gamma\vdashappl \vtrue:\boolty}
  \and
  \inferrule{}{\Gamma\vdashappl \vfalse:\boolty}
  \and
  \inferrule{\Gamma\vdashappl M:\boolty\\ \Gamma\vdashappl N:A \\ \Gamma\vdashappl O:A}
    {\Gamma\vdashappl \applite{M}{N}{O}:A}
  \and
  \inferrule{}{\Gamma\vdashappl \flip~p:\applgiry\boolty}
  \and
  \inferrule{}{\Gamma\vdashappl r:\realty}
  \and
  \inferrule{\Gamma\vdashappl M:\realty\\ \Gamma\vdashappl N:\realty}{\Gamma\vdashappl M\oplus N:\realty}
  \and
  \inferrule{\Gamma\vdashappl M:\realty\\ \Gamma\vdashappl N:\realty}{\Gamma\vdashappl M\prec N:\boolty}
  \and
  \inferrule{}{\Gamma\vdashappl \applunif:\applgiry\realty}
  \and
  \inferrule{\Gamma\vdashappl M_k:A\textrm{ for all }1\le k\le n}{\Gamma\vdashappl \applarray{M_1,\dots,M_n}:A^n}
  \and
  \inferrule{\Gamma\vdashappl M:A^n \\ \Gamma\vdashappl N:\indexty}{\Gamma\vdashappl \applindex MN:A}
  \and
  \inferrule{\Gamma\vdashappl M_\textrm{init}:A \\ \Gamma,i:\indexty,X:A\vdashappl M_\textrm{step}:\applgiry A}
    {\Gamma\vdashappl \applfor n{M_\textrm{init}}iX{M_\textrm{step}}:\applgiry A}
\end{mathpar}
\end{center}
}}

\clearpage

\subsection{Semantics}
\label{app:appl-semantics}

Types and typing contexts are interpreted as measurable spaces,
arithmetic operators as maps $\R\times\R\mto\R$,
and comparison operators as maps $\R\times\R\mto\sembr\boolty$.
{\small{
\begin{align*}
  \begin{aligned}
    \sembr{A \times B }&= \sembr{A }\otimes\sembr{B } \\
    \sembr\boolty &= (\{\vtrue,\vfalse\}, \calP(\{\vtrue,\vfalse\})) \\
    \sembr\realty &= (\R, \calB(\R)) \\
    \sembr{A ^n}&= \underbrace{\sembr{A }\otimes\cdots\otimes\sembr{A }}_{n\textrm{ times}} \\
    \sembr{\indexty}&= (\N, \calP(\N)) \\ 
    \sembr{\applgiry A} &= \giry \sembr A \\
    \\
    \sembr\cdot &= \textrm{the one-point space} \\
    \sembr{\Gamma,X\ofty A } &= \sembr\Gamma\otimes \sembr A 
  \end{aligned}
  &&
  \begin{aligned}
    x\sembr{+}y &= x+y \\
    x\sembr{-}y &= x-y \\
    x\sembr{\times}y &= xy \\
    x\sembr{/}y &= \begin{cases} x/y, & y\ne 0 \\ 0, & \textrm{otherwise} \end{cases} \\
    x\sembr{\string^}y &= \begin{cases} x^y, & x>0 \\ 0, & \textrm{otherwise} \end{cases} \\
  \end{aligned}
  &&
  \begin{aligned}
    x\sembr{<}y &= \begin{cases} \vtrue, & x < y \\ \vfalse, & \textrm{otherwise} \end{cases} \\
    x\sembr{\le}y &= \begin{cases} \vtrue, & x \le y \\ \vfalse, & \textrm{otherwise} \end{cases} \\
    x\sembr{=}y &= \begin{cases} \vtrue, & x = y \\ \vfalse, & \textrm{otherwise} \end{cases} 
  \end{aligned}
\end{align*}
}}

\noindent Terms $\Gamma\vdash M:A$ are interpreted as maps $\sembr M:\sembr\Gamma\mto\sembr A$.
{\small{
  \begin{align*}
  \sembr{X}\rho &~~=~~ \rho(X) \\
  \sembr{\applret M}\rho &~~=~~ \ret\sembr M\rho\\
  \sembr{\letin X{M}{N}}\rho &~~=~~ v\gets \sembr{M}\rho;~ \sembr{N}\rho[X\mapsto v] \\
  \sembr{\appltuple{M,N}}\rho &~~=~~ (\sembr M\rho,\sembr N\rho)\\
  \sembr{\applfst~M}\rho &~~=~~ \pi_1(\sembr M\rho) \\
  \sembr{\applsnd~M}\rho &~~=~~ \pi_2(\sembr M\rho) \\
  \sembr{\vtrue}\rho &~~=~~ \vtrue \\
  \sembr{\vfalse}\rho &~~=~~ \vfalse \\
  \sembr{\applite{M}{N}{O}}\rho &~~=~~ \begin{cases}
                                           \sembr{N}\rho,& \sembr M\rho=\vtrue \\
                                           \sembr{O}\rho,& \sembr M\rho=\vfalse
                                   \end{cases}\\
  \sembr{\flip~p}\rho &~~=~~ \ber p \\
  \sembr{r}\rho &~~=~~ r \\
  \sembr{M\oplus N}\rho &~~=~~ (\sembr{M}\rho)\sembr{\oplus}(\sembr{N}\rho) \\
  \sembr{\applunif}\rho &~~=~~ \unif~[0,1]\\
  \sembr{\applarray{M_1,\dots,M_n}}\rho &~~=~~ (\sembr{M_1}\rho,\dots,\sembr{M_n}\rho)\\
  \sembr{\applindex MN:A }\rho &~~=~~ \begin{cases}
                                  \pi_{\sembr{N}\rho}(\sembr{M}\rho), & 1\le \sembr{N}\rho\le n \\
                                  \arbitrary(A ), & \textrm{otherwise}
                                \end{cases} \\
  \sembr{\applfor n{M_\textrm{i}}iX{M_\textrm{s}}}\rho &~~=~~
    \mathrm{loop}(1, \sembr{M_\textrm{i}}\rho, \lambda~k~v.~\sembr{M_\textrm{s}}\rho[i\mapsto k,X\mapsto v])\\
  &\hspace{-9em}\textrm{where } \mathrm{loop}(k,v,f) = \begin{cases}
      \ret v, & k>n \\
      v'\gets f(k,v);~
      \mathrm{loop}(k+1,v',f), & \textrm{otherwise}
    \end{cases}
  \end{align*}
}}
To make array indexing total, $\arbitrary(A )$ produces an arbitrary inhabitant
of $\sembr A $.
{\small{
  \begin{align*}
    \arbitrary(A \times B ) &= (\arbitrary(A ),\arbitrary(B )) \\
    \arbitrary(\boolty) &= \vtrue \\
    \arbitrary(\realty) &= 0 \\
    \arbitrary(A ^n) &= \underbrace{(\arbitrary(A ),\dots,\arbitrary(A ))}_{n\textrm{ times}} \\
    \arbitrary(\indexty) &= 0
  \end{align*}
}}

\section{Syntax and semantics of Lilac}

\subsection{Syntax}
\label{app:lilac-syntax}

{\small{
\begin{align*}
  S,T \,\in~ &\mathbf{Set} \\
  A,B \,\in~ &\mathbf{Meas} \\
  P,Q ::=~ &\top \mid \bot \mid P \land Q \mid P \lor Q \mid P\to Q \mid \\
           &P * Q \mid P\wand Q \mid \nec P \mid \\
           &\forall x\ofty S.P \mid \exists x\ofty S. P \mid
            \forallrv X\ofty A.P \mid \existsrv X\ofty A.P \mid \\
           &E \sim \mu \mid \own E \mid E\asequal E \mid \Ex[E] = e \mid \weakpre(M, X\ofty A. Q)
\end{align*}
}}

\subsection{Typing}
\label{app:lilac-typing}

{\small{
\begin{align*}
  \Gamma &::= \cdot \mid \Gamma,x:S\\
  \Delta &::= \cdot \mid \Delta,X:A
\end{align*}
\begin{align*}
  \sembr{x_1:S_1,\dots,x_n:S_n} &= S_1\times\cdots\times S_n \\
  \sembr{X_1:A_1,\dots,X_n:A_n} &= A_1\otimes\cdots\otimes A_n 
\end{align*}
\begin{mathpar}
  \inferrule{E \in \sembr\Gamma\to\sembr\Delta\mto A}
            {\Gamma;\Delta\vdashrv E:A}
  \and
  \inferrule{e \in \sembr\Gamma\to A}
            {\Gamma\vdashdet e:A}
  \and
  \inferrule{\mu \in \sembr\Gamma\to \giry A}
            {\Gamma\vdashdet \mu:A}
  \and
  \inferrule{M\in \sembr\Gamma\to \sembr\Delta\mto \giry A}{\Gamma;\Delta\vdashprog M:A}
\end{mathpar}
\begin{mathpar}
  \inferrule{}{\Gamma;\Delta\vdash \top}
  \and
  \inferrule{}{\Gamma;\Delta\vdash \bot}
  \and
  \inferrule{\Gamma;\Delta\vdash P\\\Gamma;\Delta\vdash Q}{\Gamma;\Delta\vdash P\land Q}
  \and
  \inferrule{\Gamma;\Delta\vdash P\\\Gamma;\Delta\vdash Q}{\Gamma;\Delta\vdash P\lor Q}
  \and
  \inferrule{\Gamma;\Delta\vdash P\\\Gamma;\Delta\vdash Q}{\Gamma;\Delta\vdash P\to Q}
  \and
  \inferrule{\Gamma;\Delta\vdash P\\\Gamma;\Delta\vdash Q}{\Gamma;\Delta\vdash P*Q}
  \and
  \inferrule{\Gamma;\Delta\vdash P\\\Gamma;\Delta\vdash Q}{\Gamma;\Delta\vdash P\wand Q}
  \and
  \inferrule{\Gamma;\Delta\vdash P}{\Gamma;\Delta\vdash \nec P}
  \\
  \inferrule{\Gamma,x\ofty S;\Delta\vdash P}{\Gamma;\Delta\vdash \forall x\ofty S. P}
  \and
  \inferrule{\Gamma,x\ofty S;\Delta\vdash P}{\Gamma;\Delta\vdash \exists x\ofty S. P}
  \and
  \inferrule{\Gamma;\Delta,X\ofty A\vdash P}{\Gamma;\Delta\vdash \forallrv X\ofty A. P}
  \and
  \inferrule{\Gamma;\Delta,X\ofty A\vdash P}{\Gamma;\Delta\vdash \existsrv X\ofty A. P}
  \\
  \inferrule{\Gamma;\Delta\vdashrv E:A \\ \Gamma\vdashdet \mu:A}
            {\Gamma;\Delta\vdash E\sim\mu}
  \and
  \inferrule{\Gamma;\Delta\vdashrv E:A}
            {\Gamma;\Delta\vdash \own E}
  \and
  \inferrule{\Gamma;\Delta\vdashrv E_1:A \\ \Gamma;\Delta\vdashrv E_2:A}
            {\Gamma;\Delta\vdash E_1\asequal E_2}
  \and
  \inferrule{\Gamma;\Delta\vdashrv E:\R \\ \Gamma\vdashdet e:\R}
            {\Gamma;\Delta\vdash \Ex[E] = e}
  \and
  \inferrule{\Gamma;\Delta\vdashprog M:A \\
             \Gamma;\Delta,X\ofty A\vdash Q}
            {\Gamma;\Delta\vdash \weakpre(M,X\ofty A.Q)}
\end{mathpar}
}}

\subsection{Independent combination of probability spaces}
\label{app:indep-combination}
\label{app:independent-combination-app}
 \begin{lemma}[independent combinations are unique]
  Suppose $(\Omega, \calG,\rho)$ and $(\Omega, \calG',\rho')$ are independent combinations of
  $(\Omega, \calE,\mu)$ and $(\Omega, \calF,\nu)$.  Then $\calG = \calG'$ and $\rho=\rho'$.
\end{lemma}
\begin{proof}
  It is straightforward to establish that $\calG = \calG'$: they are both the 
  smallest $\sigma$-algebra containing $\calE$ and $\calF$.
  Showing $\rho = \rho'$ requires the use of more heavyweight machinery 
  from probability theory: we apply the well-known
  \emph{Dynkin $\pi$-$\lambda$ theorem}~\citep{kallenberg1997foundations}.
  The set of events on which $\rho$ and $\rho'$ agree forms a $\lambda$-system,
  and the set $\{E\cap F\mid E\in\calE,F\in\calF\}$ of intersections of events in $\calE$ and
  $\calF$ forms a $\pi$-system that generates $\angled{\calE,\calF} = \calG$. So the $\pi$-$\lambda$ theorem 
  states that it suffices to show
  $\rho(E\cap F)=\rho'(E\cap F)$ for all $E\in\calE$ and $F\in\calF$; this follows since by 
  assumption both sides of the equation factorize into $\mu(E)\nu(F)$.
\end{proof}

\begin{lemma}\label{pairwise-intersections-pi-system}
If $\calF$ and $\calG$ are $\sigma$-algebras on $\Omega$
then the set $\calE:=\{F\cap G\mid F\in\calF,G\in\calG\}$
of intersections of events in $\calF$ and $\calG$
is a $\pi$-system that generates $\angled{\calF,\calG}$.
\end{lemma}
\begin{proof}
First let's show that $\calE$ is a $\pi$-system.
The set $\calE$ is nonempty because it at least has to contain $\emptyset$.
It's closed under finite intersections because if $(F_1\cap G_1)\in\calE$ and $(F_2\cap
G_2)\in\calE$
then $(F_1\cap G_1)\cap(F_2\cap G_2)=\underbrace{(F_1\cap F_2)}_{\in\calF}\cap\underbrace{(G_1\cap
G_2)}_{\in\calG}\in\calE$, where the last step follows from the fact that
$\mathcal F$ and $\mathcal G$ are both $\sigma$-algebras and hence closed under intersections.

Now we just have to show $\angled{\calE}=\angled{\calF,\calG}$.
As sets of generators, $\calE$ contains the union of $\calF$ and $\calG$: because $\calF$ and $\calG$
are $\sigma$-algebras $\calE$ includes intersections of the form $F\cap\Omega=F$ and $\Omega\cap G=G$ for all
$F\in\calF$ and $G\in\calG$. This implies $\angled{\calF,\calG}\subseteq\angled{\calE}$.
For the other direction, note that
every generator $(F\cap G)\in\calE$ is an
intersection of generators $F\in\calF,G\in\calG$.
\end{proof}

\begin{lemma}\label{independence-lambda-system}
If $(\calF,\mu)$ is a probability space then $\calF^\perp := \{E\mid E\perp\calF\}$\footnote{$E\perp\calF$ iff $\mu(E\cap F)=\mu(E)\mu(F)$ for all $F\in\calF$.} is a
$\lambda$-system. 
\end{lemma}
\begin{proof}
Clearly $\emptyset\in\calF^\perp$ because $\emptyset\perp E$ for any $E$.
If $E\perp\calF$ then $E^c\perp\calF$, so $\calF$ is closed under complements.
Finally, if $\{A_n\}_{n\in\N}$ is a collection of disjoint sets in $\calF^\perp$
then $\mu(\bigcup_n A_n\cap E)=\sum_n\mu(A_n\cap E)=\sum_n\mu(A_n)\mu(E)=\mu(E)\mu(\bigcup_n A_n)$
for all $E$, so $\calF^\perp$ is closed under countable disjoint union.
\end{proof}

\begin{thm} \label{app:spaces-form-a-krm}
  Let $\calM$ be the set of probability spaces over a fixed sample space $\Omega$.
  Let $(\pdot)$ be the partial function mapping two probability spaces to their independent
  combination if it exists.
  Let $(\plte)$ be the ordering such that $(\calF,\mu)\plte(\calG,\nu)$ iff
  $\calF\subseteq\calG$ and $\mu=\nu|_\calF$.\footnotemark
  The tuple $(\calM,\plte,\pdot,\mathbf{1})$ is a Kripke resource monoid,
  where $\mathbf{1}$ is the trivial probability space $(\calF_\pmbone,\mu_\pmbone)$
  with $\calF_\pmbone=\{\emptyset,\Omega\}$ and $\mu_\pmbone(\Omega) = 1$.
\end{thm}
\begin{proof}
$\bfone$ is indeed a unit: if $(\calF,\mu)$ is some other probability space on $\Omega$
then $\angled{\calF,\calF_\pmbone}=\calF$ and $\mu$ witnesses the independent
combination of itself with $\mu_\pmbone$.
And the relation ``$\calP$ is an independent combination of $\calQ$ and $\calR$''
is clearly symmetric in $\calQ$ and $\calR$, so $(\pdot)$ is commutative.
We just need to show $(\pdot)$ is associative and respects $(\plte)$.

For associativity, suppose
$(\calF_1,\mu_1)\pdot(\calF_2,\mu_2)=(\calF_{12},\mu_{12})$
and $(\calF_{12},\mu_{12})\pdot(\calF_3,\mu_3)=(\calF_{(12)3},\mu_{(12)3})$.
There are three things to check: \begin{itemize}
\item Some $\mu_{23}$ witnesses the combination of $(\calF_2,\mu_2)$ and $(\calF_3,\mu_3)$.
\item Some $\mu_{1(23)}$ witnesses the combination of $(\calF_1,\mu_1)$ and $(\calF_{23},\mu_{23})$.
\item $(\angled{\calF_1,\angled{\calF_2,\calF_3}},\mu_{1(23)})=
(\angled{\angled{\calF_1,\calF_2},\calF_3},\mu_{(12)3})$.
\end{itemize}
We'll show this as follows: \begin{enumerate}
\item $\angled{\calF_1,\angled{\calF_2,\calF_3}}= \angled{\angled{\calF_1,\calF_2},\calF_3}$.
\item Define $\mu_{23}:=\mu_{(12)3}|_{\calF_{23}}$.
      This is a witness for $(\calF_2,\mu_2)$ and $(\calF_3,\mu_3)$.
\item Define $\mu_{1(23)}:=\mu_{(12)3}$. This
      is a witness for $(\calF_1,\mu_1)$ and $(\calF_{23},\mu_{23})$.
\end{enumerate}
To show the left-to-right inclusion for (1): by the universal property of freely-generated
$\sigma$-algebras, we just need to show
$\angled{\angled{\calF_1,\calF_2},\calF_3}$ is a $\sigma$-algebra containing
$\calF_1$ and $\angled{\calF_2,\calF_3}$. It clearly contains $\calF_1$.
To show it contains $\angled{\calF_2,\calF_3}$, we just need to show it contains
$\calF_2$ and $\calF_3$ (by the universal property again), which it clearly does. The right-to-left
inclusion is similar.

For (2), if $E_2\in\mathcal F_2$ and $E_3\in\mathcal F_3$ then
$\mu_{23}(E_2\cap E_3)=\mu_{(12)3}(E_2\cap E_3)=\mu_{(12)3}((\Omega\cap E_2)\cap E_3)
=\mu_{12}(\Omega\cap
E_2)\mu_3(E_3)=\mu_1(\Omega)\mu_2(E_2)\mu_3(E_3)=\mu_2(E_2)\mu_3(E_3)$ as desired.

For (3), we need
$\mu_{(12)3}(E_1\cap E_{23})=\mu_1(E_1)\mu_{23}(E_{23})$
for all $E_1\in\calF_1$ and $E_{23}\in\angled{\calF_2,\calF_3}$.
For this we use the $\pi$-$\lambda$ theorem.
Let $\calE$ be the set $\{E_2\cap E_3\mid E_2\in\calF_2,E_3\in\calF_3\}$ of intersections of
events in $\calF_2$ and $\calF_3$. $\calE$ is a $\pi$-system that generates
$\angled{\calF_2,\calF_3}$ (lemma \ref{pairwise-intersections-pi-system}).
Let $\calG$ be the set of events $E_{23}$ such that $\mu_{(12)3}(E_1\cap
E_{23})=\mu_1(E_1)\mu_{23}(E_{23})$ for all $E_1\in\calF_1$.
We are done if $\angled{\calE}\subseteq\calG$.
By the $\pi$-$\lambda$ theorem, we just need to check that $\calE\subseteq\calG$ and
that $\calG$ is a $\lambda$-system.
We have $\calE\subseteq\calG$ because if $E_2\in\calF_2$ and $E_3\in\calF_3$
then $\mu_{(12)3}(E_1\cap (E_2\cap E_3))=\mu_1(E_1)\mu_2(E_2)\mu_3(E_3)=\mu_1(E_1)\mu_{23}(E_2\cap
E_3)$.
To see that $\calG$ is a $\lambda$-system, note that
$\mu_1(E_1)\mu_{23}(E_{23})=\mu_{(12)3}(E_1)\mu_{(12)3}(E_{23})$
and so $\calG$ is actually equal to $\calF_1^\perp$ (the set of events independent of $\calF_1$), a $\lambda$-system by Lemma
\ref{independence-lambda-system}.

To show $(\pdot)$ respects $(\plte)$,
  suppose $(\calF,\mu)\plte (\calF',\mu')$ and
  $(\calG,\nu)\plte (\calG',\nu')$
  and $(\calF',\mu')\pdot(\calG',\nu') =(\angled{\calF',\calG'},\rho')$.
  We need to show (1) $(\calF,\mu)\pdot(\calG,\nu)
  =(\angled{\calF,\calG},\rho)$ and
  (2) $(\angled{\calF,\calG},\rho)\plte(\angled{\calF',\calG'},\rho')$
  for some $\rho$. Define $\rho$ to be the restriction
  of $\rho'$ to $\angled{\calF,\calG}$.
  Now (1) holds because
  $\rho(F\cap G)
  =\rho'(F\cap G)
  =\rho'(F)\rho'(G)
  =\rho(F)\rho(G)
  $ for all $F\in\calF$ and $G\in\calG$ (the second step follows from $\calF\subseteq\calF'$ and
  $\calG\subseteq\calG'$). For (2),
  $\angled{\calF,\calG}\subseteq\angled{\calF',\calG'}$
  because $\calF\subseteq\calF'$ and $\calG\subseteq\calG'$, and
  $\rho=\rho'|_{\angled{\calF,\calG}}$ by
  construction.
\end{proof}

\subsection{Semantics}
\label{app:lilac-semantics}

Let $\Omega$ be the Hilbert cube $[0,1]^\N$,
and let $\Sigma_\Omega$ be the standard Borel $\sigma$-algebra on the Hilbert cube
generated by the product topology.

\begin{definition}
  A sub-$\sigma$-algebra $\calF$ of $\Sigma_\Omega$
  has \emph{finite footprint}
  if there is some $n$ such that every $F\in \calF$ is of the form
  $F'\times\hilbertcube$ for some $F'\subseteq[0,1]^n$.
\end{definition}

\begin{definition}
  A random variable $X:(\Omega,\Sigma_\Omega)\to (A,\Sigma_A)$ has \emph{finite footprint}
  if the pullback $\sigma$-algebra $\{X^{-1}(E)\mid E\in\Sigma_A\}$ has finite footprint.
\end{definition}

\begin{lemma} \label{app:lem:finite-footprint-remains-krm}
Let $\Mfin$ be the set of probability spaces $\calP$ with finite footprint whose
$\sigma$-algebras are sub-$\sigma$-algebras of the standard Borel $\sigma$-algebra on $[0,1]^\N$.
  The restriction of the KRM given by Theorem~\ref{app:spaces-form-a-krm} to $\Mfin$ is still a KRM.
\end{lemma}
\begin{proof}
  If $m$ and $n$ witness the finite footprints of independently-combinable
  probability spaces $\calP$ and $\calQ$ then $\max(m,n)$ witnesses the finite footprint of 
  their independent combination $\calP\pdot\calQ$,
  and if $\calF$ and $\calG$ are two sub-$\sigma$-algebras of the Borel $\sigma$-algebra on the
  Hilbert cube, then so is the $\sigma$-algebra $\angled{\calF,\calG}$.
  Thus $(\pdot)$ remains closed under $\Mfin$,
  which suffices to show that it remains a KRM.
\end{proof}

Let $\rv A$ be the set of measurable maps $[0,1]^\N\mto A$ with finite footprint.
Interpret propositions $\Gamma;\Delta\vdash P$ as sets of configurations
$(\gamma,D,\calP)$ where $\gamma\in\sembr\gamma$, $D\in\rv{\sembr\Delta}$,
and $\calP\in \Mfin$.

\begin{lemma}
  The following interpretations of basic connectives is well-formed:
{\small{
\begin{center}
\begin{tabular}{lcl}
  $\gamma,D,\calP\vDash \top$       &   always  \\
  $\gamma,D,\calP\vDash \bot$       &   never  \\
  $\gamma,D,\calP\vDash P\land Q$   &   iff  &   $\gamma,D,\calP\vDash P$ and $\gamma,D,\calP\vDash Q$\\
  $\gamma,D,\calP\vDash P\lor Q$    &   iff  &   $\gamma,D,\calP\vDash P$ or $\gamma,D,\calP\vDash Q$\\
  $\gamma,D,\calP\vDash P\to Q$     &   iff  &   $\gamma,D,\calP'\vDash P$ implies $\gamma,D,\calP'\vDash Q$ for all $\calP'\sqsupseteq \calP$\\
  $\gamma,D,\calP\vDash P*Q$        &   iff  &   $\gamma,D,\calP_P\vDash P$ and $\gamma,D,\calP_Q\vDash Q$ for some
  $\calP_P\pdot\calP_Q\plte \calP$\\
  $\gamma,D,\calP\vDash P\wand Q$   &   iff  &   $\gamma,D,\calP_P\vDash P$ implies
  $\gamma,D,\calP_P\pdot\calP\vDash Q$ for all $\calP_P$ with $\calP_P\pdot\calP$ defined\\
  $\gamma,D,\calP\vDash \nec P$     &   iff  &   $\gamma,D,1\vDash P$ \\
  $\gamma,D,\calP\vDash \forall x\ofty S.P$  &   iff  &   $(\gamma,x),D,\calP\vDash P$ for all $x\in S$ \\
  $\gamma,D,\calP\vDash \exists x\ofty S.P$  &   iff  &   $(\gamma,x),D,\calP\vDash P$ for some $x\in S$ \\
  $\gamma,D,\calP\vDash \forallrv X\ofty A.P$  &   iff  &   $\gamma,(D,X),\calP\vDash P$ for all
    $X:\rv A$ \\
  $\gamma,D,\calP\vDash \existsrv X\ofty A.P$  &   iff  &   $\gamma,(D,X),\calP\vDash P$ for some
    $X:\rv A$ \\
  $\gamma,D,(\calF,\mu)\vDash E\sim \mu'$        & iff &  $E(\gamma)\circ D$ is $\calF$-measurable and
                                                              $\mu'(\gamma)=\monadic{
                                                                 &\omega\gets \mu;\\
                                                                 &\ret (E(\gamma)(D(\omega)))
                                                              }$\\
  $\gamma,D,(\calF,\mu)\vDash \own E$             & iff &  $E(\gamma)\circ D$ is $\calF$-measurable\\
  $\gamma,D,(\calF,\mu)\vDash E_1\asequal E_2$    & iff &  $F\in\calF$ and $\mu(F)=1$ and $F\cup(X_1,X_2)^{-1}(A)\in\calF$ for all $A\in\mathrm{cod}(X_1)\otimes\mathrm{cod}(X_2)$ \\
                                                  &     &  where $F = \{\omega \mid X_1(\omega) = X_2(\omega)\}$
                                                           and $X_i = E_i(\gamma)\circ D$ for $i\in\{1,2\}$ \\
  $\gamma,D,(\calF,\mu)\vDash \Ex[E] = e$    & iff &  $E(\gamma)\circ D$ is $\calF$-measurable and
                                                      $\Ex_{\omega\sim\mu}[E(\gamma)(D(\omega))] = e(\gamma)$ \\
  $\gamma,D,\calP\vDash \weakpre(M,X\ofty A.Q)$ 
                                     & iff &  for all $\calP_\mathrm{frame}$ and $\mu$
                                              with $\calP_\mathrm{frame}\pdot\calP
                                              \plte(\Sigma_\Omega,\mu)$\\
                                     &     &   and all $D_\textrm{ext}:\rv{\sembr{\Delta_\textrm{ext}}}$\\
                                     &     &   there exists $X:\rv A$ and $\calP'$ and $\mu'$ with 
                                                 $\calP_\mathrm{frame}\pdot \calP'
                                                  \plte(\Sigma_\Omega,\mu')$\\
                                     &     &   such that
                                                  $\monadic{
                                                     &\omega\gets \mu;\\
                                                     &v\gets M(\gamma)(D(\omega));\\
                                                     &\ret (D_\textrm{ext}(\omega), D(\omega), v)
                                                   } = \monadic{
                                                     &\omega\gets \mu';\\
                                                     &\ret (D_\textrm{ext}(\omega), D(\omega), X(\omega))
                                                   }$\\
                                     &     &   and $\gamma,(D,X),\calP'\vDash Q$
  ~
\end{tabular}
\end{center}
}}
\end{lemma}
\begin{proof}
  We must verify that the extended random substitutions
  $(D,X)$ in the interpretations of $\forallrv$, $\existsrv$, and $\weakpre$
  have finite footprint; in all cases this follows from the fact that $D$ and $X$ have finite footprint.
\end{proof}

\begin{lemma}[separating conjunction is mutual independence]  \label{app:star-is-independence}
  Fix a configuration $(\gamma,D,\calP)$.
  Abbreviating $X_i(\gamma)\circ D$ as $X_i'$,
  random variables $X_1',\dots,X_n'$ are mutually independent with respect to $\calP$
  iff $\gamma,D,\calP\vDash \own X_1 * \dots * \own X_n$.
\end{lemma}
\begin{proof}
  First suppose $\gamma,D,\calP\vDash \own X_1 * \dots * \own X_n$, so
  each $X_i'$ is $\calP_i$-measurable for some $\calP_1\pdot\dots\pdot\calP_n\plte\calP$.
  Write $\calP=(\calF,\mu)$ and $\calP_i=(\calF_i,\mu_i)$ for all $1\le i\le n$.
  For any subset $J$ of $\{1,\dots,n\}$ and any collection of events
  $\{E_j \in \calF_j\}_{j\in J}$, 
  we have \[\Pr\left[\bigwedge_{j\in J} X_j'\in E_j\right] = \mu\left(\bigcap_{j\in J} X_j'^{-1}(E_j)\right)
    \stackrel{(a)}=\prod_{j\in J} \mu_j(X_j'^{-1}(E_j)) \stackrel{(b)}= \prod_{j\in J} \mu(X_j'^{-1}(E_j))
  = \prod_{j\in J} \Pr[X_j'\in E_j]\]
  where $(a)$ and $(b)$ hold because $\calP_1\pdot\dots\pdot\calP_n\plte\calP$
  and $X_j'^{-1}(E_j)\in\calF_j$ for all $j$. Hence $X_1',\dots,X_n'$ are mutually independent.

  For the converse, suppose $X_1',\dots,X_n'$ are mutually independent with respect to some probability
  space $\calP$.
  For each $1\le i\le n$, let $\calP_i$ be the probability space $(\calF_i,\mu_i)$
  where $\calF_i$ is the pullback $\sigma$-algebra along $X_i'$ and $\mu_i$ the restriction of $\mu$ to
  $\calF_i$.
  It's enough to show that the composition
  $\calP_1\pdot\dots\pdot\calP_n$ is defined, as 
  then $\calP_1\pdot\dots\pdot\calP_n\plte \calP$ by lemma
  \ref{lem:independent-combinations-are-unique}.
  This follows by induction on $n$.
  Cases $n=0$ and $n=1$ are immediate.
  Now suppose $\calP_1\pdot\dots\pdot \calP_k$ is defined.
  It's straightforward to show that $\calP_1\pdot\dots\pdot \calP_k$ is the pullback of the random variable
  $(X_1',\dots,X_k')$,
  and $\calP_{k+1}$ is the pullback of $X_{k+1}'$ by definition.
  Mutual independence of $X_1',\dots,X_k',X_{k+1}'$
  implies independence of $(X_1',\dots,X_k')$ and $X_{k+1}'$: intersections of
  events $X_1'^{-1}(E_1)\cap\dots\cap X_k'^{-1}(E_k)$ form a $\pi$-system that
  generates $\calF_1\pdot\dots\pdot\calF_k$,
  events independent of $\calF_{k+1}$ with respect to $\mu$ form a $\lambda$-system,
  and each intersection $X_1'^{-1}(E_1)\cap\dots\cap X_k'^{-1}(E_k)$ is independent of $\calF_{k+1}$
  because $X_1',\dots,X_k',X_{k+1}'$ are mutually independent.
  Thus the pullback of $(X_1',\dots,X_k',X_{k+1}')$ is an independent combination
  of $\calP_1\pdot\dots\pdot\calP_k$ and $\calP_{k+1}$, and $\calP_1\pdot\dots\pdot\calP_k\pdot\calP_{k+1}$
  is defined.
This closes the induction, so $\gamma,D,\calP\vDash \own X_1'\pdot\dots\pdot X_n'$ as desired.
\end{proof}

\subsubsection{Substitution} \label{app:substitution-lemma}

Substitutions take the form $(s,S)$ where $s$ is a substitution
of deterministic values and $S$ a substitution of random variables.

\begin{mathpar}
  \inferrule{s\in \sembr{\Gamma'}\to\sembr\Gamma \\ S\in\sembr{\Delta'}\mto\sembr\Delta}
    {\Gamma';\Delta'\vdash (s,S):\Gamma;\Delta}
\end{mathpar}
{\small{
\begin{tabular}{cc}
  $\fbox{\inferrule{\Gamma;\Delta\vdashrv E:A \\ \Gamma';\Delta'\vdash (s,S):\Gamma;\Delta}
    {\Gamma';\Delta'\vdash E[s,S] : A}}$
  &
  $E[s,S](\gamma) = E(s(\gamma)) \circ S$
\\\\
  $\fbox{\inferrule{\Gamma\vdashdet e:A \\ \Gamma'\vdash s:\Gamma}
    {\Gamma'\vdash e[s] : A}}$
  &
  $  e[s] = e \circ s $
\\\\
  $\fbox{\inferrule{\Gamma\vdashdet \mu:A \\ \Gamma'\vdash s:\Gamma}
    {\Gamma'\vdash \mu[s] : A}}$
  &
  $ \mu[s] = \mu \circ s $
\\\\
  $\fbox{\inferrule{\Gamma;\Delta\vdashprog M:A \\ \Gamma';\Delta'\vdash (s,S):\Gamma;\Delta}
    {\Gamma';\Delta'\vdash M[s,S] : A}}$
  &
  $ M[s,S](\gamma) = M(s(\gamma)) \circ S $
\\\\
  $\fbox{\inferrule{\Gamma;\Delta\vdash P \\ \Gamma';\Delta'\vdash (s,S):\Gamma;\Delta}
    {\Gamma';\Delta'\vdash P[s,S]}}$
  &
  $\begin{aligned}
    \top[s,S] &= \top \\ 
    \bot[s,S] &= \bot \\ 
    (P\land Q)[s,S] &= (P[s,S]\land Q[s,S]) \\ 
    (P\lor Q)[s,S] &= (P[s,S]\lor Q[s,S]) \\ 
    (P\to Q)[s,S] &= (P[s,S]\to Q[s,S]) \\ 
    (P* Q)[s,S] &= (P[s,S]* Q[s,S]) \\ 
    (P\wand Q)[s,S] &= (P[s,S]\wand Q[s,S]) \\ 
    (\nec P)[s,S] &= \nec P[s,S] \\ 
    (\forall x\ofty T.P)[s,S] &= \forall x\ofty T. P[s\times 1_T,S] \\ 
    (\exists x\ofty T.P)[s,S] &= \exists x\ofty T. P[s\times 1_T,S] \\ 
    (\forallrv X\ofty A.P)[s,S] &= \forallrv X\ofty A. P[s,S\times 1_A] \\ 
    (\existsrv X\ofty A.P)[s,S] &= \existsrv X\ofty A. P[s,S\times 1_A] \\ 
    (E\sim\mu)[s,S] &= E[s,S] \sim \mu[s] \\
    (\own E)[s,S] &= \own E[s,S] \\
    (E_1 \asequal E_2)[s,S] &= E_1[s,S] \asequal E_2[s,S]\\
    (\Ex[E] = e)[s,S] &= \Ex[E[s,S]] = e[s] \\
    \weakpre(M,X\ofty A.Q)[s,S] &= \weakpre(M[s,S],X\ofty A. Q[s,S\times 1_A])
  \end{aligned}$
\end{tabular}
}}

\begin{lemma}[syntactic and semantic substitution coincide] \label{subst-lemma}
  $\gamma,D,\calP\vDash P[s,S]$ iff $s(\gamma), S\circ D,\calP\vDash P$.
\end{lemma}
\begin{proof}
  By induction on the syntax of propositions.
  The interesting cases are:
  \begin{itemize}
  \item Case $E\sim \mu$:
    \begin{align*}
       &\gamma,D,\calP\vDash (E\sim \mu)[s,S]\\
       &\textrm{iff } \gamma, D,\calP\vDash ((\gamma\mapsto E(\gamma)\circ S)\sim (\mu\circ S)) \\
       &\textrm{iff } E(s(\gamma))\circ S\circ D \textrm{ is }\calP\textrm{-measurable and } 
           \mu(\gamma) = \monadic{&\omega\gets \calP\\ &\ret E(\gamma)(S(D(\omega)))} \\
       &\textrm{iff } s(\gamma), S\circ D,\calP\vDash E\sim \mu
    \end{align*}
  \item Case $\weakpre(M,X\ofty A.Q)$:
    For the left-to-right direction,
    suppose (1) $\gamma,D,\calP\vDash \weakpre(M[s,S], X\ofty A.  Q[s,S\times 1_A])$
    and (2) $\calP_\textrm{frame}\pdot \calP\plte(\Sigma_\Omega,\mu)$
    and (3) $D_\textrm{ext}:\rv{\sembr{\Delta_\textrm{ext}}}$.
    By (1) there exist $\calP'$ and $\mu'$ with
    $\calP_\textrm{frame}\pdot\calP'\plte(\Sigma_\Omega,\mu')$
    and $X:\rv A$ such that
    \[
      \monadic{
         &\omega\gets \mu;\\
         &v\gets M[s,S](\gamma)(D(\omega));\\
         &\ret (D_\textrm{ext}, D(\omega), v)
       } = \monadic{
         &\omega\gets \mu';\\
         &\ret (D_\textrm{ext}, D(\omega), X(\omega))
       }
    \]
    and $\gamma,(D,X),\calP\vDash Q[s,S\times 1_A]$.
    By IH this is equivalent to $s(\gamma),(S\circ D,X),\calP\vDash Q$
    and simplifying the above equation gives
    \[
      \monadic{
         &\omega\gets \mu;\\
         &v\gets M(s(\gamma))(S(D(\omega)));\\
         &\ret (D_\textrm{ext}, D(\omega), v)
       } = \monadic{
         &\omega\gets \mu';\\
         &\ret (D_\textrm{ext}, D(\omega), X(\omega))
       }
    \]
    Now postcomposing both sides with the map
      $(\delta_\textrm{ext}, \delta,v)\mapsto (\delta_\textrm{ext}, S(\delta),v)$
    gives
    \[
      \monadic{
         &\omega\gets \mu;\\
         &v\gets M(s(\gamma))(S(D(\omega)));\\
         &\ret (D_\textrm{ext}, S(D(\omega)), v)
       } = \monadic{
         &\omega\gets \mu';\\
         &\ret (D_\textrm{ext}, S(D(\omega)), X(\omega))
       }
    \]
    so that $(\calP',\mu',X)$ witnesses $s(\gamma),S\circ D,\calP\vDash \weakpre(M,X\ofty A. Q)$ as
    desired.

    For the right-to-left direction, suppose 
    (1) $s(\gamma),S\circ D,\calP\vDash \weakpre(M, X\ofty A.  Q)$
    and (2) $\calP_\textrm{frame}\pdot \calP\plte(\Sigma_\Omega,\mu)$
    and (3) $D_\textrm{ext}:\rv{\sembr{\Delta_\textrm{ext}}}$.
    Specialize (1) with $D_\textrm{ext}:= (D_\textrm{ext},D)$
    to get $\calP'$, $\mu'$ and $X:\rv A$ such that
    \[
      \monadic{
         &\omega\gets \mu;\\
         &v\gets M(s(\gamma))(S(D(\omega)));\\
         &\ret ((D_\textrm{ext}, D), S(D(\omega)), v)
       } = \monadic{
         &\omega\gets \mu';\\
         &\ret ((D_\textrm{ext}, D), S(D(\omega)), X(\omega))
       }
    \]
    and $s(\gamma),(S\circ D, X),\calP'\vDash Q$.
    By IH this is equivalent to $\gamma,(D,X),\calP'\vDash Q[s,S\times 1_A]$,
    and postcomposing both sides of the above equation with the map
    $((\delta_\textrm{ext},\delta),\_,v) \mapsto (\delta_\textrm{ext},\delta,v)$
    and rewriting $M(s(\gamma))(S(D(\omega)))$ in terms of $M[s,S]$ gives \[
      \monadic{
         &\omega\gets \mu;\\
         &v\gets M[s,S](\gamma)(D(\omega));\\
         &\ret (D_\textrm{ext}, D, v)
       } = \monadic{
         &\omega\gets \mu';\\
         &\ret (D_\textrm{ext}, D, X(\omega))
       }
    \]
    so that $(\calP',\mu',X)$ witnesses
    $\gamma,D,\calP\vDash \weakpre(M[s,S], X\ofty A.  Q[s,S\times 1_A])$
    as desired.
  \end{itemize}
\end{proof}

\subsection{Properties of almost-sure equality} \label{app:asequal-properties}

\begin{definition}[spaces that support equality]
  Say that a measurable space $(\mathbf{A},\calA)$ \emph{supports equality}
  if the diagonal $\Delta_{\mathbf{A}} := \{(a,a) \mid a\in \mathbf{A}\}$ is measurable in $\calA\otimes\calA$.
  This includes all Hausdorff spaces, in particular all of the examples
  considered in this paper.
\end{definition}

\begin{lemma} \label{lem:asequal-tfae}
  Let $(\Omega,\calF,\mu)$ be a probability space, $F$ an event in $\calF$,
  $(\mathbf{A},\calA)$ a measurable space that supports equality,
  and $X_1,X_2 : \Omega\to \mathbf A$ random variables.
  The following are equivalent:
  \begin{enumerate}
    \item For all $A\in\calA\otimes\calA$ it holds that $F \cup (X_1,X_2)^{-1}(A)\in\calF$.
    \item For all $A\in\calA\otimes\calA$ it holds that $F^c \cap (X_1,X_2)^{-1}(A)\in\calF$.
    \item $F \cup (X_1,X_2)^{-1}(\calA\otimes\calA)\subseteq\calF$,
      where $(X_1,X_2)^{-1}(\calA\otimes\calA)$ denotes the pullback $\sigma$-algebra
      and $F \cup (X_1,X_2)^{-1}(\calA\otimes\calA)$ is defined to be
      the set $\{F \cup (X_1,X_2)^{-1}(A) \mid A \in \calA\otimes\calA\}$.
    \item $F\cup X_1^{-1}(\calA)\subseteq\calF$ and $F\cup X_2^{-1}(\calA)\subseteq\calF$.
  \end{enumerate}
\end{lemma}
\begin{proof}
  (1) and (3) are equivalent by definition.
  (1) and (2) are equivalent because $\calF$ is closed under complements
  and preimages $(X_1,X_2)^{-1}(A)$ are in bijection with their complements
  ${(X_1,X_2)^{-1}(A)}^c$ via ${(X_1,X_2)^{-1}(A)}^c = (X_1,X_2)^{-1}(A^c)$.
  This establishes the equivalence of (1), (2), and (3).
  Finally, (3) implies (4)
  because $(X_1,X_2)^{-1}(\calA\otimes\calA)$ contains both $X_1^{-1}(\calA)$
  and $X_2^{-1}(\calA)$, so it only remains to show (4) implies (3).
  Suppose (4) with the goal of showing
  $F\cup (X_1,X_2)^{-1}(A)\in\calF$ for all $A\in\calA\otimes\calA$.
  First note that if $A$ is of the form $A_1\times A_2$ for some $A_1,A_2\in\calA$
  then
  \begin{align*}
    F\cup (X_1,X_2)^{-1}(A_1\times A_2) &= F\cup (X_1^{-1}(A_1) \cap X_2^{-1}(A_2)) 
    = \underbrace{(F\cup X_1^{-1}(A_1))}_{\in\calF\text{ by (4)}}
       \cap \underbrace{(F\cup X_2^{-1}(A_2))}_{\in\calF\text{ by (4)}}
  \end{align*}
  from which the result follows because $\calF$ is closed under finite intersections.
  Next note that the set of events $A$ for which $F\cup (X_1,X_2)\inv(A)\in\calF$ forms a $\sigma$-algebra:
  \begin{itemize}
  \item $F\cup \emptyset = F\in\calF$,
  \item If $F\cup (X_1,X_2)\inv(A)\in\calF$ then $F\cup ((X_1,X_2)\inv(A))^c\in\calF$ because
    \begin{align*}
      F\cup ((X_1,X_2)\inv(A))^c &= \Omega \cap (F\cup ((X_1,X_2)\inv(A))^c) \\
      &= (F\cup F^c) \cap (F\cup ((X_1,X_2)\inv(A))^c) \\
      &= F \cup (F^c \cap ((X_1,X_2)\inv(A))^c) \\
      &= F \cup {\underbrace{(F \cup (X_1,X_2)\inv(A))}_{\in\calF}}^c
    \end{align*}
    from which the result follows because $\calF$ is closed under complements and finite unions.
  \item If $\{A_i\}$ is a countable family with $F\cup (X_1,X_2)\inv(A_i)\in\calF$
    for all $i$, then
    $F\cup (X_1,X_2)\inv(\bigcup_i A_i) = \bigcup_i (F\cup (X_1,X_2)\inv(A_i))\in\calF$ because $\calF$
    is closed under countable unions.
  \end{itemize}
  Together these two points imply that the set of events $A$ for which $F\cup (X_1,X_2)\inv(A)\in\calF$
  is a $\sigma$-algebra containing all measurable boxes $A_1\times A_2$ for $A_1,A_2\in\calA$.
  Since $\calA\otimes \calA$ is the smallest $\sigma$-algebra containing all such boxes,
  we have $F\cup (X_1,X_2)\inv(A)\in\calF$ for all $A \in \calA\otimes\calA$ as desired.
\end{proof}

\begin{lemma}[almost-sure equality is an equivalence relation]
  Let $(\mathbf{A}, \calA)$ be a measurable space that supports equality.
  Let $E_1$, $E_2$, and $E_3$ be random expressions of type $\mathbf{A}$.
  The following entailments hold:
  \begin{mathpar}
    \inferrule*[lab=Refl]{}{\vdash E_1 \asequal E_1}
    \and
    \inferrule*[lab=Sym]{}{E_1\asequal E_2\vdash E_2\asequal E_1}
    \and
    \inferrule*[lab=Trans]{}{E_1\asequal E_2\land E_2\asequal E_3\vdash E_1\asequal E_3}
  \end{mathpar}
\end{lemma}
\begin{proof}
  Fix a configuration $(\gamma,D,(\calF,\mu))$. Define
  $X_i = E_i(\gamma)\circ D$ for $i\in\{1,2,3\}$
  and let $F_{ij}$ be the event $\{\omega\mid X_i(\omega) =X_j(\omega)\}$
  that $X_i$ and $X_j$ are equal.
  \begin{itemize}
    \item \textsc{Refl}: we need to show $F_{11} \in \calF$ and $\mu(F_{11}) = 1$,
      and that $F_{11} \cup (X_1,X_1)^{-1}(A) \in \calF$ for all $A\in\calA\otimes\calA$.
      By definition $F_{11} = \{\omega\mid X_1(\omega) = X_1(\omega)\} = \Omega$,
      and because $(\calF,\mu)$ is a probability space, we have that
      $\Omega\in\calF$ and $\mu(\Omega)=1$.
      Since $F_{11} = \Omega$, we have $F_{11} \cup F' = F_{11} \in \calF$ for all $F'$,
      so $F_{11} \cup (X_1,X_1)^{-1}(A)\in \calF$ for all $A$ as required.
    \item \textsc{Sym}: we have $F_{12} \in \calF$
      and $\mu(F_{12}) = 1$ and $F_{12} \cup (X_1,X_2)^{-1}(A)\in\calF$ for all $A\in\calA\otimes\calA$,
      and need $F_{21}\in\calF$ and $\mu(F_{21}) = 1$ and $F_{21} \cup (X_2,X_1)^{-1}(A)\in\calF$ for all $A\in\calA\otimes\calA$.
      This follows from $F_{12} = F_{21}$ and the fact that preimages $(X_2,X_1)^{-1}(A)$
      are in bijective correspondence with preimages $(X_1,X_2)^{-1}(A)$ by
      $(X_2,X_1)^{-1}(A) = (X_1,X_2)^{-1}(\mathrm{swap}(A))$ where $\mathrm{swap}(x,y) = (y,x)$.
    \item \textsc{Trans}: by Lemma~\ref{lem:asequal-tfae}, we have 
      \begin{enumerate}
        \item $F_{12}\in\calF$ and $\mu(F_{12}) = 1$ and $F_{12} \cup (X_1,X_2)^{-1}(\calA\otimes\calA)\subseteq\calF$
        \item $F_{23}\in\calF$ and $\mu(F_{23}) = 1$ and $F_{23} \cup (X_2,X_3)^{-1}(\calA\otimes\calA)\subseteq\calF$
      \end{enumerate}
      and need $F_{13}\in\calF$ and $\mu(F_{13}) = 1$ and $F_{13} \cup (X_1,X_3)^{-1}(\calA\otimes\calA)\subseteq\calF$.
      \begin{itemize}
        \item $F_{13}\in\calF$: 
           note that $F_{13} = (F_{12} \cap F_{23}) \uplus (F_{12}^c \cap F_{23}^c \cap F_{13})$;
            unwinding the notation, this states the following equivalence of events:
            \[
              X_1 = X_3
              \iff (X_1 = X_2 \land X_2 = X_3)
              \lor (X_1 \ne X_2 \land X_2 \ne X_3\land X_1 = X_3)
            \]
          Since $\calF$ is closed under finite unions and intersections, we are done if we can show that
          $F_{12}$, $F_{23}$, and $(F_{12}^c \cap F_{23}^c \cap F_{13})$ are in $\calF$.
          By (1) and (2) and Lemma~\ref{lem:asequal-tfae} we have that $\calF$ contains
          \[
            F_{12}^c \cap X_1^{-1}(\calA) \qquad
            F_{12}^c \cap X_2^{-1}(\calA) \qquad
            F_{23}^c \cap X_2^{-1}(\calA) \qquad
            F_{23}^c \cap X_3^{-1}(\calA)
          \]
          We have $\mathbf{A}\in\calA$ because $\calA$ is a $\sigma$-algebra,
          so $\calF$ also contains $F_{12}^c$ and $F_{23}^c$.
          By closure under intersections $\calF$ also contains
          \[
            F_{23}^c \cap (F_{12}^c \cap X_1^{-1}(\calA)) \qquad
            F_{12}^c \cap (F_{23}^c \cap X_3^{-1}(\calA))
          \]
          Now by Lemma~\ref{lem:asequal-tfae} again, we have that $\calF$ contains
            $F_{12}^c \cap F_{23}^c \cap (X_1,X_3)^{-1}(\calA\otimes\calA)$.
          In particular, since $\calA$ supports equality, we have $\Delta_{\mathbf A}\in \calA\otimes\calA$,
          so $F_{12}^c \cap F_{23}^c \cap (X_1,X_3)^{-1}(\Delta_{\mathbf A})\in \calF$.
          Now $(X_1,X_3)^{-1}(\Delta_{\mathbf A}) = F_{13}$ by definition,
          so $F^c_{12} \cap F^c_{23} \cap F_{13}\in\calF$.
          Along the way we have shown $F_{12}^c\in\calF$ and $F_{23}^c\in\calF$,
          which implies $F_{12}\in\calF$ and $F_{23}\in\calF$ by closure under complements,
          as required.
        \item $\mu(F_{13}) = 1$: we have $F_{12}\cap F_{23}\subseteq F_{13}$
          (by transitivity of equality on the functions $X_1,X_2,X_3$)
          and $\mu(F_{12}) = \mu(F_{23}) = 1$ by assumption. Thus
          \begin{align*}
            \mu(F_{13})
            \ge \mu(F_{12}\cap F_{23})
            = 1 - \mu(F_{12}^c\cup F_{23}^c)
            \ge 1 - (\mu(F_{12}^c) + \mu(F_{23}^c))
            = 1
          \end{align*}
          as required.
        \item $F_{13} \cup (X_1,X_3)^{-1}(\calA\otimes\calA)\subseteq\calF$:
          by Lemma~\ref{lem:asequal-tfae} it suffices to show
           $F_{13}^c \cap X_1^{-1}(\calA)\subseteq\calF$ and
           $F_{13}^c \cap X_3^{-1}(\calA)\subseteq\calF$.
          To this end fix arbitrary $A\in \calA$ with aim to show
          $F_{13}^c\cap X_1\inv(A)\in\calF$ and $F_{13}^c\cap X_3\inv(A)\in\calF$.
          \begin{itemize}
            \item $F_{13}^c \cap X_1\inv(A)\in\calF$:
              note that
              $F_{13}^c = (F_{12}^c \cap F_{23}) \uplus (F_{12}^c \cap F_{23}^c \cap F_{13}^c) \uplus (F_{12} \cap F_{23}^c)$;
              unwinding the notation, this states the following equivalence of events:
              \[ X_1 \ne X_3 \iff \begin{aligned}
                  &(X_1 \ne X_2 = X_3) \\
                  &\lor~(X_1\ne X_2 \land X_2 \ne X_3 \land X_1 \ne X_3) \\
                  &\lor~(X_1=X_2\ne X_3)
                \end{aligned} \]
              Thus the intersection $F_{13}^c \cap X_1\inv(A)$ can be rewritten as the following union:
              \[
                  (F_{12}^c \cap F_{23} \cap X_1\inv(A)) \uplus
                  (F_{12}^c \cap F_{23}^c \cap F_{13}^c \cap X_1\inv(A)) \uplus
                  (F_{12} \cap F_{23}^c \cap X_1\inv(A))
              \]
              It only remains to show that each component of this union is in $\calF$.
              As in the proof of $F_{13}\in\calF$ above, we have that $\calF$ contains
              each of the following:
              \[ F_{12} \qquad F_{23} \qquad F_{13} \qquad
                 (F_{12}^c \cap F_{23}^c) \cap (X_1,X_3)\inv(\calA\otimes\calA) \]
              The first component $F_{12}^c\cap F_{23}\cap X_1^{-1}(A)$ is
              an intersection of two events
              $F_{23}$ and $F_{12}^c \cap X_1^{-1}(A)$ that are in $\calF$ by assumption.
              Similarly the second component
              is equal to $F_{12}^c \cap F_{23}^c \cap (X_1,X_3)\inv(\Delta_{\mathbf A}^c\cap (A\times \mathbf{A}))$,
              an event in $\calF$ by assumption.
              Finally, to show the third component
                  $(F_{12} \cap F_{23}^c \cap X_1\inv(A))$
              is in $\calF$ note that we have the equality
                  $(F_{12} \cap F_{23}^c \cap X_1\inv(A))
                  =(F_{12} \cap F_{23}^c \cap X_2\inv(A))$
              corresponding to the following equivalence of events:
              \[ X_1 = X_2 \ne X_3 \land X_1 \in A
                \iff X_1 = X_2\ne X_3 \land X_2\in A \]
              Thus the third component is an intersection of events
              $F_{12}$ and $F_{23}^c \cap X_2\inv(A)$ that are in $\calF$
              by assumption.
            \item $F_{13}^c \cap X_3\inv(A)\in\calF$:
              this case is symmetrical to the one above, with $X_3$ replaced by $X_1$.
              The strategy is the same:
              note that
              $F_{13}^c = (F_{12}^c \cap F_{23}) \uplus (F_{12}^c \cap F_{23}^c \cap F_{13}^c) \uplus (F_{12} \cap F_{23}^c)$,
              so the intersection $F_{13}^c \cap X_3\inv(A)$ can be rewritten as:
              \[
                  (F_{12}^c \cap F_{23} \cap X_3\inv(A)) \uplus
                  (F_{12}^c \cap F_{23}^c \cap F_{13}^c \cap X_3\inv(A)) \uplus
                  (F_{12} \cap F_{23}^c \cap X_3\inv(A))
              \]
              The first component $F_{12}^c\cap F_{23}\cap X_3^{-1}(A)$ is
              equivalent to $F_{12}^c\cap F_{23}\cap X_2^{-1}(A)$,
              due to the following equivalence of events:
              \[ X_1 \ne X_2 = X_3 \land X_3 \in A
                \iff X_1 \ne X_2 = X_3 \land X_2\in A \]
              Thus the first component is an intersection of events
              $F_{23}$ and $F_{12}^c \cap X_2^{-1}(A)$ that are in $\calF$ by assumption.
              The second component
              is equal to
              $F_{12}^c \cap F_{23}^c \cap (X_1,X_3)\inv(\Delta_{\mathbf A}^c\cap (\mathbf{A}\times A))$,
              an event in $\calF$ by assumption.
              Finally, the third component $F_{12}\cap F_{23}^c\cap X_3^{-1}(A)$ is
              an intersection of two events
              $F_{12}$ and $F_{23}^c \cap X_3^{-1}(A)$ that are in $\calF$ by assumption.
          \end{itemize}
      \end{itemize}
  \end{itemize}
\end{proof}

\begin{lemma} \label{lem:full-generator-is-deterministic}
  Let $(\Omega,\calF,\mu)$ be a probability space and $G$ a collection of
  full sets. (A full set is an event with probability $1$.)
  For all $E\in\angled{G}$ it holds that $\mu(E) = 0$
  or $\mu(E) = 1$.
\end{lemma}
\begin{proof}
  Without loss of generality we may assume that $G$ is a $\pi$-system:
  if $G$ is empty then $\angled{G}$ is the trivial $\sigma$-algebra and we are done;
  if $G$ is nonempty, it generates the same $\sigma$-algebra as its
  closure under finite intersections, and finite intersections of full sets
  remain full.
  By the $\pi$-$\lambda$ theorem, we are done if we can show that the collection
  of events $E$ for which $\mu(E) = 0$ or $\mu(E) = 1$ is a $\lambda$-system
  that contains $G$.
  \begin{itemize}
    \item Contains $G$: if $E\in G$ then $\mu(E) = 1$ by assumption.
    \item Contains $\emptyset$: $\mu(\emptyset) = 0$ because $\mu$ is a measure.
    \item Closed under complements: if $\mu(E) \in \{0,1\}$ then $\mu(E^c) = 1 - \mu(E) \in \{0,1\}$.
    \item Closed under countable disjoint unions: let $\{E_i\}_i$ be a pairwise-disjoint
      countable family of sets for which $\mu(E_i) \in \{0,1\}$
      for all $i$. By countable additivity of measures,
      $\mu(\biguplus_i E_i) = \sum_i \mu(E_i) \in \mathbb N $.
      This combined with the fact that $\mu$ is a probability measure forces
      $\mu(\biguplus_i E_i) \in \{0,1\}$ as required.
  \end{itemize}
\end{proof}

\begin{lemma} \label{lem:asequal-good-conjunct}
  Let $(\mathbf{A}, \calA)$ be a measurable space that supports equality,
  and let $E_1$ and $E_2$ be random expressions of type $\mathbf A$.
  The following double-entailment holds:
  \[ P \land (E_1\asequal E_2) \dashv\vdash P * (E_1\asequal E_2) \]
\end{lemma}
\begin{proof}
  The right-to-left entailment follows from the fact that our separation logic is
  affine. For the left-to-right entailment, fix a configuration $\gamma,D,(\calF,\mu)$,
  let $X_i = E_i(\gamma)\circ D$ for $i\in\{1,2\}$, let $F_{12}$ be the event $X_1 = X_2$,
  and suppose that
  \begin{enumerate}
    \item $\gamma,D,(\calF,\mu)\vDash P$
    \item $F_{12}\in\calF$ and $\mu(F) = 1$
    \item $F_{12}\cup (X_1,X_2)^{-1}(\calA\otimes\calA)\subseteq\calF$
  \end{enumerate}
  with the aim of showing $\gamma,D,(\calF,\mu)\vDash P * (E_1\asequal E_2)$.
  Let $\calG$ be the sub-$\sigma$-algebra of $\calF$
  generated by $F_{12}\cup(X_1,X_2)^{-1}(\calA\otimes\calA)$,
  and let $\nu$ be the restriction of $\calF$ to $\calG$.
  The probability space $(\calG,\nu)$ witnesses $E_1\asequal E_2$ by (2) and (3).
  By (1), we are done if we can show $(\calF,\mu)\pdot(\calG,\nu) = (\calF,\mu)$.
  By definition of independent combination, it suffices to show
  $\mu(F\cap G) = \mu(F)\nu(G)$ for all $F\in\calF$ and $G\in\calG$.
  Since $\nu$ is defined as a restriction of $\mu$ to $\calG$, this reduces
  to showing $\mu(F\cap G) = \mu(F)\mu(G)$.
  The $\sigma$-algebra $\calG$ is generated by events $F_{12}\cup (X_1,X_2)^{-1}(\calA\otimes\calA)$
  that are all full, so by Lemma~\ref{lem:full-generator-is-deterministic}
  we have that $\mu(G)\in\{0,1\}$ for all $G\in\calG$. There are thus two cases:
  \begin{itemize}
    \item If $\mu(G) = 0$, then $\mu(F\cap G) = 0 = \mu(F)\mu(G)$ as required.
    \item If $\mu(G) = 1$, then $\mu(F\cap G) = \mu(F) - \mu(F\setminus G) = \mu(F) - 0 = \mu(F)\mu(G)$ as required.
  \end{itemize}
\end{proof}

\begin{corollary}[almost-sure equality is duplicable]
  Let $(\mathbf{A}, \calA)$ be a measurable space that supports equality,
  and let $E_1$ and $E_2$ be random expressions of type $\mathbf A$.
  The following entailment holds:
  \[ E_1\asequal E_2 \vdash (E_1\asequal E_2) * (E_1\asequal E_2) \]
\end{corollary}
\begin{proof}
  $E_1\asequal E_2 \vdash (E_1\asequal E_2)\land (E_1\asequal E_2) \stackrel{\ref{lem:asequal-good-conjunct}}\vdash (E_1\asequal E_2) * (E_1\asequal E_2)$.
\end{proof}

\begin{lemma}[transfer of ownership]
  Let $(\mathbf{A}, \calA)$ be a measurable space that supports equality.
  Let $E_1$ and $E_2$ be random expressions of type $\mathbf A$.
  The following entailments hold:
  \begin{mathpar}
    \inferrule*[lab=Transfer-Own]{}{\own E_1 \land (E_1\asequal E_2) \vdash \own E_2}
    \and
    \inferrule*[lab=Transfer-Dist]{}{(E_1 \sim\nu)\land (E_1\asequal E_2) \vdash E_2\sim\nu}
  \end{mathpar}
\end{lemma}
\begin{proof}
  We prove \textsc{Transfer-Dist}; the proof of \textsc{Transfer-Own} is identical.
  Fix a configuration $(\gamma,D,(\calF,\mu))$,
  let $X_i = E_i(\gamma)\circ D$ for $i\in\{1,2\}$,
  let $F_{12}$ be the event $X_1=X_2$,
  and suppose
  \begin{enumerate}
    \item $X_1$ is $\calF$-measurable with distribution $\nu(\gamma)$
    \item $F_{12}\in\calF$ and $\nu(F_{12}) = 1$
    \item $F_{12} \cup (X_1,X_2)\inv(\calA\otimes\calA)\subseteq\calF$
  \end{enumerate}
  with the aim of showing $X_2$ is $\calF$-measurable with distribution $\nu(\gamma)$.
  It suffices to show that $X_2$ is $\calF$-measurable, as then it follows that
  $X_1$ and $X_2$ are almost-surely equal random variables with respect to $\mu$
  and so have the same distribution. 
  Fix arbitrary $A\in\calA$ with the aim of showing $X_2^{-1}(A) \in \calF$.
  Write $X_2^{-1}(A)$ as the disjoint union
  $(X_2^{-1}(A) \cap F_{12}) \uplus (X_2^{-1}(A) \cap F_{12}^c)$.
  The first disjunct is equal to $X_1\inv(A)\cap F_{12}$
  because $F_{12}$ is the event $X_1=X_2$; this is in $\calF$ because
  $X_1\inv(A)\in\calF$ by (1) and $F_{12}\in\calF$ by (2).
  The second disjunct is in $\calF$ by (3) and Lemma~\ref{lem:asequal-tfae}.
  Thus $X_2^{-1}(A)$ is a union of events in $\calF$ as required.
\end{proof}

\begin{lemma}[congruence]
  Let $(\mathbf{A}, \calA)$ and $(\mathbf B,\calB)$ be measurable spaces that support equality.
  Let $E_1$ and $E_2$ be random expressions of type $\mathbf A$.
  Let $F[X]$ be a random expression of type $\mathbf B$ with a free variable $X$ of type $\mathbf A$.
  The following entailment holds:
  \begin{mathpar}
    \inferrule*[lab=Congruence]{}{\own (F[E_1],F[E_2]) \land(E_1\asequal E_2) \vdash F[E_1] \asequal F[E_2]}
  \end{mathpar}
\end{lemma}
\begin{proof}
  Fix a configuration $(\gamma,D,(\calF,\mu))$.
  Let $X_i = E_i(\gamma)\circ D$ for $i\in\{1,2,3\}$
  and let $E_{12}$ be the event $X_1=X_2$.
  Let $F_1 = F(\gamma)\circ (D,X_1)$
  and $F_2 = F(\gamma)\circ (D,X_2)$,
  and let $F_{12}$ be the event $F_1=F_2$.
  We have
  \begin{itemize}
    \item $F_1$ and $F_2$ are $\calF$-measurable
    \item $E_{12}\in\calF$ and $\mu(E_{12}) = 1$
    \item $E_{12} \cup (X_1,X_2)\inv(\calA\otimes\calA)\subseteq\calF$
  \end{itemize}
  The goal is to show $F_{12}\in\calF$ and $\mu(F_{12}) = 1$ and
  $F_{12}\cup (F_1,F_2)\inv(\calB\otimes\calB)\subseteq\calF$.
  \begin{itemize}
    \item 
       $F_{12} = (F_1,F_2)\inv(\Delta_{\mathbf B})\in\calF$ because $F_1$ and
       $F_2$ are $\calF$-measurable and $\mathbf B$ supports equality.
    \item $\mu(F_{12}) \ge \mu(E_{12}) = 1$ because $E_{12}\subseteq F_{12}$.
    \item
      $\calF$ contains $F_{12}$ and the entire pullback $\sigma$-algebra
      $(F_1,F_2)\inv(\calB\otimes\calB)$, and so contains
      $F_{12} \cup (F_1,F_2)\inv(\calB\otimes\calB)$ too.
  \end{itemize}
\end{proof}

\begin{lemma}[($\asequal$) as derived notion]
  Let $\mathbf A$ be a measurable space that supports equality and
  let $E_1$ and $E_2$ be expressions of type $A$.
  Let $\bf (A\otimes A)_\bot$ be the measurable space with underlying set
  $\bf (A\otimes A) \cup \{\bot\}$ and $\sigma$-algebra generated by
  measurable subsets of $\bf A\otimes A$ and the singleton set $\{\bot\}$.
  The following equivalence holds:
  \[ 
    E_1 \asequal E_2 \dashv\vdash \own(\ite{E_1=E_2}{\bot}{(E_1,E_2)})\land (\Ex[\ind[E_1=E_2]]=1)
    \]
\end{lemma}
\begin{proof}
  Fix configuration $\gamma,D,(\calF,\mu)$. Let $X_i = E_i(\gamma)\circ D$ for $i\in\{1,2\}$
  and let $F$ be the event $X_1=X_2$.
  Let $Y$ be the random variable
  $\ite{X_1=X_2}{\bot}{(E_1,E_2)}$.
  The left-hand side asserts $F$ has probability $1$ and $F\cup (X_1,X_2)\inv(A)\in\calF$
  for all $A\in \bf A\otimes A$.
  The right-hand side asserts $F$ has probability $1$ and
  that $Y$ is $\calF$-measurable.
  To show the equivalence of these two assertions, it suffices to show that
  measurability of $Y$ is equivalent to having $F\cup (X_1,X_2)\inv(A)\in\calF$
  for all $A\in \bf A\otimes A$.
  First suppose $Y$ is measurable and fix arbitrary $A\in\bf A\otimes A$.
  Then measurability of $Y$ says \begin{align*}
  Y\inv(A) &= (F\cap Y\inv(A))\uplus (F^c\cap Y\inv(A))
   = (F\cap \bot\inv(A)) \uplus (F^c \cap (X_1,X_2)\inv(A))\\
   &= (F\cap \emptyset) \uplus (F^c \cap (X_1,X_2)\inv(A))
   = F^c \cap (X_1,X_2)\inv(A)
   \in\calF\end{align*}
   as required.
   Conversely suppose $\calF$ contains $F^c\cap (X_1,X_2)\inv(A)\in\calF$ for all $A\in\bf A\otimes A$
   and fix arbitrary $A\in \bf (A\otimes A)_\bot$. Then
  \begin{align*}
  Y\inv(A) &= (F\cap Y\inv(A))\uplus (F^c\cap Y\inv(A))
   = (F\cap \bot\inv(A)) \uplus (F^c \cap (X_1,X_2)\inv(A))\\
   &= (F\cap [\bot\in A]) \uplus (F^c \cap (X_1,X_2)\inv(A\setminus\{\bot\}))
   \end{align*}
   where $[\bot\in A]$ is $\bf A$ if $\bot\in A$ and $\emptyset$ otherwise.
   There are two cases.
   If $\bot\notin A$ then $Y\inv(A) = F^c\cap (X_1,X_2)\inv(A)\in \calF$ by assumption.
   If $\bot\in A$ then $Y\inv(A) = F \cup (F^c\cap (X_1,X_2)\inv(A\setminus\{\bot\}))$,
   a finite union of elements in $\calF$ by assumption.
\end{proof}

\subsection{Derived rules} \label{app:wp-law-section}

\begin{lemma} \label{app:structural-rules}
  The following structural rules hold:
  \begin{mathpar}
    \inferrule*[lab=H-Consequence]{P\vdash Q}{\weakpre(M,X.P)\vdash \weakpre(M,X.Q)}
    \and
    \inferrule*[lab=H-Frame,right=$(X\notin F)$]{}{F*\weakpre(M,X.Q)\vdash \weakpre(M,X.F * Q)}
    \and
    \inferrule*[lab=H-Disjunction]{}{\weakpre(M,X.P)\lor \weakpre(M,X.Q)\vdash \weakpre(M,X.P\lor Q)}
  \end{mathpar}
\end{lemma}
\begin{proof}
  We show the proof of the frame rule; the others are standard.
  Suppose (1)
  $\gamma,D,\calP_F\pdot \calP_M\vDash F*\weakpre(M,X.Q)$
  for some $\gamma,D,\calP_F\vDash F$ and $\gamma,D,\calP_M\vDash \weakpre(M,X.Q)$.
  To show $\weakpre(M,X.F*Q)$, further suppose
  $\calP_\textrm{frame}\pdot(\calP_F\pdot\calP_M)\plte(\Sigma_\Omega,\mu)$
  and $D_\textrm{ext}:\rv{\sembr{\Delta_\textrm{ext}}}$.
  By associativity, $\calP_\textrm{frame}\pdot(\calP_F\pdot\calP_M)
  =(\calP_\textrm{frame}\pdot\calP_F)\pdot\calP_M$
  so specializing (1) with $\calP_\textrm{frame}:=\calP_\textrm{frame}\pdot\calP_F$
  gives $(\calP_\textrm{frame}\pdot\calP_F)\pdot\calP'\plte(\Sigma_\Omega,\mu')$
  and $X$ such that
  \[
    \monadic{
      &\omega\gets \mu;\\
      &v\gets M(\gamma)(D(\omega));\\
      &\ret (D_\textrm{ext}(\omega), D(\omega), v)
    } = \monadic{
      &\omega\gets \mu';\\
      &\ret (D_\textrm{ext}(\omega), D(\omega), X(\omega))
    }
  \]
  and $\gamma,(D,X),\calP'\vDash Q$.
  Since $X\notin F$, $F[\mathrm{weak}_X] = F$
  so $\gamma,(D,X),\calP_F\vDash F$
  by lemma \ref{subst-lemma}.
  And since the composition $(\calP_\textrm{frame}\pdot\calP_F)\pdot\calP'$ is defined,
  the composition $\calP_F\pdot\calP'$ must be as well, so
  $\gamma,(D,X),\calP_F\pdot\calP'\vDash F*Q$ as desired.
\end{proof}

\begin{lemma} \label{app:general-wp-laws}
  The following $\weakpre$ laws hold:
  \begin{mathpar}
    \inferrule*{}
    {Q[e/X]\vdash \weakpre(\ret e,X.Q)}
    \and
    \inferrule*{}
    {\weakpre(M,X.\weakpre(N,Y.Q)) \vdash
    \weakpre\left((X \gets M; N),Y.Q\right)}
    \\
    \inferrule*
    {}{(\forallrv X\ofty A.~ X\sim \operatorname{Unif}[0,1]\wand Q) \vdash \weakpre(\unif~[0,1],X\ofty A.Q)}
    \\
    \inferrule*
      {}{(\forallrv X\ofty A.~ X\sim \ber p\wand Q) \vdash \weakpre(\ber p,X\ofty A.Q)}
    \\
    \inferrule*
    {}{I(1,e) * (\forall i\ofty\N.~\forallrv X\ofty A.~\triple{I(i,X)}M{X'}{I(i+1,X')}) 
      \vdash \weakpre(\forloop neM,X\ofty A.~I(n+1,X))}
    \\
    \inferrule*
    {}{\weakpre(M,X\ofty A.~\weakpre(N,Y\ofty A.~ Q(\ite EXY)))
      \vdash \weakpre(\ite EMN,X\ofty A.~Q)}
  \end{mathpar}
  where $\forloop nef$ is defined by
  \begin{align*}
    &{\forloop nef} = \mathrm{loop}(1,e,f)
      \textrm{where } \mathrm{loop}(k,e,f) = \begin{cases}
      \ret e, & k>n \\
      v\gets f(k,e);~
      \mathrm{loop}(k+1,v,f), & \textrm{otherwise}
    \end{cases}
  \end{align*}
\end{lemma}
\begin{proof}~
\begin{itemize}
\item Ret:
  suppose $\gamma,D,\calP\vDash Q[e(\gamma)/X]$.
  By lemma \ref{subst-lemma} this is equivalent to $\gamma,(D,e(\gamma)),\calP\vDash Q$.
  To show $\weakpre(\ret e,X.Q)$ suppose
  $\calP_\textrm{frame}\pdot\calP\plte(\Sigma_\Omega,\mu)$
  and $D_\textrm{ext}:\rv\sembr{\Delta_\textrm{ext}}$.
  Choose $\calP':=\calP$ and $\mu':=\mu$ and $X(\omega):=e(\gamma)$.
  Then
  \[
    \monadic{
      &\omega\gets \mu;\\
      &v\gets (\ret e)(\gamma)(D(\omega));\\
      &\ret (D_\textrm{ext}(\omega), D(\omega), v)
    }
    = \monadic{
      &\omega\gets \mu;\\
      &\ret (D_\textrm{ext}(\omega), D(\omega), e(\gamma))
    }
    = \monadic{
      &\omega\gets \mu';\\
      &\ret (D_\textrm{ext}(\omega), D(\omega), X(\omega))
    }
  \]
  and $\gamma,(D,e(\Gamma)),\calP\vDash Q$ as desired.
\item Let:
  suppose $\gamma,D,\calP\vDash \weakpre(M,X.\weakpre(N,Y.Q))$.
  To show $\weakpre((X\gets M; N),Y.Q)$ suppose
  $\calP_\textrm{frame}\pdot\calP\plte(\Sigma_\Omega,\mu)$
  and $D_\textrm{ext}:\rv\sembr{\Delta_\textrm{ext}}$.
  By assumption, there exist $\calP_\textrm{frame}\pdot\calP_X\plte(\Sigma_\Omega,\mu_X)$
  and $X$ such that
  \begin{align} \label{let-binding-step1}
     \monadic{
       &\omega\gets \mu;\\
       &x\gets M(\gamma)(D(\omega));\\
       &\ret (D_\textrm{ext}(\omega), D(\omega), x)
     } = \monadic{
       &\omega\gets \mu_X;\\
       &\ret (D_\textrm{ext}(\omega), D(\omega), X(\omega))
     }
  \end{align}
  and $\gamma,(D,X),\calP_X\vDash \weakpre(N,Y.Q)$.
  Applying this assumption gives $\calP_\textrm{frame}\pdot\calP_Y\plte(\Sigma_\Omega,\mu_Y)$
  and $Y$ with
  \begin{align} \label{let-binding-step2}
     \monadic{
       &\omega\gets \mu_X;\\
       &y\gets N(\gamma)((D,X)(\omega));\\
       &\ret (D_\textrm{ext}(\omega), (D,X)(\omega), y)
     } = \monadic{
       &\omega\gets \mu_Y;\\
       &\ret (D_\textrm{ext}(\omega), (D,X)(\omega), Y(\omega))
     }
  \end{align}
  and $\gamma,(D,X,Y),\calP_Y\vDash Q$.
  Since $X\notin Q$, this implies $\gamma,(D,Y),\calP_Y\vDash Q$ by \ref{subst-lemma}, so
  it only remains to show 
  \[
     \monadic{
       &\omega\gets \mu;\\
       &y\gets (X\gets M; N)(\gamma)(D(\omega));\\
       &\ret (D_\textrm{ext}(\omega), D(\omega), y)
     } = \monadic{
       &\omega\gets \mu_Y;\\
       &\ret (D_\textrm{ext}(\omega), D(\omega), Y(\omega))
     }
  \]
  Calculate:
  \begin{align*}
     &\monadic{
       &\omega\gets \mu;\\
       &y\gets (X\gets M; N)(\gamma)(D(\omega));\\
       &\ret (D_\textrm{ext}(\omega), D(\omega), v)
     }
     = \monadic{
       &\omega\gets \mu;\\
       &x\gets M(\gamma)(D(\omega)); \\
       &y\gets N(\gamma)(D(\omega), x);\\
       &\ret (D_\textrm{ext}(\omega), D(\omega), y)
     }\\
     &= \monadic{
       &(\delta_\textrm{ext},\delta,x)\gets\\
       &\hspace{1em}\monadic{
           &\omega\gets \mu;\\
           &x\gets M(\gamma)(D(\omega)); \\
           &\ret (D_\textrm{ext}(\omega),D(\omega),x) }\\
       &y\gets N(\gamma)(\delta, x);\\
       &\ret (\delta_\textrm{ext}, \delta, y)
     }
    \stackrel{\ref{let-binding-step1}}= \monadic{
       &(\delta_\textrm{ext},\delta,x)\gets\\
       &\hspace{1em}\monadic{
           &\omega\gets \mu_X;\\
           &\ret (D_\textrm{ext}(\omega),D(\omega),X(\omega)) }\\
       &y\gets N(\gamma)(\delta, x);\\
       &\ret (\delta_\textrm{ext}, \delta, y)
     }\\
    &= \monadic{
        &\omega\gets \mu_X;\\
        &y\gets N(\gamma)(\delta, X(\omega));\\
        &\ret (D_\textrm{ext}(\omega),D(\omega),y) }
    \stackrel{\ref{let-binding-step2}}= \monadic{
        &\omega\gets \mu_Y;\\
        &\ret (D_\textrm{ext}(\omega),D(\omega),Y(\omega)) 
      }
  \end{align*}
\item Uniform: 
  suppose (1) $\gamma,D,\calP\vDash \forallrv X\ofty A.~ X\sim \operatorname{Unif}[0,1] \wand Q$.
  To show $\weakpre(\unif~[0,1],X\ofty A.Q)$, suppose
  $\calP_\textrm{frame}\pdot\calP\plte(\Sigma_\Omega,\mu)$
  and $D_\textrm{ext}:\rv\sembr{\Delta_\textrm{ext}}$.
  Let $n$ witness $(D_\textrm{ext},\calP_\textrm{frame}\pdot\calP)$'s finite footprint.
  Write the Hilbert cube as $[0,1]^\N\cong [0,1]^n \otimes [0,1]^\N$.
  Define $\mu'$ via this isomorphism as the product measure
  $\mu|_{[0,1]^n}\otimes\lambda$, where $\lambda$ assigns
  to each finite-dimensional box $\prod_{i=1}^n [a_i,b_i]\times[0,1]^\N$
  the measure $\prod_{i=1}^n|b_i-a_i|$ and extends to a measure on the whole Hilbert cube
  by the Carath{\'e}odory extension theorem.
  Let $\calP_n$ be the restriction of $\mu'$ to measurable sets of the form
  $[0,1]^n \times F\times [0,1]^\N$.
  Let $X$ be the projection $\pi_{n+1} = (\dots,\omega_{n+1},\dots)\mapsto\omega_{n+1}$.
  By construction, the composite $\calP_\textrm{frame}\pdot\calP\pdot\calP_n$ is defined
  and $\calP_\textrm{frame}\pdot\calP\pdot\calP_n\plte(\Sigma_\Omega,\mu')$
  and $X$ is $\calP_n$-measurable and uniformly distributed in $[0,1]$.
  Therefore $\gamma,(D,X),\calP\pdot\calP_n\vDash Q$ by (1), and it only remains to show
  \[
    \monadic{
      &\omega\gets \mu;\\
      &v\gets \unif~[0,1];\\
      &\ret (D_\textrm{ext}(\omega), D(\omega), v)
    } = \monadic{
      &\omega\gets \mu';\\
      &\ret (D_\textrm{ext}(\omega), D(\omega), X(\omega))
    }
  \]
  Calculate:
  \begin{align*}
    &\monadic{
      &\omega\gets \mu;\\
      &v\gets \unif~[0,1];\\
      &\ret (D_\textrm{ext}(\omega), D(\omega), v)
    }
    = \monadic{
      &\omega_{1\cdots n}\gets \mu|_{[0,1]^n};\\
      &v\gets \unif~[0,1];\\
      &\ret (D_\textrm{ext}(\omega_{1\cdots n}), D(\omega_{1\cdots n}), v)
    }\\
    &= \monadic{
      &(\omega_{1\cdots n},v)\gets \mu|_{[0,1]^n} \otimes \unif~[0,1];\\
      &\ret (D_\textrm{ext}(\omega_{1\cdots n}), D(\omega_{1\cdots n}), v)
    }
    = \monadic{
      &\omega_{1\cdots n+1}\gets \mu'|_{[0,1]^{n+1}};\\
      &\ret (D_\textrm{ext}(\omega_{1\cdots n}), D(\omega_{1\cdots n}), \omega_{n+1})
    }\\
    &= \monadic{
      &\omega\gets \mu';\\
      &\ret (D_\textrm{ext}(\omega), D(\omega), X(\omega))
    }
  \end{align*}
\item Flip: analogous to Uniform.
\item For:
  suppose (1) $\forall i\ofty\N.~\forallrv X\ofty\rv A.~\triple{I(i,X)}M{X'}{I(i+1,X')}$.
  We need to show \[I(1,e)\wand \weakpre(\forloop neM,X\ofty A.~I(n+1,X)).\]
  We generalize, and show \[I(n+1-k,V) \wand \weakpre(\mathrm{loop}(n+1-k,V,M),X'.~I(n+1,X'))\]
  for all $V$ and all $0\le k\le n$ by induction on $k$, from which this follows at $k=n$.
  \begin{itemize}
    \item Case $k=0$:
      \begin{align*}
        \top
        &\vdash
        I(n+1,V) \wand I(n+1,V)\\
        &\vdash
        I(n+1,V) \wand \weakpre(\ret V,X'.~I(n+1,X'))\\
        &\vdash
        I(n+1,V) \wand \weakpre(\mathrm{loop}(n+1,V,M),X'.~I(n+1,X'))\\
        &\vdash
        I(n+1-k,V) \wand \weakpre(\mathrm{loop}(n+1-k,V,M),X'.~I(n+1,X'))
      \end{align*}
    \item Case $k=j+1\le n$: 
      backwards reasoning from the goal gives
      \begin{align*}
        &I(n+1-(j+1),V) \wand \weakpre(\mathrm{loop}(n+1-(j+1),V,M),X'.~I(n+1,X'))\\
        &\!\!\dashv I(n-j,V) \wand \weakpre(\mathrm{loop}(n-j,V,M),X'.~I(n+1,X'))\\
        &\!\!\dashv I(n-j,V) \wand \weakpre((V'\gets M[n-j/i,V/X];~ \mathrm{loop}(n+1-j,V',M)),X'.~I(n+1,X'))\\
        &\!\!\dashv I(n-j,V) \wand \weakpre(M[n-j/i,V/X],V'.~\weakpre(\mathrm{loop}(n+1-j,V',M),X'.~I(n+1,X')))
      \end{align*}
      Now $I(n-j,V)\wand \weakpre(M[n-j/i,V/X],V'.~I(n-j+1,V'))$ by (1)
      so it suffices to show
      \begin{align*}
        &\weakpre(M[n-j/i,V/X],V'.~ I(n+1-j,V')) \\
        \wand &\weakpre(M[n-j/i,V/X],V'.~\weakpre(\mathrm{loop}(n+1-j,V',M),X'.~I(n+1,X'))).
      \end{align*}
      The outer $\weakpre$s are the same, so by the consequence rule it suffices to show
      \begin{align*}
        I(n+1-j,V')
        \wand \weakpre(\mathrm{loop}(n+1-j,V',M),X'.~I(n+1,X'))
      \end{align*}
      for all $V'$, which is exactly the induction hypothesis at $j$.
  \end{itemize}
\item If: applying properties of Markov kernels and rules for Let and Ret,
  \begin{align*}
    &\weakpre(\ite EMN, Z.~Q(Z))\\
    &\dashv \weakpre\left(\monadic{
      &X\gets M; \\
      &Y\gets N; \\
      &\ret (\ite EXY)
    }, Z.~Q(Z)\right) \\
    &\dashv \weakpre\left(M, X.~\weakpre\left(\monadic{
      &Y\gets N; \\
      &\ret (\ite EXY)
    }, Z.~Q(Z)\right)\right) \\
    &\dashv \weakpre(M, X.~\weakpre(N,Y.~
      \weakpre(\ret (\ite EXY), Z.~Q(Z)))) \\
    &\dashv \weakpre(M, X.~\weakpre(N,Y.~ Q(\ite EXY)))
  \end{align*}
  as desired.
\end{itemize}
\end{proof}

\begin{corollary} \label{app:wp-laws}
  The following $\weakpre$ laws hold:
  \begin{mathpar}
    \inferrule*[lab=W-Ret]{}
    {Q[\sembr M/X]\vdash \weakpre(\sembr{\applret M},X.Q)}
    \and
    \inferrule*[lab=W-Let]{}
    {\weakpre(\sembr M,X.\weakpre(\sembr N,Y.Q)) \vdash
    \weakpre\left(\sembr{\letin XMN},Y.Q\right)}
    \and
    \inferrule*[lab=W-Uniform]
      {}{(\forallrv X\ofty A.~ X\sim \operatorname{Unif}[0,1]\wand Q) \vdash \weakpre(\sembr\applunif,X\ofty A.Q)}
    \and
    \inferrule*[lab=W-Flip]
      {}{(\forallrv X\ofty A.~ X\sim \ber p\wand Q) \vdash \weakpre(\sembr{\flip~p},X\ofty A.Q)}
    \and
    \inferrule*[lab=W-For]
    {}{I(1,e) * (\forall i\ofty\N.~\forallrv X\ofty A.~\triple{I(i,X)}M{X'}{I(i+1,X')}) 
      \vdash \weakpre(\sembr{\applfor neiXM},X\ofty A.~I(n+1,X))}
    \\
    \inferrule*
    {}{\weakpre(\sembr M,X\ofty A.~\weakpre(\sembr N,Y\ofty A.~ Q(\ite EXY)))
      \vdash \weakpre(\sembr{\applif EMN},X\ofty A.~Q)}
  \end{mathpar}
\end{corollary}
\begin{proof} Unfold $\sembr-$ and apply lemma \ref{app:general-wp-laws}. \end{proof}

\begin{lemma} \label{app:proof-rules}
  The following proof rules hold:
  \begin{mathpar}
    \inferrule*[right=H-Consequence]{P\vdash P' \\ Q'\vdash Q \\ \triple{P'}MX{Q'}}{\triple PMXQ}
    \and
    \inferrule*[right=H-Frame $(X\notin F)$]{\triple PMXQ}{\triple{F*P}MX{F*Q}}
    \and
    \inferrule*[right=H-Disjunction]{\triple PMXQ \\ \triple{P'}MX{Q'}}{\triple{P\lor P'}MX{Q\lor Q'}}
    \and
    \inferrule*[right=H-Ret]{~}{\triple{Q\big[\sembr M/X\big]}{\applret M}XQ}
    \and
    \inferrule*[right=H-Let]
      {\triple PMXQ \\ \forallrv X.~\triple QNYR}
      {\triple P{\letin XMN}YR}
    \and
    \inferrule*[right=H-Uniform]{~}{\triple\top{\applunif}X{X\sim\unif~[0,1]})}
    \and
    \inferrule*[right=H-Flip]{~}{\triple\top{\flip~p}X{X\sim\ber p}}
    \and
    \inferrule*[right=H-For]
      {\forall i\ofty\N.~\forallrv X\ofty A.~\triple{I(i,X)}M{X'}{I(i+1,X')}}
      {\triple{I(1,e)}{\applfor neiXM}{X\ofty A}{I(n+1,X)}}
    \and
    \inferrule*[right=H-If]
      {\triple P M X {Q(X)}
       \\\\
       \forallrv X.~\triple {Q(X)} N Y {R(\ite EXY)}
       }
      {\triple{P}{\applite EMN}{Z}{R(Z)}}
  \end{mathpar}
\end{lemma}
\begin{proof}
  Rules \textsc{H-Consequence}, \textsc{H-Frame}, and \textsc{H-Disjunction}
  follow from Lemma~\ref{app:structural-rules};
  the remaining rules follow from Corollary~\ref{app:wp-laws}.
  All proofs go by unfolding the definition of the Hoare triple
  and applying the relevant $\weakpre$ law.
\end{proof}

\subsection{Disintegration}
\label{app:disintegration-proofs}

\begin{lemma} \label{app:restricted-krm}
  Let $(M,\pdot,\plte,1)$ be a KRM with $1\plte x$ for all $x$.
  Let $A$ be a downward-closed subset of $M$ (i.e., $x\plte y\in A$ implies $x\in A$).
  Let $(\pdot')$ be the restriction of $(\pdot)$ to $A$; that is,
  \begin{align*}
    x \pdotp y := \begin{cases}
      x\pdot y, & x\pdot y\in A\\
      \mathrm{undefined}, & \textrm{otherwise}
    \end{cases}
  \end{align*}
  Then $(A,\pdot',\plte,1)$ is a KRM.
\end{lemma}
\begin{proof}
  Unit and commutativity are straightforward.
  For associativity, note that if $1\plte x$ for all $x$
  then by monotonicity of $(\pdot)$ we have $x\plte x\pdot y$ for all $x,y$.
  Now suppose $x\pdotp y$ defined and $(x\pdotp y)\pdotp z$ defined.
  Then $(x\pdot y)\pdot z \in A$ and by downward closure so is $y\pdot z$
  and by associativity so is $x\pdot (y\pdot z)$, so both $y\pdotp z$ and $x\pdotp (y\pdotp z)$
  are defined and associativity is inherited from associativity of $(\pdot)$.
\end{proof}

\begin{thm} \label{app:final-krm}
  Let $\Mdis$ be the set of countably-generated probability spaces $\calP$ that have finite footprint
  and can be extended to a Borel measure on the entire Hilbert cube.
  The restriction of the KRM given by Theorem~\ref{thm:spaces-form-a-krm} to $\Mdis$ is still a KRM.
\end{thm}
\begin{proof}
  By Lemma~\ref{app:lem:finite-footprint-remains-krm},
  the restriction to $\Mfin$ is a KRM, so it suffices to show that restricting to
  countably-generated spaces that can be extended to a Borel measure still yields a KRM.
  First, restricting to countably-generated spaces still yields a KRM because
  the independent combination of two countably-generated spaces remains couuntably-generated.
  Then, the restriction to spaces that can be extended to a Borel measure
  still yields a KRM by Lemma~\ref{app:restricted-krm},
  as the set of spaces that can be extended to a Borel measure is downward-closed.
\end{proof}

\begin{mathpar}
  \inferrule{\Gamma;\Delta\vdashrv E:A \\ \Gamma,x\ofty A;\Delta\vdash P}
    {\Gamma;\Delta\vdash \D_{x\ofty A\gets E} P}
\end{mathpar}
\begin{center}
\begin{tabular}{rcl}
  $\gamma,D,\calP\vDash \D_{x\ofty A\gets E} P$ & iff & for all $(\Sigma_\Omega,\mu)\sqsupseteq \calP$\\
                                                &     & and all $(\Sigma_\Omega,\mu)$-disintegrations
                                                        $\{\nu_x\}_{x\in A}$ with respect to $E(\gamma)\circ D$ \\
                                         &     & and $(E\circ D)_*\mu$-almost-all $x\in A$, \\
                                         &     & $(\gamma,x),D,\nu_x|_\calP\vDash P$
\end{tabular}
\end{center}

\begin{lemma} \label{distributions-agree-pi-generated}
  If $\mu$ and $\nu$ are probability measures on a space $(\Omega,\calF)$
  generated by a $\pi$-system $\calB$, then
  $\mu=\nu$ iff $\mu(B)=\nu(B)$ for all $B\in\calB$.
\end{lemma}
\begin{proof}
  The left-to-right direction is straightforward.
  The right-to-left direction follows from the $\pi$-$\lambda$ theorem:
  the set $S:=\{B\in\calB\mid \mu(B)=\nu(B)\}$
  is a $\lambda$-system
  and $\mu$ and $\nu$ agree on a $\pi$-system that generates $\calF$ by assumption.
\end{proof}

\begin{lemma} \label{app:C-intro}
  $\displaystyle\own E * P \vdash \D_{x\ofty A\gets E} P$.
\end{lemma}
\begin{proof}
  Suppose $\gamma,D,\underbrace{\calP_E\pdot\calP_P}_\calP\vDash\own E*P$
  and $\gamma,D,\calP_E\vDash \own E$
  and (1) $\gamma,D,\calP_P\vDash P$.
  To show $\D_{x\gets E}P$, suppose
  $\calP_E\pdot\calP_P\plte(\Sigma_\Omega,\mu)$
  and let $\{\nu_x\}_{x\in A}$ be a $\mu$-disintegration with respect to $E(\gamma)\circ D$.
  We need to show $(\gamma,x),D,\nu_x|_\calP\vDash P$ for almost all $x\in A$.
  Since $x\notin P$, $(\gamma,x),D,\nu_x|_\calP\vDash P$ is equivalent to
  $\gamma,D,\nu_x|_\calP\vDash P$,
  so by
  assumption (1) and monotonicity it suffices to show $\nu_x|_{\calP_P} = \calP_P$ for almost all $x\in A$.

  Write $\calP_P=(\calF_P,\mu_P)$ and
  let $S:=\{x\in A\mid \nu_x|_{\calF_P} = \mu_P\}$.
  It's enough to show that $S$ is $\Sigma_A$-measurable and has probability $1$.
  Let $\calB=\{B_n\}_{n\in\N}$ be a countable basis of $\calF_P$; without loss of generality we may
  assume $\calB$ is a $\pi$-system because any countable collection of sets has countable closure
  under finite intersections.
  By lemma \ref{distributions-agree-pi-generated}, we can write $S$ as the countable intersection
  $S=\bigcap_{n\in\N}S_n$ where $S_n := \{x\in A\mid \nu_x(B_n) = \mu_P(B_n)\}$.
  Because $\sigma$-algebras are closed under countable intersections and measures are countably
  subadditive, $S$ is measurable with probability $1$ if each $S_n$ is.

  Each $S_n$ is $\Sigma_A$-measurable:
  $S_n$ is equal to the preimage of the singleton set $\{\mu_P(B_n)\}$
  under the map $\nu_{(-)}(B_n)$; since $\nu$ is a Markov kernel and singletons are Borel,
  this preimage must be
  $\Sigma_A$-measurable.
  It only remains to show each $S_n$ has probability $1$.
  Now, suppose for the sake of contradiction that there is some $k$ such that $S_k$ does not have probability $1$,
  so $\nu_x(B_n)\ne \mu_P(B_n)$ for all $x\in N:=A\setminus S_k$.
  We can write $N$ as a disjoint union of two subsets $N_<$ and $N_>$, defined as follows:
  \begin{align*}
    N_< &:= \{x\in N\mid \nu_x(B_n) < \mu_P(B_n)\} \\
    N_> &:= \{x\in N\mid \nu_x(B_n) > \mu_P(B_n)\} 
  \end{align*}
  These are both measurable, since they can be written as preimages of $[0,\mu_P(B_n))$
  and $(\mu_P(B_n),1]$ under $\nu_{(-)}(B_n)$.
  Because $N$ has nonzero probability, at least one of $N_<$ or $N_>$ must have nonzero probability too.
  Suppose it's $N_<$; the case where $N_>$ has nonzero probability is analogous.
  Because $\nu$ is a disintegration of $\mu$ with respect to $E(\gamma)\circ D$, we have
  \[
    \Ex_{\omega\sim\mu}f(\omega) = \Ex_{x\sim (E(\gamma)\circ D)_*\mu}\Ex_{\omega\sim\nu_x}f(\omega)
  \]
  for all $f:\Omega\mto\R_{\ge 0}$. Choose $f(\omega):=\ind[\omega\in B_n]\ind[E(\gamma)(D(\omega))\in N_<]$.
  Then simplifying LHS gives
  \begin{align*}
     \Ex_{\omega\sim\mu}f(\omega)
     &=
     \Ex_{\omega\sim\mu}\ind[\omega\in B_n]\ind[E(\gamma)(D(\omega))\in N_<]\\
     &\hspace{-0.3em}\stackrel{(a)}=
     \Ex_{\omega\sim\mu}\ind[\omega\in B_n]\Ex_{\omega\sim\mu}\ind[E(\gamma)(D(\omega))\in N_<]\\
     &=
     \mu(B_n)(E(\gamma)\circ D)_*\mu(N_<)
  \end{align*}
  Step (a) uses independence of $B_n$ and $(E(\gamma)\circ D)^{-1}(N_<)$:
  $\calP_E\pdot\calP_P$ defined and $B_n\in\calP_P$ and $(E(\gamma)\circ D)^{-1}(N_<)\in\calP_E$.
  Meanwhile, simplifying RHS gives
  \begin{align*}
      \Ex_{x\sim (E(\gamma)\circ D)_*\mu}\Ex_{\omega\sim\nu_x} f(\omega)
     &=
       \Ex_{x\sim (E(\gamma)\circ D)_*\mu}\Ex_{\omega\sim\nu_x}
       \ind[\omega\in B_n]\ind[E(\gamma)(D(\omega))\in N_<]\\
     &\hspace{-0.3em}\stackrel{(a)}=
     \Ex_{x\sim (E(\gamma)\circ D)_*\mu} \Ex_{\omega\sim\nu_x}\ind[\omega\in B_n] \ind[x\in N_<]\\
     &= \Ex_{x\sim (E(\gamma)\circ D)_*\mu} \ind[x\in N_<] \Ex_{\omega\sim\nu_x}\ind[\omega\in B_n] \\
     &= \Ex_{x\sim (E(\gamma)\circ D)_*\mu} \ind[x\in N_<] \nu_x(B_n) \\
     &\hspace{-0.3em}\stackrel{(b)}< \Ex_{x\sim (E(\gamma)\circ D)_*\mu} \ind[x\in N_<] \mu_P(B_n) \\
     &= \mu_P(B_n)\Ex_{x\sim (E(\gamma)\circ D)_*\mu}\ind[x\in N_<] \\
     &= \mu_P(B_n)(E(\gamma)\circ D)_*\mu(N_<)\\
     &= \mu(B_n)(E(\gamma)\circ D)_*\mu(N_<)
  \end{align*}
  Step (a) holds because $(E(\gamma)\circ D)_*\nu_x(\{x\}) = 1$ for almost all $x$.
  Step (b) holds because the expectation is taken over $x\in N_<$,
  where the inequality holds by assumption; the inequality remains strict because $N_<$ is
  nonnegligible.
  Putting these two together gives $\mathrm{LHS}=\mathrm{RHS}$ and $\mathrm{LHS}<\mathrm{RHS}$, a
  contradiction.
\end{proof}

\begin{lemma}[law of total expectation] \label{app:C-elim}
  The following entailment holds:
  \[
    \Ex[e[X/x]] = v ~~\land~~ \D_{x:A\gets X} \Ex[E] = e ~~~\vdash~~~ \Ex[E] = v
  \]
\end{lemma}
\begin{proof}
  Fix $(\gamma,D,\calP)$.
  By the first conjunct $\Ex_{\omega\sim\calP} e(\gamma,X(\gamma)(D(\omega))) = v(\gamma)$.
  By assumption, $\calP$ extends to a Borel measure $\mu$ on the Hilbert cube.
  By the disintegration theorem, there exists at least one $\mu$-disintegration
  with respect to $X(\gamma)\circ D$; call it $\{\nu_x\}_{x\in A}$.
  By the second conjunct $\Ex_{\omega\sim\nu_x}[E(\gamma)(D(\omega))] = e(\gamma,x)$ for almost all $x\in A$.
  This along with the existence of the disintegration $\nu$ implies
  \begin{align*}
    \Ex_{\omega\sim\calP} E(\gamma)(D(\omega))
    = \Ex_{x\sim {(X(\gamma)\circ D)}_*\mu} \Ex_{\omega\sim\nu_x} E(\gamma)(D(\omega))
    = \Ex_{x\sim {(X(\gamma)\circ D)}_*\mu} e(\gamma,x)
    = \Ex_{\omega\in\mu} e(\gamma,X(\gamma)(D(\omega)))
    = v(\gamma)
  \end{align*}
  as desired.
\end{proof}

\begin{lemma} \label{app:C-laws}
  The following entailments hold:
  \begin{mathpar}
  \inferrule*[lab=C-Entail]{P~\vdash~ Q}{\D_{x\gets E} P~\vdash~ \D_{x\gets E} Q}
  \and
  \inferrule*[lab=\DIndepName]{}{\own E * P ~\vdash~ \D_{x\gets E} P}
  \and
  \inferrule*[lab=C-Subst]
    {}
    {\own X \vdash \D_{x\gets X} \big(X \asequal x \big)}
  \and
  \inferrule*[lab=C-Own]
    {}
    {\own E \vdash \D_{x\gets X} \own E}
  \and
  \inferrule*[lab=C-Total-Expectation]
    {}
    {\D_{x\gets X} \Ex[E] = e \hspace{.5em}\land\hspace{.5em} \Ex[e[X/x]] = v
      \hspace{.75em}\vdash\hspace{.75em} \Ex[E] = v}
  \end{mathpar}
\end{lemma}
\begin{proof}
  \DIndep~ and \textsc{C-Total-Expectation} follow from lemmas
  \ref{app:C-intro} and \ref{app:C-elim} respectively.
  \begin{itemize}
    \item \textsc{C-Entail}: suppose $P\vdash Q$ and $\gamma,D,\calP\vDash \D_{x\ofty A\gets E}P$.
      Let $\{\mu_x\}_{x\in A}$ be a disintegration of $\calP$ with respect to $E(\gamma)\circ D$;
      let $\{\calP_x\}_{x\in A}$ be the corresponding restrictions of $\{\mu_x\}_{x\in A}$ to $\calP$.
      By assumption, $\gamma,D,\calP_x\vDash P$ for almost-all $x\in A$.
      Since $P\vdash Q$, this implies $\gamma,D,\calP_x\vDash Q$ for almost-all $x\in A$ as desired.
    \item \textsc{C-Subst}:
      Fix $(\gamma,D,\calP)$.
      Let $\{\mu_x\}_{x\in A}$ be a disintegration of $\calP$ with respect to $X(\gamma)\circ D$.
      Let $E_x$ be the event that $X(\gamma)\circ D$ is equal to $x$.
      By assumption we have that $X$ is $\calP$-measurable, so it only remains to show
      that $E_x$ holds almost-surely with respect to $\mu_x$ for almost all $x\in A$.
      By the definition of disintegration,
      the event $\{\omega \mid X(\omega) = x\}$ has probability $1$ under $\mu_x$
      for almost all $x$ as required.
    \item \textsc{C-Own}:
      Fix $(\gamma,D,(\calF,\mu))$ and
      let $\{\mu_x\}_{x\in A}$ be a disintegration of $(\calF,\mu)$ with respect to
      $X(\gamma)\circ D$.  Disintegration only changes the probability measure;
      the underlying $\sigma$-algebra remains fixed. Thus if $E(\gamma)\circ D$ is
      $\calF$-measurable then it remains $\calF$-measurable under each
      conditional probability space.
  \end{itemize}
\end{proof}

\begin{lemma} \label{app:C-modal-laws}
  The following entailments hold:
  \begin{itemize}
  \item Necessitation: if $\vdash P$ then $\displaystyle\vdash \D_{x\gets X}P$.
  \item Distribution: $\displaystyle\D_{x\gets X}(P\to Q)\vdash \D_{x\gets X} P\to \D_{x\gets X}Q$.
  \item Distributes over $(\land)$: $\displaystyle\D_{x\gets X} (P\land Q) \vdashiff \D_{x\gets X} P \land
    \D_{x\gets X} Q$.
  \item Semidistributes over $(\lor)$: $\displaystyle\D_{x\gets X} P \lor \D_{x\gets X} Q \vdash \D_{x\gets
    X} (P\lor Q)$.
  \end{itemize}
\end{lemma}
\begin{proof}~
  \begin{itemize}
  \item Necessitation: if $P$ holds in all configurations
    then it holds for all disintegrated configurations as well.
  \item Distribution: it suffices to show $D_{x\gets X}(P\to Q)\land \D_{x\gets X}P\vdash \D_{x\gets X}Q$.
    By \textsc{C-And} the premise is equivalent to $\D_{x\gets X}((P\to Q)\land P)$;
    the result then follows from \textsc{C-Entail}
    via the entailment $(P\to Q)\land P\vdash Q$.
  \item Distributes over $(\land$):
      the left-to-right direction follows from \textsc{C-Entail}
      via the entailments $P\land Q\vdash P$ and $P\land Q\vdash Q$.
      For the right-to-left entailment, suppose $\gamma,D,\calP\vDash \D_{x\gets E} P$
      and $\gamma,D,\calP\vDash \D_{x\gets E} Q$
      and let $\{\mu_x\}_{x\in A}$ be a disintegration of $\calP$ with respect to $E(\gamma)\circ D$;
      let $\{\calP_x\}_{x\in A}$ be the corresponding restrictions of $\{\mu_x\}_{x\in A}$ to $\calP$.
      By assumption, there are two sets $F_1,F_2\subseteq A$ of measure $1$
      such that $\gamma,D,\calP_x\vDash P$ for all $x\in F_1$
      and $\gamma,D,\calP_x\vDash Q$ for all $x\in F_2$.
      Therefore, $\gamma,D,\calP_x\vDash P\land Q$ for all $x\in F_1\cap F_2$.
      Moreover, $F_1\cap F_2$ has measure $1$ by subadditivity,
      so $\gamma,D,\calP_x\vDash P\land Q$ for almost-all $x$ as desired.
  \item Semidistributes over $(\lor)$:
    it suffices to show $\D_{x\gets X} P\vdash \D_{x\gets X}(P\lor Q)$
    and $\D_{x\gets X}Q\vdash \D_{x\gets X}(P\lor Q)$.
    These follow from \textsc{C-Entail} via the entailments
    $P\vdash P\lor Q$ and $Q\vdash P\lor Q$ respectively.
  \end{itemize}
\end{proof}

\section{Annotated \ref{prog:commoncause} program} \label{sec:annotated-commoncause}

{\small
\begin{align*}
  &\assertc{ \top }\\
  &\hspace{1em}\appllet Z \appleq \flip~1/2 \applin \\
  &\assertc{ Z \sim \ber 1/2 }\\
  &\hspace{1em}\appllet X \appleq \flip~1/2 \applin \\
  &\assertc{ Z \sim \ber 1/2 ~~*~~  X\sim\ber 1/2 }\\
  &\hspace{1em}\appllet Y \appleq \flip~1/2 \applin \\
  &\assertc{ Z \sim \ber 1/2 ~~*~~  X\sim\ber 1/2 ~~*~~ Y\sim\ber 1/2 }\\
  &\hspace{1em}\appllet A \appleq X~\applkw{||}~Z \applin \\
  &\assertc{ Z \sim \ber 1/2 ~~*~~  X\sim\ber 1/2 ~~*~~ Y\sim\ber 1/2
              ~~*~~ A\asequal (X\lor Z) }\\
  &\hspace{1em}\appllet B \appleq Y~\applkw{||}~Z \applin \\
  &\assertc{ Z \sim \ber 1/2 ~~*~~  X\sim\ber 1/2 ~~*~~ Y\sim\ber 1/2
              ~~*~~ A\asequal (X\lor Z) ~~*~~ B\asequal (Y\lor Z)}\\
  &\hspace{1em}\applret (Z, X, Y, A, B) \\
  &\assertc{ \own Z \land \D_{z\gets Z}  \left(X\sim\ber 1/2 ~~*~~ Y\sim\ber 1/2
              ~~*~~ A\asequal (X\lor Z) ~~*~~ B\asequal (Y\lor Z)\right)}~(\DIndep)\\
  &\assertc{ \D_{z\gets Z}  \underbrace{\left(X\sim\ber 1/2 ~~*~~ Y\sim\ber 1/2
              ~~*~~ A\asequal (X\lor z) ~~*~~ B\asequal (Y\lor z)\right)}_P}~(\textsc{C-Subst})\\
  &\assertc{ \D_{z\gets Z} \left((z = \vtrue\to P[\vtrue/z]) \land (z=\vfalse\to P[\vfalse/z])\right)} \\
  &\assertc{ \D_{z\gets Z} \left(
    \begin{aligned}
      &(z = \vtrue~~\to~~ X\sim\ber 1/2 ~~*~~ Y\sim\ber 1/2
              ~~*~~ A\asequal (X\lor \vtrue) ~~*~~ B\asequal (Y\lor \vtrue)) ~~\land \\
      &(z=\vfalse ~~\to~~
              X\sim\ber 1/2 ~~*~~ Y\sim\ber 1/2
              ~~*~~ A\asequal (X\lor \vfalse) ~~*~~ B\asequal (Y\lor \vfalse))
    \end{aligned}\right)} \\
  &\assertc{ \D_{z\gets Z} \left(
    \begin{aligned}
      &(z = \vtrue~~\to~~ A\asequal \vtrue ~~*~~ B\asequal \vtrue) ~~\land\\
      &(z=\vfalse ~~\to~~
              X\sim\ber 1/2 ~~*~~ Y\sim\ber 1/2
              ~~*~~ A\asequal X ~~*~~ B\asequal Y)
    \end{aligned}\right)} \\
  &\assertc{ \D_{z\gets Z} \left(
    \begin{aligned}
      &(z = \vtrue~~\to~~ \own A ~~*~~ \own B) ~~\land \\
      &(z=\vfalse ~~\to~~ A\sim\ber 1/2 ~~*~~ B\sim\ber 1/2)
    \end{aligned}\right)} \\
  &\assertc{ \D_{z\gets Z} \left(
    \begin{aligned}
      &(z = \vtrue\to \own A * \own B) ~~\land\\
      &(z=\vfalse \to \own A * \own B)
    \end{aligned}\right)} \\
  &\assertc{ \D_{z\gets Z} \left(\own A * \own B\right) }
\end{align*}
}

\section{An Example of Conditional Independence via Control Flow} \label{sec:condsamples-annotated}

{\small
\begin{align*}
  &\assertc{ \top }\\
  &\hspace{1em}\appllet Z \appleq \flip~1/2 \applin \\
  &\assertc{ Z \sim \ber 1/2 }\\
  &\hspace{1em}\applif~Z~\applthen \\
  &\hspace{1em}\hspace{2em}\appllet X_1 \appleq \flip~p \applin \\
  &\hspace{2em}\assertc{ X_1 \sim \ber p }\\
  &\hspace{1em}\hspace{2em}\appllet Y_1 \appleq \flip~p \applin \\
  &\hspace{2em}\assertc{ X_1 \sim \ber p ~~*~~ Y_1 \sim\ber p }\\
  &\hspace{1em}\hspace{2em}\applret (Z,X_1,Y_1) \\
  &\hspace{1em}\applelse \\
  &\hspace{2em}\assertc{ \existsrv~X_1~Y_1.~~X_1 \sim \ber p ~~*~~ Y_1 \sim\ber p }\\
  &\hspace{1em}\hspace{2em}\appllet X_2 \appleq \flip~q \applin \\
  &\hspace{2em}\assertc{ \left(\existsrv~X_1~Y_1.~~X_1 \sim \ber p ~~*~~ Y_1 \sim\ber p\right)
                          ~~*~~ X_2 \sim \ber q }\\
  &\hspace{1em}\hspace{2em}\appllet Y_2 \appleq \flip~q \applin \\
  &\hspace{2em}\assertc{ \left(\existsrv~X_1~Y_1.~~X_1 \sim \ber p ~~*~~ Y_1 \sim\ber p\right)
                          ~~*~~ X_2 \sim \ber q ~~*~~ Y_2 \sim \ber q }\\
  &\hspace{1em}\hspace{2em}\applret (Z,X_2,Y_2) \\
  &\assertc{
    \begin{aligned}
      &Z \sim \ber 1/2 ~~*~~ \existsrv~X_1~Y_1~X_2~Y_2.~~ X_1 \sim \ber p ~~*~~ Y_1 \sim\ber p
      ~~*~~ X_2 \sim \ber q ~~*~~ Y_2 \sim \ber q ~~*~~ \\
      &X \asequal (\ite Z{X_1}{X_2}) ~~*~~ Y \asequal (\ite Z{Y_1}{Y_2})
    \end{aligned}}\\
  &\assertc{
    \own Z \land \D_{z\gets Z}\left( \existsrv~X_1~Y_1~X_2~Y_2.~~
    \begin{aligned}
      &X_1 \sim \ber p ~~*~~ Y_1 \sim\ber p
      ~~*~~ X_2 \sim \ber q ~~*~~ Y_2 \sim \ber q ~~*~~ \\
      &X \asequal (\ite Z{X_1}{X_2}) ~~*~~ Y \asequal (\ite Z{Y_1}{Y_2})
    \end{aligned}\right)}~(\DIndep)\\
  &\assertc{
    \D_{z\gets Z}\underbrace{\left( \existsrv~X_1~Y_1~X_2~Y_2.~~
    \begin{aligned}
      &X_1 \sim \ber p ~~*~~ Y_1 \sim\ber p
      ~~*~~ X_2 \sim \ber q ~~*~~ Y_2 \sim \ber q ~~*~~ \\
      &X \asequal (\ite z{X_1}{X_2}) ~~*~~ Y \asequal (\ite z{Y_1}{Y_2})
    \end{aligned}\right)}_P}~(\textsc{C-Subst})\\
  &\assertc{
    \D_{z\gets Z}\left(
    \begin{aligned}
      &(z=\vtrue ~~\to~~ P[\vtrue/z]) ~~\land \\
      &(z=\vfalse ~~\to~~ P[\vfalse/z])
    \end{aligned}\right)}\\
  &\assertc{
    \D_{z\gets Z}\left(
    \begin{aligned}
      &(z=\vtrue ~~\to~~ \existsrv~X_1~Y_1.~~X_1\sim\ber p~~*~~ Y_1\sim\ber p
           ~~*~~ X\asequal X_1
           ~~*~~ Y\asequal Y_1) ~~\land\\
      &(z=\vfalse ~~\to~~ \existsrv~X_2~Y_2.~~X_2\sim\ber q~~*~~ Y_2\sim\ber q
           ~~*~~ X\asequal X_2
           ~~*~~ Y\asequal Y_2)
    \end{aligned}\right)}\\
  &\assertc{
    \D_{z\gets Z}\left(
    \begin{aligned}
      &(z=\vtrue ~~\to~~ X\sim\ber p~~*~~ Y\sim\ber p) ~~\land\\
      &(z=\vfalse ~~\to~~ X\sim\ber q~~*~~ Y\sim\ber q)
    \end{aligned}\right)}\\
  &\assertc{
    \D_{z\gets Z}\left(
    \begin{aligned}
      &(z=\vtrue ~~\to~~ \own X ~~*~~ \own Y) ~~\land\\
      &(z=\vfalse ~~\to~~\own X ~~*~~ \own Y)
    \end{aligned}\right)}\\
  &\assertc{ \D_{z\gets Z}(\own X * \own Y) }\\
\end{align*}
}

\section{Proving a Weighted Sampling Algorithm Correct (Full)} \label{sec:full-verification}

\begin{figure*}
  \centering
\begin{subfigure}{0.25\linewidth}
{\fontsize{7.5pt}{9pt}\selectfont {\begin{align*}
  &1~\,\appllet W \appleq \applret(w_1+\dots+w_n) \applin \\
  &2~\,\appllet (v_1,\dots,v_n) \appleq \applret(w_1/W,\dots,w_n/W)\applin \\
  &3~\,\appllet U \appleq \applunif \applin \\
  &4~\,\applforopen n0iJ{\\
  &5~\,\hspace{1em}\applif~v_1+\dots+v_{i-1} \le U < v_1 + \dots + v_i \\
  &6~\,\hspace{1em}\applthen~\applret i \\
  &7~\,\hspace{1em}\applelse~\applret J}\applforclose
\end{align*}}}
\caption{Naive linear-space implementation of weighted sampling.}
\label{fig:motiv-naive-full}
\end{subfigure}
\hspace{1em}
\begin{subfigure}{0.25\linewidth}
{\fontsize{7.5pt}{9pt}\selectfont \begin{align*}
  &1~\,\appllet w \appleq \applret\applarray{w_1,\dots,w_n} \applin \\
  &2~\,\appllet M \appleq \applret (-\infty) \applin \appllet K \appleq \applret(0)\applin \\
  &3~\,\applforopen n{\appltuple{M,K}}i{\appltuple{M,K}}{\\
  &4~\,\hspace{1em}\appllet S \appleq\applunif \applin \\
  &5~\,\hspace{1em}\appllet U \appleq \applret(S\applexp\applpar{1\applkw{/}\applindex wi})\applin \\
  &6~\,\hspace{1em}\applif~U>M\\
  &7~\,\hspace{1em}\applthen~\applret \appltuple{U,i}\\
  &8~\,\hspace{1em}\applelse~\applret \appltuple{M,K}}\applforclose
\end{align*}}
\caption{Constant-space version.}
\label{fig:motiv-const-full}
\end{subfigure}
\hspace{1em}
\begin{subfigure}{0.3\linewidth}
  \begin{tikzpicture}
    \begin{axis}[
      width=\linewidth,
        xmin = 0, xmax = 1,
        ymin = 0, ymax = 1]
        \addplot[
            domain = 0:1,
            samples = 100,
            smooth,
            thick,
            blue,
        ] {x};
        \addplot[
            domain = 0:1,
            samples = 100,
            smooth,
            thick,
            gray,
        ] {x^(1/4)};
        \addplot[
            domain = 0:1,
            samples = 100,
            smooth,
            thick,
            red,
        ] {x^4};
    \end{axis}
    \node at (1,1) [rotate=40] {$w=1$};
    \node at (0.4,1.5) [rotate=40] {$w=4$};
    \node at (1.4,0.5) [rotate=40] {$w=\frac{1}{4}$};
    \end{tikzpicture}
    \caption{Visualization of the function $f(x)=x^{1/w}$.}
    \label{fig:funcvis-full}
\end{subfigure}
\caption{Weighted sampling example. The constants $w_i$ are inputs.}
\end{figure*}

To exercise Lilac's support for conditional reasoning, continuous random variables,
and substructural handling of independence, we now prove a sophisticated 
constant-space \emph{weighted sampling algorithm} correct using Lilac.  Suppose
you are given a collection of items $\{x_1, \dots, x_n\}$ each with associated
weight $w_i \in \mathbb{R}^+$.  The task is to draw a sample from the collection
$\{x_i\}$ in a manner where each item is drawn with probability proportional to
its weight. This problem is an instance of \emph{reservoir
sampling}~\citep{efraimidis2006weighted}, and is an important primitive in
distributed systems.

First, we consider a naive solution
that requires space linear in the number of weights; pseudocode
for this algorithm is presented in Figure~\ref{fig:motiv-naive-full}. The first 
pass over the weights occurs on Line~1, which computes the normalizing constant
$W=\sum_i w_i$. Line~2 then divides each weight by $W$ so that the result
$(v_1,\dots,v_n)$ forms a probability distribution.
This distribution can be thought of as a partitioning of the interval $[0,1]$
into $n$ subintervals 
with lengths $(v_1,\dots,v_n)$; to sample from it
we can choose a point $U$ uniformly at random from $[0,1]$ (Line~3)
and select the item corresponding to the subinterval that $U$ lands in (Lines~4--7). 

While simple to understand and implement, this naive approach has a critical
flaw that makes it inappropriate for application in large-scale systems: it requires storing
all previously encountered weights and scanning over them before a single sample can be
drawn, and so does not scale to a streaming setting where new weights are acquired
one at a time (for instance, as each user visits a website). To fix this
limitation, \citet{efraimidis2006weighted} proposed the very clever
\emph{constant-space solution} presented in Figure~\ref{fig:motiv-const-full}. The
core of this approach is to generate a value $S$ uniformly at random from $[0,1]$ on \emph{every}
iteration (Line~3), perturb $S$ according to the next weight $w_i$ in the stream
(Line~4), and track only the \emph{greatest} perturbed sample (Lines~5--8). Figure~\ref{fig:funcvis-full}
gives some intuition for the perturbed quantity $S^{1/w_i}$ on Line~4: if $w_i$
is large (i.e., item $i$ has high weight), then $S^{1/w_i}$ is likely to be large (visualized
by the curve $w=4$); if $w_i$ is small, then $S^{1/w_i}$ is likely to be small (visualized
by the curve $w=1/4$).
The fact that this program is equivalent to the naive one is quite surprising,
and proving it requires the simultaneous application of several important theorems from
probability theory. We show how this can be done formally in Lilac
in a manner similar to a typical informal proof.
Correctness is captured by the following Lilac postcondition:
\begin{align}
 {\forall k. \Pr(K=k) = \frac{w_k}{\sum_j w_j}}
 \label{eq:post-full}
\end{align}
To establish this postcondition, a typical informal proof
begins by declaring mutually independent, uniformly distributed random variables $\{S_i\}_{1\le i\le n}$,
where $S_i$ denotes the value sampled by Line~4 on the $i$th loop iteration,
and a random variable $K=\argmax_i S_i^{1/w_i}$ that denotes the final result.
Implicit in this setup are the assumptions that each $S_i$ produced by the program is actually
independent and uniformly distributed, and that the for-loop actually computes
the specified $\argmax$. We can formally establish this
by mechanically applying the proof rules described in Section~\ref{sec:derived}
to conclude the following at program termination:
\begin{align}
  {
  \existsrv S_1 \dots S_n.~\hugestar_i S_i \sim \mathrm{Unif}~[0,1] ~~*~~ 
  K \asequal \argmax_i S_i^{1/w_i}}
  \label{eq:post-full1}
\end{align}
The proof makes use of the following
invariant $I_j$ for the loop on
Line 2, which must hold immediately before the execution of the $j$th iteration for all $1\le j\le n+1$:
\begin{align}
  I_j \hspace{0.75em}:=\hspace{0.75em} {
   \existsrv S_1 \dots S_j.~~ \hugestar_{1\le i< j}~ S_i \sim \mathrm{Unif}~[0,1]
      \hspace{0.5em}*\hspace{0.5em} K \asequal \argmax_{1\le i<j} S_i^{1/w_i}
      \hspace{0.5em}*\hspace{0.5em} M \asequal \max_{1\le i<j} S_i^{1/w_i}  }
  \label{eq:loopinv1-full} 
\end{align}
The proof that our program maintains this invariant is completely standard for separation logics,
so we elide the details and focus on the challenge of
deriving the desired post-condition (\ref{eq:post-full}) given the
setup (\ref{eq:post-full1}).
To show (\ref{eq:post-full}) in the case $i=k$,
note that
\begin{align}
  \Pr(K=k) = \Pr\left(S_k^{1/w_k} > S_j^{1/w_j} \text{ for all } j \ne k\right),
\end{align}
since $K$ is defined to be the $\argmax$ of $j$ over all $S_j^{1/w_j}$.
This is an unwieldy probability to compute directly. The trick is to use conditioning:
in this case, fixing $S_k$ to a deterministic $s_k$ gives
  {\small 
\begin{align}
  \Pr(K=k \mid S_k = s_k) &= \Pr\left(s_k^{1/w_k} > S_j^{1/w_j} \text{ for all } j \ne k\right) \label{eq:eqset1:start-full} \\
    &= \Pr\left(s_k^{w_j/w_k} > S_j \text{ for all } j \ne k\right) & \text{Exponentiating} \\
    &= \prod_{j \ne k} \Pr\left(s_k^{w_j/w_k} > S_j\right) & \text{By conditional independence} \label{eq:indep-ex-full}\\ 
    &= \prod_{j \ne k} s_k^{w_j/w_k} & S_j\text{ uniform}\footnotemark \\
    &= \pow{s_k}{\frac{\sum_{j\ne k} w_j}{w_k}}.\label{eq:eqset1:end-full} 
\end{align}
  }%
\footnotetext{If $U$ uniform in $[0,1]$ then $\Pr(u > U) = u$.}%
Formally, this calculation occurs under the modality $\D_{s_k\gets S_k}$, which is introduced via \DIndep.
The expression $\Pr(E)$ abbreviates $\Ex[\ind[E]]$, the expectation of the indicator random variable $\ind[E]$.\footnote{%
If $E$ is an event then the random variable $\ind[E]$ is $1$ if $E$ holds and $0$ otherwise.}
A critical step occurs in Equation~\ref{eq:indep-ex-full}: each
$S_j$ is conditionally independent from all others given $S_k = s_k$. This
permits a critical simplification: the probability of the conjunction becomes a
product of simpler probabilities. This is an application of the derived rule
\begin{align}
  \tag{\textsc{Indep-Prod}}
  {
  \hugestar_i \own E_i ~~~\vdash~~~ \Pr\Big( \bigcap_i E_i \Big) = \prod_i \Pr(E_i),
  } \label{eq:indep-prod-full}
\end{align}
an immediate consequence of Lemma~\ref{lem:star-is-independence}.
Note that our modal treatment of conditioning leads to a nice separation of concerns here.
Because $\D$ respects entailment, facts like \ref{eq:indep-prod-full} that appear to be 
only about unconditional independence and unconditional
probability are automatically lifted to facts like Equation~(\ref{eq:indep-ex-full}),
with the expected conditional reading.

Finally, to complete the proof we connect the conditional
$\Pr(K=k \mid S_k = s_k)$ to the unconditional $\Pr(K=k)$
using the law of total expectation:
{\small 
\begin{align}
  \Pr(K=k)
  &= \Ex\left[\Pr(K=k\mid S_k)\right] & \text{Law of Total Expectation}\label{eq:eqset2:start-full}\\
  &= \Ex\left[\pow{S_k}{\frac{\sum_{j\ne k} w_j}{w_k}}\right] & \text{By }(\ref{eq:eqset1:end-full})\\
  &= \left(\frac{\sum_{j\ne k} w_j}{w_k} + 1\right)^{-1} & \text{$S_k$ uniform}\footnotemark\\
  &= \frac{w_k}{\sum_j w_j}.\label{eq:eqset2:end-full}
\end{align}}%
Unlike the calculation in Equations~\ref{eq:eqset1:start-full}--\ref{eq:eqset1:end-full},
which take place inside the modality $\D_{s_k\gets S_k}$,
this second calculation (Equations~\ref{eq:eqset2:start-full}--\ref{eq:eqset2:end-full}) takes place outside 
of it, as it computes the \emph{unconditional} probability $\Pr(K=k)$.
The gap between the two calculations is bridged by
the following instantiation of \textsc{C-Total-Expectation}:
{\footnotesize
\begin{align*}
  \D_{s_k \gets S_k} \Big(\Ex[\underbrace{\ind[K=k]}_{E}] 
  = \underbrace{\pow{s_k}{\frac{\sum_{j \ne k}{w_j}}{w_k}}}_{e} \Big)
  \land 
  \Big(
  \Ex\Big[\underbrace{\mathrm{pow}\Big(S_k, \frac{\sum_{j \ne k} w_j}{w_k} \Big)}_{e[S_k/s_k]}\Big]
  =
  \underbrace{\frac{w_k}{\sum_j w_j}}_{v}
  \Big)
  ~~\mathlarger{\vdash}~~
  \Ex[\underbrace{\ind[K=k]}_{E}] = \underbrace{\frac{w_k}{\sum_j w_j}}_{v}
\end{align*}
}%
\footnotetext{If $U$ uniform in $[0,1]$ then $\Ex[U^\alpha] = 1/({\alpha + 1})$.}%
Putting all this together yields a formal proof of correctness in Lilac.
The next page gives a fully annotated program.

\clearpage
{\small
\begin{align*}
  &\assert{\{ \top \}}\\
  &\hspace{1em}\appllet W \appleq \applret\applarray{w_1,\dots,w_n} \applin \\
  &\assert{\{ W \asequal (w_1,\dots,w_n) \}}\\
  &\hspace{1em}\appllet M \appleq \applret(-\infty) \applin \\
  &\assert{\{ W \asequal (w_1,\dots,w_n) ~~*~~ M\asequal -\infty  \}}\\
  &\hspace{1em}\appllet K \appleq \applret 0 \applin \\
  &\assert{\{ W\asequal (w_1,\dots,w_n) ~~*~~ M\asequal -\infty ~~*~~ K\asequal 0 \}}\\
  &\assert{\text{Let } \begin{aligned}
    I(i,M,K)
    \quad:=\quad\left(\begin{aligned}
       &W \asequal (w_1,\dots,w_n) ~~*~~ 1\le i\le n ~~*~~\\
       &\exists S_1 \dots S_i.~~ \hugestar_{1\le j< i}~ S_j \sim \mathrm{Unif}~[0,1]
          \hspace{0.3em}*\hspace{0.3em} K \asequal \argmax_{1\le j<i} S_j^{1/w_j}
          \hspace{0.3em}*\hspace{0.3em} M \asequal \max_{1\le j<i} S_j^{1/w_j}  
    \end{aligned}\right)
   \end{aligned}}\\
  &\assert{\{ I(1,M,K) \}}\\
  &\hspace{1em}\applforopen n{\appltuple{M,K}}i{\appltuple{M,K}}{\\
  &\hspace{1em}\assert{\{ I(i,M,K) \}}\\
  &\hspace{1em}\hspace{1em}\appllet S \appleq\applunif \applin \\
  &\hspace{1em}\assert{\{ I(i,M,K) ~~*~~ S\sim\unif~[0,1] \}}\\
  &\hspace{1em}\hspace{1em}\appllet U \appleq \applret(S\applexp\applpar{1\applkw{/}\applindex wi}) \applin \\
  &\hspace{1em}\assert{\{ I(i,M,K) ~~*~~ S\sim\unif~[0,1] ~~*~~ U\asequal S^{1/w_i} \}}\\
  &\hspace{1em}\hspace{1em}\applif~U>M\\
  &\hspace{1em}\hspace{1em}\applthen~\applret \appltuple{U,i}\\
  &\hspace{1em}\hspace{1em}\applelse~\applret \appltuple{M,K}}\applforclose\\
  &\hspace{1em}\assertc{ (M',K').\hspace{1em}
     \begin{aligned}
      &I(i,M,K) ~~*~~ S\sim\unif~[0,1] ~~*~~ U\asequal S^{1/w_i} ~~*~~ \\
      &M' \asequal \max(U,M) ~~*~~ U' \asequal \ite{U>M}iK
     \end{aligned}
    }\\
  &\hspace{1em}\assert{\{ (M',K').~~ I(i+1,M',K') \}}\\
  &\assert{\{ (M',K').~~I(n+1,M',K') \}}\\
  &\assertc{
    \existsrv S_1 \dots S_n.~\hugestar_i S_i \sim \mathrm{Unif}~[0,1] ~~*~~ 
    K \asequal \argmax_i S_i^{1/w_i}}\\
  &\assertc{\forall k. \Pr(K=k) = \frac{w_k}{\sum_j w_j}}
\end{align*}
}

To illustrate the proof of the final entailment,
we animate the proof state at each step in inference-rule notation,
in the style of interactive theorem provers such as Coq.
First we work backwards from the goal:

\begin{mathpar}
  \inferrule{\hugestar_i S_i \sim \mathrm{Unif}~[0,1] \\ K \asequal  \argmax_i S_i^{1/w_i}}
    {\Pr[K=k]=\frac{w_k}{\sum_j w_j}}
  \\
  \inferrule{\hugestar_i S_i \sim \mathrm{Unif}~[0,1] \\ K \asequal  \argmax_i S_i^{1/w_i}}
    {\Pr[\forall j\ne k.~S_k^{1/w_k}> S_j^{1/w_j}]=\frac{w_k}{\sum_j w_j}}
  \\
  \inferrule{\hugestar_i S_i \sim \mathrm{Unif}~[0,1] \\ K \asequal  \argmax_i S_i^{1/w_i}}
    {\Pr[\forall j\ne k.~S_k^{w_j/w_k}> S_j] =\frac{w_k}{\sum_j w_j}}
  \\
  \inferrule{\hugestar_i S_i \sim \mathrm{Unif}~[0,1] \\ K \asequal  \argmax_i S_i^{1/w_i}}
    {\Ex\left[\ind[\forall j\ne k.~S_k^{w_j/w_k}> S_j]\right] =\frac{w_k}{\sum_j w_j}}
  \\
  \inferrule{\hugestar_i S_i \sim \mathrm{Unif}~[0,1] \\ K \asequal  \argmax_i S_i^{1/w_i}}
    { \Ex\left[\prod_{j\ne k}\ind[S_k^{w_j/w_k}> S_j]\right] =\frac{w_k}{\sum_j w_j}}
\end{mathpar}
At this point we begin working forwards from the hypotheses,
using \DIndep~
to introduce the conditioning modality $\D_{s_k\gets S_k}$
with the aim of computing the conditional probability $\Pr(K=k\mid S_k=s_k)$.
\begin{mathpar}
  \inferrule{S_k\sim\mathrm{Unif}~[0,1]\land \D_{s_k\gets S_k} \hugestar_{j\ne k} S_j \sim \mathrm{Unif}~[0,1] 
        \\ K \asequal  \argmax_i S_i^{1/w_i}}
    { \Ex\left[\prod_{j\ne k}\ind[S_k^{w_j/w_k}> S_j]\right] =\frac{w_k}{\sum_j w_j}}
  \\
  \inferrule{S_k\sim\mathrm{Unif}~[0,1]\land \D_{s_k\gets S_k} \hugestar_{j\ne k} 
    S_j \sim \mathrm{Unif}~[0,1] \land \left(
    \Ex\left[\prod_{j\ne k}\ind[S_k^{w_j/w_k}> S_j]\right]
    =\Ex\left[\prod_{j\ne k}\ind[s_k^{w_j/w_k}> S_j]\right]\right)
      \\ K \asequal  \argmax_i S_i^{1/w_i}}
    { \Ex\left[\prod_{j\ne k}\ind[S_k^{w_j/w_k}> S_j]\right] =\frac{w_k}{\sum_j w_j}}
\end{mathpar}
Next, we use conditional independence of $\{S_j\}_{j\ne k}$ given $S_k$, encoded in the
iterated separating conjunction underneath $\D_{s_\gets S_k}$, to interchange product and expectation:
\begin{mathpar}
  \inferrule{S_k\sim\mathrm{Unif}~[0,1]\land \D_{s_k\gets S_k} \hugestar_{j\ne k} 
    S_j \sim \mathrm{Unif}~[0,1] \land    \left(
    \Ex\left[\prod_{j\ne k}\ind[S_k^{w_j/w_k}> S_j]\right]
     = \prod_{j\ne k}\Ex[\ind[s_k^{w_j/w_k}> S_j]]
    \right)
      \\ K \asequal  \argmax_i S_i^{1/w_i}}
    { \Ex\left[\prod_{j\ne k}\ind[S_k^{w_j/w_k}> S_j]\right] =\frac{w_k}{\sum_j w_j}}
\end{mathpar}
Now significant simplifications are possible, completing the first calculation
(Equations~\ref{eq:eqset1:start-full}--\ref{eq:eqset1:end-full}):
\begin{mathpar}
  \inferrule{S_i\sim\mathrm{Unif}~[0,1]\land \D_{s_k\gets S_k} \hugestar_{j\ne k} 
    S_j \sim \mathrm{Unif}~[0,1] \land \left(
    \Ex\left[\prod_{j\ne k}\ind[S_k^{w_j/w_k}> S_j]\right]
     = \prod_{j\ne k}s_k^{w_j/w_k}
    \right)
      \\ K \asequal  \argmax_i S_i^{1/w_i}}
    { \Ex\left[\prod_{j\ne k}\ind[S_k^{w_j/w_k}> S_j]\right] =\frac{w_k}{\sum_j w_j}}
  \\
  \inferrule{S_k\sim\mathrm{Unif}~[0,1]\land \D_{s_k\gets S_k} \hugestar_{j\ne k} 
    S_j \sim \mathrm{Unif}~[0,1] \land \left(
    \Ex\left[\prod_{j\ne k}\ind[S_k^{w_j/w_k}> S_j]\right]
    = \exp\left(s_k,\frac{\sum_{j\ne k} w_j}{w_k}\right)
    \right)
      \\ K \asequal  \argmax_i S_i^{1/w_i}}
    { \Ex\left[\prod_{j\ne k}\ind[S_k^{w_j/w_k}> S_j]\right] =\frac{w_k}{\sum_j w_j}}
  \\
  \inferrule{S_k\sim\mathrm{Unif}~[0,1]\land \D_{s_k\gets S_k}
    \Ex\left[\prod_{j\ne k}\ind[S_k^{w_j/w_k}> S_j]\right]
     = \exp\left(s_k,\frac{\sum_{j\ne k} w_j}{w_k}\right)}
    { \Ex\left[\prod_{j\ne k}\ind[S_k^{w_j/w_k}> S_j]\right] =\frac{w_k}{\sum_j w_j}}
\end{mathpar}
Having completed the computation of the conditional probability $\Pr(K=k\mid S_k=s_k)$
by working forwards from the hypotheses, we eliminate the conditioning modality by applying the
law of total expectation (\textsc{C-Total-Expectation}):
\begin{mathpar}
  \inferrule{S_k\sim\mathrm{Unif}~[0,1]\land
    \Ex\left[\prod_{j\ne k}\ind[S_k^{w_j/w_k}> S_j]\right]
     = \Ex\left[\exp\left(S_k,\frac{\sum_{j\ne k} w_j}{w_k}\right)\right]}
    { \Ex\left[\prod_{j\ne k}\ind[S_k^{w_j/w_k}> S_j]\right] =\frac{w_k}{\sum_j w_j}}
\end{mathpar}
The remainder of the calculation is straightforward, following
Equations~\ref{eq:eqset2:start-full}--\ref{eq:eqset2:end-full}:
\begin{mathpar}
  \inferrule{S_k\sim\mathrm{Unif}~[0,1]\land
    \Ex\left[\prod_{j\ne k}\ind[S_k^{w_j/w_k}> S_j]\right]
      = \frac1{\frac{\sum_{j\ne k} w_j}{w_k} + 1}}
    { \Ex\left[\prod_{j\ne k}\ind[S_k^{w_j/w_k}> S_j]\right] =\frac{w_k}{\sum_j w_j}}
  \\
  \inferrule{ \Ex\left[\prod_{j\ne k}\ind[S_k^{w_j/w_k}> S_j]\right]
  = \frac{w_k}{\sum_j w_j}}
    { \Ex\left[\prod_{j\ne k}\ind[S_k^{w_j/w_k}> S_j]\right] =\frac{w_k}{\sum_j w_j}}
  \\
  \inferrule{}{\mathbf{QED}}
\end{mathpar}

\section{Examples from Barthe et.\ al.} \label{sec:barthe-examples}

In this section we consider three of the five examples presented in \citet{barthe2019probabilistic}:
one-time pad, oblivious transfer, and private information retrieval.
In each example, the goal is to
verify the perfect secrecy of a cryptographic protocol. Perfect secrecy is established
via two methods: uniformity, which aims to show that each agent's view of others' data is
uniformly distributed at exit, and input independence, which aims to show that the
encrypted output of the protocol is independent of the input.

\citet{barthe2019probabilistic} use PSL to establish perfect secrecy of one-time pad and
private information retrieval via both uniformity and input independence,
and perfect secrecy of oblivious transfer via uniformity.
We will show how the same can be done in Lilac.
\citet{barthe2019probabilistic} also observe that
the input independence proof for oblivious transfer gets stuck, mentioning that 
even an informal proof sketch does not seem easy.
We will show that the
postcondition specifying input indendence for oblivious transfer is in fact unsatisfiable
by giving a countermodel.

To do this, we add some support for length-$n$ bitvectors and reasoning about uniformity.
For bitvectors,
\newcommand\appland{\applkw{\&\&}}
\begin{itemize}
\item Let $\ber^n 1/2$ be the uniform distribution on boolean-valued $n$-tuples.
\item Let $\flip^n~p$ be the $n$-ary generalization of $\flip$ that produces $n$-tuples of i.i.d. $\ber p$
       random variables, with the evident semantics.
\item If $X$ is a random variable valued in boolean $n$-tuples, let $\bigoplus X$ 
         be the random variable given by the $\oplus$ of all $n$ components.
\item Let $\appland^n$ and $\oplus^n$ be the lifting of boolean $\appland$ and $\oplus$ to $n$-tuples.
\item We will make use of algebraic properties of the bitvector xor operator $\oplus^n$ throughout;
  in particular the property that $x\oplus^n -$ is invertible.
\end{itemize}
Next, we import the requisite probability theory facts as derived rules.
For clarity of exposition, we suppress components of Lilac's semantic model 
(like underlying probability spaces, the random substitution, and the deterministic substitution)
in the proofs of these rules in favor of a presentation that more closely mirrors textbook probability.
The first few facts concern uniformity of random bitvectors:

\begin{lemma} \label{lem:unif-bij}
  If $X$ is a random $n$-bitvector and $f:\sembr\boolty^n\to\sembr\boolty^n$ a bijection then
  \[
    X \sim \ber^n 1/2 \vdash f(X) \sim \ber^n1/2.
  \]
\end{lemma}
\begin{proof}
  We have $\Pr[f(X) = x] = \Pr[X = f^{-1}(x)] = 1/2^n$ for all $x$.
\end{proof}

\begin{lemma}\label{lem:unif-fuse}
  If $X$ is a random $n$-bitvector and $Y$ a random $m$-bitvector then
  \[
    X\sim \ber^n 1/2 ~~*~~ Y \sim \ber^m 1/2 \vdashiff (X,Y) \sim \ber^{m+n}1/2
  \]
\end{lemma}
\begin{proof}
  Calculation gives
  \[
    \Pr[(X,Y)=(x,y)]=\Pr[X=x,Y=y]\stackrel{(a)}=\Pr[X=x]\Pr[Y=y] = (1/2^m)(1/2^n) = 1/2^{m+n}\]
    for all $x,y$ as desired. Equation $(a)$ follows from independence of $X$ and $Y$.
\end{proof}
This next lemma encodes the key fact of probability theory underlying the perfect secrecy of the examples
we will consider in the next section. Intuitively, it states that any random variable which is ``conditionally 
uniformly distributed'' (that is, uniformly distributed conditional on some other random variable) is
uniformly distributed proper.
\begin{lemma} \label{lem:indep-unif}
  If $X$ is a random variable taking on finitely many values\footnote{%
  Though we expect this restriction can be lifted, $\operatorname{cod}(X)$ finite suffices for our examples.}
  and $Y$ a random $n$-bitvector then
  \[
    \own X \land \D_{x\gets X} (Y \sim \ber^n 1/2) \hspace{0.5em}\vdash\hspace{0.5em} \own X ~~*~~ Y \sim \ber^n1/2
  \]
\end{lemma}
\begin{proof}
  It suffices to show $X$ and $Y$ are independent and $Y$ uniform.
  This amounts to showing the equality $\Pr[X=x,Y=y] = \Pr[X=x]/2^n$ for all $x,y$.
  By assumption $Y$ is uniformly distributed conditional on $X$, so
  \[ \Pr[X=x,Y=y] = \Pr[Y=y\mid X=x]\Pr[X=x] = \frac1{2^n}\Pr[X=x] \]
  as desired.
\end{proof}
To avoid verbosity, we use the abbreviation
$\displaystyle \Down_{x\gets X} P := \own X \land \D_{x\gets X} P$.\\
\noindent Thus the entailment given by Lemma~\ref{lem:indep-unif} can be written 
$\displaystyle\Down_{x\gets X} (Y\sim\ber^n 1/2) \hspace{0.5em}\vdash\hspace{0.5em} \own X ~~*~~ Y \sim \ber^n1/2$.\\
We will also frequently make use of the following variant of \DIndep:
\begin{lemma}\label{lem:dindep-variant}
  The following entailment holds:
    \[  \own X ~~*~~ P  \hspace{0.5em}\vdash\hspace{0.5em} \Down_{x\gets X} P. \]
\end{lemma}
\begin{proof} By the following chain:
  \begin{align*}
    &\own X * P \\
    &\vdash (\own X * P)\land (\own X * P) & \text{idempotency of }\land \\
    &\vdash \own X\land (\own X * P) & \text{drop a conjunct} \\
    &\vdash \own X\land \D_{x\gets X} P & \DIndep  \\
    &\vdash \Down_{x\gets X} P & \text{definition of }\Down
  \end{align*}
\end{proof}
The final derived rule we will make use of encodes the probability-theoretic fact that,
when proving an assertion of the form $X\sim\mu ~~*~~ P$,
one can first establish that $X$ has distribution $\mu$,
and then separately establish independence of $X$ from other random variables.
\begin{lemma}\label{lem:separate-own-dist}
  Let $X$ be a random variable, $\mu$ a distribution, and $P$ a proposition.
  \[
    (X\sim \mu)\land (\own X ~~*~~ P)\vdash (X\sim \mu) ~~*~~ P
  \]
\end{lemma}
\begin{proof}
  Suppose $X$ is distributed as $\mu$ with respect to probability space $\mathcal P_X$,
  that $P$ holds in space $\mathcal P$,
  and that there exists a space $\mathcal P_X'$ independent of $\mathcal P$ for which $X$ is $\mathcal P_X'$-measurable.
  Let $\mathcal Q$ be the pullback $\sigma$-algebra of $X$.
  We have that $\mathcal P_X\sqsupseteq \mathcal Q\sqsubseteq\mathcal P_X'$ and that
  $\mathcal P_X$ and $\mathcal P_X'$ agree on $\mathcal Q$. Thus, $\mathcal Q$ is independent of $\mathcal P$
  and the composite $\mathcal Q\pdot \mathcal P$ witnesses $X\sim\mu~~*~~ P$ as desired.
\end{proof}

\subsection{Verification of one-time pad example} \label{sec:one-time-pad}

The one-time pad protocol is modelled by the following probabilistic program,
parameterized by a constant $m$ representing the message being encrypted:

\begin{equation}
  {\small
  \tag{\textsc{OneTimePad}}
\begin{aligned}
  &\appllet K\appleq\flip~1/2 \applin\\
  &\appllet C\appleq\applret(K\oplus m) \applin\\
  &\applret C
\end{aligned}
\label{prog:onetimepad}
  }
\end{equation}

\subsubsection{Uniformity}

Uniformity is specified by the triple
$
  \triple\top{\ref{prog:onetimepad}(m)}C{C\sim\flip~1/2}
$.\\
This can be established by the following annotation:
\begin{align*}
  &\assert{\{\top\}} \\
  &\hspace{1em}\appllet K\appleq\flip~1/2 \applin\\
  &\assert{\{K\sim \ber1/2\}} \\
  &\hspace{1em}\appllet C\appleq\applret(K\oplus m) \applin\\
  &\assert{\{K\sim \ber1/2~~*~~C\asequal K\oplus m\}} \\
  &\hspace{1em}\applret C\\
  &\assert{\{C.~K\sim \ber1/2~~*~~C\asequal K\oplus m\}} & \\
  &\assert{\{C.~K\sim \ber1/2~~*~~C\oplus m\asequal K\}} & \text{rearranging the equality} \\
  &\assert{\{C.~(C\oplus m)\sim \ber1/2\}} & \text{substituting away }K \\
  &\assert{\{C.~C\sim \ber1/2\}} & \text{Lemma~\ref{lem:unif-bij}} \\
\end{align*}

\subsubsection{Input independence}

Input independence is specified by the triple
\[\triple{\own M}{\ref{prog:onetimepad}(M)}C{\own M ~~*~~ C\sim\ber~1/2}.\]
This can be established by the following annotation:

\begin{align*}
  &\assert{\{\own M\}} \\
  &\hspace{1em}\appllet K\appleq \flip~1/2 \applin\\
  &\assert{\{\own M ~~*~~ K\sim \ber1/2\}} \\
  &\hspace{1em}\appllet C\appleq \applret(K\oplus M) \applin\\
  &\assert{\{\own M ~~*~~ K\sim \ber1/2~~*~~C\asequal K\oplus M\}} \\
  &\hspace{1em}\applret C\\
  &\assert{\{C.~\own M ~~*~~ K\sim \ber1/2~~*~~C\asequal K\oplus M\}} \\
  &\assert{\{C.~\own M ~~*~~ K\sim \ber1/2~~*~~C\oplus M\asequal K\}} & \text{rearranging the equality} \\
  &\assert{\{C.~\own M ~~*~~ (C\oplus M)\sim \ber1/2\}} & \text{substituting away }K \\
  &\assertc{C.~\Down_{m\gets M} ((C\oplus M)\sim \ber1/2)} & \text{Lemma~\ref{lem:dindep-variant}} \\
  &\assertc{C.~\Down_{m\gets M} ((C\oplus m)\sim \ber1/2)} & \text{replacing }M\text{ with }m \\
  &\assertc{C.~\Down_{m\gets M} (C\sim \ber1/2)} & \text{Lemma~\ref{lem:unif-bij}} \\
  &\assertc{C.~\own M ~~*~~ C\sim \ber1/2} & \text{Lemma~\ref{lem:indep-unif}} \\
\end{align*}

\subsection{Verification of private information retrieval example} \label{sec:pir}

Let $i$ be an $n$-tuple of boolean values with only a single component set to $\vtrue$.
The private information retrieval protocol is modelled by the following probabilistic program,
parameterized by $i$:
\begin{equation}
  {\small
  \tag{\textsc{PrivateInformationRetrieval}}
\begin{aligned}
  &\appllet Q_0\appleq \flip^n~1/2 \applin\\
  &\appllet Q_1\appleq Q_0 \oplus^n i \applin\\
  &\appllet A_0\appleq Q_0 \appland^n d \applin\\
  &\appllet A_1\appleq Q_1 \appland^n d \applin\\
  &\appllet R_0\appleq \applret\left(\bigoplus A_0\right) \applin \\
  &\appllet R_1\appleq \applret\left(\bigoplus R_1\right) \applin \\
  &\appllet R\appleq \applret(R_0\oplus R_1) \applin \\
  &\applret (Q_0, Q_1,A_0, A_1, R_0, R_1, R)
\end{aligned}
\label{prog:pir}
  }
\end{equation}

\subsubsection{Uniformity}

Uniformity is specified by the triple
\[
  \triple\top{\ref{prog:pir}(i)}{(Q_0, Q_1, A_0, A_1, R_0, R_1, R)}
    {~ Q_0 \sim \ber^n1/2 ~~\land~~ Q_1 \sim \ber^n 1/2}.
\]
This can be established by the following annotation:
{\small
\begin{align*}
  &\assert{\{\top\}} \\
  &\hspace{1em}\appllet Q_0\appleq \flip^n~1/2 \applin\\
  &\assert{\{Q_0 \sim \ber^n1/2\}} \\
  &\hspace{1em}\appllet Q_1\appleq Q_0 \oplus^n i \applin\\
  &\assert{\{Q_0 \sim \ber^n1/2 ~~*~~ Q_1 \asequal Q_0 \oplus^n i\}} \\
  &\hspace{1em}\appllet A_0\appleq Q_0 \appland^n d \applin\\
  &\assert{\{Q_0 \sim \ber^n1/2 ~~*~~ Q_1 \asequal Q_0 \oplus^n i ~~*~~ A_0 \asequal Q_0 \appland^n d \}} \\
  &\hspace{1em}\appllet A_1\appleq Q_1 \appland^n d \applin\\
  &\assert{\{Q_0 \sim \ber^n1/2 ~~*~~ Q_1 \asequal Q_0 \oplus^n i ~~*~~ A_0 \asequal Q_0 \appland^n d ~~*~~
     A_1 \asequal Q_1 \appland^n d \}} \\
  &\hspace{1em}\appllet R_0\appleq \applret\left(\bigoplus A_0\right) \applin \\
  &\assertc{
     Q_0 \sim \ber^n1/2 ~~*~~ Q_1 \asequal Q_0 \oplus^n i ~~*~~ A_0 \asequal Q_0 \appland^n d ~~*~~ A_1\asequal Q_1\appland^n d~~*~~
     R_0=\bigoplus A_0
    } \\
  &\hspace{1em}\appllet R_1\appleq \applret\left(\bigoplus R_1\right) \applin \\
  &\assertc{
     Q_0 \sim \ber^n1/2 ~~*~~ Q_1 \asequal Q_0 \oplus^n i ~~*~~ A_0 \asequal Q_0 \appland^n d ~~*~~ A_1\asequal Q_1\appland^n d~~*~~
     R_0=\bigoplus A_0 ~~*~~ R_1=\bigoplus A_1
    } \\
  &\hspace{1em}\appllet R\appleq \applret(R_0\oplus R_1) \applin \\
  &\assertc{\begin{aligned}
     &Q_0 \sim \ber^n1/2 ~~*~~ Q_1 \asequal Q_0 \oplus^n i ~~*~~ A_0 \asequal Q_0 \appland^n d ~~*~~ A_1\asequal Q_1\appland^n d~~*~~\\
     &R_0=\bigoplus A_0 ~~*~~ R_1=\bigoplus A_1 ~~*~~ R\asequal R_0 \oplus R_1
    \end{aligned}} \\
  &\hspace{1em}\applret (Q_0, Q_1,A_0, A_1, R_0, R_1, R) \\
  &\assertc{(Q_0, Q_1,A_0, A_1, R_0, R_1, R).\hspace{0.5em}\begin{aligned}
     &Q_0 \sim \ber^n1/2 ~~*~~ Q_1 \asequal Q_0 \oplus^n i ~~*~~ A_0 \asequal Q_0 \appland^n d ~~*~~ A_1\asequal Q_1\appland^n d~~*~~\\
     &R_0=\bigoplus A_0 ~~*~~ R_1=\bigoplus A_1 ~~*~~ R\asequal R_0 \oplus R_1
    \end{aligned}} \\[-0.75em]
  &\hspace{13em}\assert{\vdots} \\[-0.25em]
  &\assertc{(Q_0, Q_1,A_0, A_1, R_0, R_1, R).~~
     Q_0 \sim \ber^n1/2 ~~\land~~ Q_1 \sim \ber^n1/2} \\
\end{align*}}
The final entailment (hidden by the vertical ellipses) can be established as follows:
\begin{align*}
  &\left(\begin{aligned}
     &Q_0 \sim \ber^n1/2 ~~*~~ Q_1 \asequal Q_0 \oplus^n i ~~*~~ A_0 \asequal Q_0 \appland^n d ~~*~~ \\
     &R_0=\bigoplus A_0 ~~*~~ R_1=\bigoplus A_1 ~~*~~ R\asequal R_0 \oplus R_1
    \end{aligned}\right)\\
  &\vdash Q_0 \sim \ber^n1/2 ~~*~~ Q_1 \asequal Q_0 \oplus^n i & \text{dropping some conjuncts}\\
  &\vdash (Q_0 \sim \ber^n1/2 ~~*~~ Q_1 \asequal Q_0 \oplus^n i)\land
     (Q_0 \sim \ber^n1/2 ~~*~~ Q_1 \asequal Q_0 \oplus^n i)
     & \text{by }P\vdash P\land P\\
  &\vdash (Q_0 \sim \ber^n1/2)\land
     (Q_0 \sim \ber^n1/2 ~~*~~ Q_1 \asequal Q_0 \oplus^n i)
     & \text{weakening first conjunct}\\
  &\vdash (Q_0 \sim \ber^n1/2) ~~\land~~ (Q_0 \sim \ber^n1/2 ~~*~~ Q_1 \oplus^n i \asequal Q_0)
    & \text{rearranging the equality}\\
  &\vdash (Q_0 \sim \ber^n1/2) ~~\land~~ ((Q_1\oplus^n i) \sim \ber^n1/2)
    & \text{substituting away }Q_0\\
  &\vdash (Q_0 \sim \ber^n1/2) ~~\land~~ (Q_1 \sim \ber^n1/2)
    & \text{Lemma~\ref{lem:unif-bij}}
\end{align*}

\subsubsection{Input independence}

Let $I$ be a random $n$-bitvector.\footnote{%
Following \citet{barthe2019probabilistic},  we don't even need require that $I$
always have only a single component set to $\vtrue$
to establish input independence.}
Input independence is specified by the triple
\[
  \triple{\own I}{\ref{prog:pir}(I)}C{(\own I * \own Q_0) \land (\own I * \own Q_1)}.
\]
This can be established by the following annotation:
{\small
\begin{align*}
  &\assert{\{\own I\}} \\
  &\hspace{1em}\appllet Q_0\appleq \flip^n~1/2 \applin\\
  &\assert{\{\own I ~~*~~ Q_0 \sim \ber^n1/2\}} \\
  &\hspace{1em}\appllet Q_1\appleq Q_0 \oplus^n I \applin\\
  &\assert{\{\own I ~~*~~ Q_0 \sim \ber^n1/2 ~~*~~ Q_1 \asequal Q_0 \oplus^n I\}} \\
  &\hspace{1em}\appllet A_0\appleq Q_0 \appland^n d \applin\\
  &\assert{\{\own I ~~*~~ Q_0 \sim \ber^n1/2 ~~*~~ Q_1 \asequal Q_0 \oplus^n I ~~*~~ A_0 \asequal Q_0 \appland^n d \}} \\
  &\hspace{1em}\appllet A_1\appleq Q_1 \appland^n d \applin\\
  &\assert{\{\own I ~~*~~ Q_0 \sim \ber^n1/2 ~~*~~ Q_1 \asequal Q_0 \oplus^n I ~~*~~ A_0 \asequal Q_0 \appland^n d ~~*~~
     A_1 \asequal Q_1 \appland^n d \}} \\
  &\hspace{1em}\appllet R_0\appleq \applret\left(\bigoplus A_0\right) \applin \\
  &\assertc{
     \own I ~~*~~ Q_0 \sim \ber^n1/2 ~~*~~ Q_1 \asequal Q_0 \oplus^n I ~~*~~ A_0 \asequal Q_0 \appland^n d ~~*~~ A_1\asequal Q_1\appland^n d ~~*~~
     R_0=\bigoplus A_0
    } \\
  &\hspace{1em}\appllet R_1\appleq \applret\left(\bigoplus R_1\right) \applin \\
  &\assertc{\begin{aligned}
     &\own I ~~*~~ Q_0 \sim \ber^n1/2 ~~*~~ Q_1 \asequal Q_0 \oplus^n I ~~*~~ A_0 \asequal Q_0 \appland^n d ~~*~~ A_1\asequal Q_1\appland^n d ~~*\\
     &R_0=\bigoplus A_0 ~~*~~ R_1=\bigoplus A_1 ~~*~~ R\asequal R_0 \oplus R_1
    \end{aligned}} \\
  &\hspace{1em}\appllet R\appleq \applret(R_0\oplus R_1) \applin \\
  &\assertc{\begin{aligned}
     &\own I ~~*~~ Q_0 \sim \ber^n1/2 ~~*~~ Q_1 \asequal Q_0 \oplus^n I ~~*~~ A_0 \asequal Q_0 \appland^n d ~~*~~ A_1\asequal Q_1\appland^n d ~~*\\
     &R_0=\bigoplus A_0 ~~*~~ R_1=\bigoplus A_1 ~~*~~ R\asequal R_0 \oplus R_1
    \end{aligned}} \\
  &\hspace{1em}\applret (Q_0, Q_1,A_0, A_1, R_0, R_1, R) \\
  &\assertc{(Q_0, Q_1,A_0, A_1, R_0, R_1, R).\hspace{0.5em}\begin{aligned}
     &\own I ~~*~~ Q_0 \sim \ber^n1/2 ~~*~~ Q_1 \asequal Q_0 \oplus^n I ~~*~~ A_0 \asequal Q_0 \appland^n d ~~* \\
     &A_1\asequal Q_1\appland^n d ~~*~~ R_0=\bigoplus A_0 ~~*~~ R_1=\bigoplus A_1 ~~*~~ R\asequal R_0 \oplus R_1
    \end{aligned}} \\
  &\assertc{(Q_0, Q_1,A_0, A_1, R_0, R_1, R).~~
     \own I ~~*~~ Q_0 \sim \ber^n1/2 ~~*~~ Q_1 \asequal Q_0 \oplus^n I
  } \\[-0.5em]
  &\hspace{18em}\assert{\vdots}\\
  &\assertc{(Q_0, Q_1,A_0, A_1, R_0, R_1, R).~~
     (\own I * \own Q_0) \land
     (\own I * \own Q_1)
  }
\end{align*}
}%
The final entailment (hidden by the vertical ellipses) can be established as follows.
The left conjunct of the postcondition follows from $Q_0\sim \ber^n 1/2\vdash \own Q_0$ and dropping the equality
$Q_1\asequal Q_0\oplus^n I$.
The right conjunct is established by the following chain of entailments:
\begin{align*}
  & \own I ~~*~~ Q_0 \sim \ber^n1/2~~*~~ Q_1\asequal Q_0\oplus^n I \\
  & \vdash\Down_{i\gets I} (Q_0 \sim \ber^n1/2~~*~~ Q_1\asequal Q_0\oplus^n I) & \text{Lemma~\ref{lem:dindep-variant}} \\
  & \vdash\Down_{i\gets I} (Q_0 \sim \ber^n1/2~~*~~ Q_1\asequal Q_0\oplus^n i) & \text{substituting }i\text{ for }I \\
  & \vdash\Down_{i\gets I} (Q_0 \sim \ber^n1/2~~*~~ (Q_1\oplus^n i)\asequal Q_0) & \text{rearranging the equality} \\
  & \vdash\Down_{i\gets I} ((Q_1\oplus^n i) \sim \ber^n1/2) & \text{substituting away }Q_0 \\
  & \vdash\Down_{i\gets I} (Q_1 \sim \ber^n1/2) & \text{Lemma~\ref{lem:unif-bij}} \\
  & \vdash\own I ~~*~~ Q_1 \sim \ber^n1/2 & \text{Lemma~\ref{lem:indep-unif}} \\
  & \vdash\own I ~~*~~ \own Q_1 & \text{by }X\sim\mu\vdash\own X
\end{align*}

\subsection{Verification of oblivious transfer example} \label{sec:oblivious-transfer}

The oblivious transfer protocol is modelled by the following probabilistic program:
\begin{equation}
  {\small
  \tag{\textsc{ObliviousTransfer}}
\begin{aligned}
  &\appllet R_0\appleq \flip^n~1/2 \applin\\
  &\appllet R_1\appleq \flip^n~1/2 \applin\\
  &\appllet D\appleq \flip~1/2 \applin\\
  &\appllet (R_D, R_{\neg D})\appleq \applite{D}{\applret(R_1,R_0)}{\applret(R_0,R_1)} \applin\\
  &\appllet E\appleq \applret(c\oplus D) \applin\\
  &\appllet (F_0,F_1)\appleq \applite{E}{\applret(m_0\oplus R_1,m_1\oplus R_0)}{\applret(m_0\oplus R_0,m_1\oplus R_1)} \applin\\
  &\appllet (m_c,F_{\neg c}) \appleq  \applite{c}{\applret (F_1\oplus R_D,m_0\oplus R_{\neg D})}{\applret(F_0\oplus R_D,m_1\oplus R_{\neg D})} \applin\\
  &\applret (R_0, R_1, D, R_D, R_{\neg D}, E, F_0, F_1, m_c, F_{\neg c})
\end{aligned}
\label{prog:ot}
  }
\end{equation}
This program is parameterized by the two messages $m_0$ and $m_1$ on offer and the bit $c$ encoding
the receiver's choice.

\subsubsection{Uniformity}
Uniformity is specified by the triple
\[
  \lrtriple\top{\ref{prog:ot}(c,m_0,m_1)}{
      \left(\begin{aligned}
        &R_0, R_1, D, R_D, R_{\neg D}, \\
        &E, F_0, F_1, m_c, F_{\neg c}
      \end{aligned}\right)}
    {~ \begin{aligned}
      &((R_0,R_1)\sim \ber^{2n}1/2 ~~*~~ E\sim\ber1/2) ~~\land\\
      &(D\sim\ber1/2 ~~*~~ (R_D,F_{\neg c})\sim \ber^{2n}1/2)
    \end{aligned}}.
\]
The next page gives an annotated program that establishes this specification.
\clearpage
{\small
\begin{align*}
  &\assertc{\top} \\
  &\hspace{1em}\appllet R_0\appleq \flip^n~1/2 \applin\\
  &\assertc{R_0 \sim \ber^n 1/2} \\
  &\hspace{1em}\appllet R_1\appleq \flip^n~1/2 \applin\\
  &\assertc{R_0 \sim \ber^n 1/2 ~~*~~ R_1\sim \ber^n 1/2} \\
  &\hspace{1em}\appllet D\appleq \flip~1/2 \applin\\
  &\assertc{R_0 \sim \ber^n 1/2 ~~*~~ R_1\sim \ber^n 1/2 ~~*~~ D\sim\ber1/2} \\
  &\hspace{1em}\appllet (R_D, R_{\neg D})\appleq \applite{D}{\applret(R_1,R_0)}{\applret(R_0,R_1)} \applin\\
  &\assertc{\begin{aligned}
    &R_0 \sim \ber^n 1/2 ~~*~~ R_1\sim \ber^n 1/2 ~~*~~ D\sim\ber1/2 ~~*~~ \\
    &R_D \asequal (\ite D{R_1}{R_0}) ~~*~~ R_{\neg D}\asequal (\ite D{R_0}{R_1})
    \end{aligned}} \\
  &\hspace{1em}\appllet E\appleq \applret(c\oplus D) \applin\\
  &\assertc{\begin{aligned}
    &R_0 \sim \ber^n 1/2 ~~*~~ R_1\sim \ber^n 1/2 ~~*~~ D\sim\ber1/2 ~~*~~ \\
    &R_D \asequal (\ite D{R_1}{R_0}) ~~*~~ R_{\neg D}\asequal (\ite D{R_0}{R_1}) ~~*~~ E\asequal c\oplus D \\
    \end{aligned}} \\
  &\hspace{1em}\appllet (F_0,F_1)\appleq \applite{E}{\applret(m_0\oplus R_1,m_1\oplus R_0)}{\applret(m_0\oplus R_0,m_1\oplus R_1)} \applin\\
  &\assertc{\begin{aligned}
    &R_0 \sim \ber^n 1/2 ~~*~~ R_1\sim \ber^n 1/2 ~~*~~ D\sim\ber1/2 ~~*~~ \\
    &R_D \asequal (\ite D{R_1}{R_0}) ~~*~~ R_{\neg D}\asequal (\ite D{R_0}{R_1}) ~~*~~ E\asequal c\oplus D ~~*~~ \\
    &F_0 \asequal (\ite E{m_0\oplus R_1}{m_1\oplus R_0}) ~~*~~ F_1\asequal (\ite D{m_0\oplus R_0}{m_1\oplus R_1}) \\
    \end{aligned}} \\
  &\hspace{1em}\appllet (m_c,F_{\neg c}) \appleq  \applite{c}{\applret (F_1\oplus R_D,m_0\oplus R_{\neg D})}{\applret(F_0\oplus R_D,m_1\oplus R_{\neg D})} \applin\\
  &\assert{\underbrace{\assertc{\begin{aligned}
    &R_0 \sim \ber^n 1/2 ~~*~~ R_1\sim \ber^n 1/2 ~~*~~ D\sim\ber1/2 ~~*~~ \\
    &R_D \asequal (\ite D{R_1}{R_0}) ~~*~~ R_{\neg D}\asequal (\ite D{R_0}{R_1}) ~~*~~ E\asequal c\oplus D ~~*~~ \\
    &F_0 \asequal (\ite E{m_0\oplus R_1}{m_1\oplus R_0}) ~~*~~ F_1\asequal (\ite D{m_0\oplus R_0}{m_1\oplus R_1}) ~~*~~ \\
    &m_c \asequal (\ite c{F_1\oplus R_D}{F_0\oplus R_D}) ~~*~~ F_{\neg c}\asequal (\ite c{m_0\oplus R_{\neg D}}{m_1\oplus R_{\neg D}}) \\
    \end{aligned}}}_P} \\
  &\hspace{1em}\applret (R_0, R_1, D, R_D, R_{\neg D}, E, F_0, F_1, m_c, F_{\neg c})\\
  &\assertc{(R_0, R_1, D, R_D, R_{\neg D}, E, F_0, F_1, m_c, F_{\neg c}).~P} \\[1em]
  &\hspace{9em}\assert{\vdots}\\
  &\assert{{\bigg\{}(R_0, R_1, D, R_D, R_{\neg D}, E, F_0, F_1, m_c, F_{\neg c}).~
    \begin{aligned}
      &\overbrace{((R_0,R_1)\sim \ber^{2n}1/2 ~~*~~ E\sim\ber1/2)}^Q ~~\land\\
      &\underbrace{(D\sim\ber1/2 ~~*~~ (R_D,F_{\neg c})\sim \ber^{2n}1/2)}_R
    \end{aligned} {\bigg\}}} \\
\end{align*}
}%
The final entailment (hidden by the vertical ellipses), abbreviated $P\vdash Q\land R$,
can be established as follows.
First, to show $P\vdash Q$,
{\small\begin{align*}
  &P \\
  &\vdash R_0 \sim \ber^n1/2 ~~*~~ R_1\sim \ber^n 1/2 ~~*~~ D\sim \ber1/2 ~~*~~ E\asequal c\oplus D
    & \text{dropping conjuncts} \\
  &\vdash R_0 \sim \ber^n1/2 ~~*~~ R_1\sim \ber^n 1/2 ~~*~~ D\sim \ber1/2 ~~*~~ E\oplus c\asequal D
    & \text{rearranging the equality} \\
  &\vdash R_0 \sim \ber^n1/2 ~~*~~ R_1\sim \ber^n 1/2 ~~*~~ (E\oplus c)\sim \ber1/2
    & \text{substituting away }D \\
  &\vdash R_0 \sim \ber^n1/2 ~~*~~ R_1\sim \ber^n 1/2 ~~*~~ E\sim \ber1/2
    & \text{Lemma~\ref{lem:unif-bij}} \\
  &\vdash (R_0,R_1) \sim \ber^{2n}1/2 ~~*~~ E\sim \ber1/2
    & \text{Lemma~\ref{lem:unif-fuse}} \\
  &= Q
\end{align*}}
Second, $P\vdash R$ can be established by the following chain:
{\small
\begin{align*}
  &P \\
  &\vdash {\left({\begin{aligned}
    &R_0 \sim \ber^n 1/2 ~~*~~ R_1\sim \ber^n 1/2 ~~*~~ D\sim\ber1/2 ~~*~~ \\
    &R_D \asequal (\ite D{R_1}{R_0}) ~~*~~ R_{\neg D}\asequal (\ite D{R_0}{R_1})~~*~~\\
    &F_{\neg c}\asequal (\ite c{m_0\oplus R_{\neg D}}{m_1\oplus R_{\neg D}}) 
    \end{aligned}}\right)} & \text{dropping some conjuncts} \\
  &\vdash \dots \\
  &\vdash {\left({\begin{aligned}
    &R_D \sim \ber^n 1/2 ~~*~~ R_{\neg D}\sim \ber^n 1/2 ~~*~~ D\sim\ber1/2 ~~*~~ \\
    &F_{\neg c}\asequal (\ite c{m_0\oplus R_{\neg D}}{m_1\oplus R_{\neg D}}) \\
    \end{aligned}}\right)}  \\
  &\vdash {\left({\begin{aligned}
    &R_D \sim \ber^n 1/2 ~~*~~ R_{\neg D}\sim \ber^n 1/2 ~~*~~ D\sim\ber1/2 ~~*~~ \\
    &F_{\neg c}\asequal (\ite c{m_0}{m_1})\oplus R_{\neg D} \\
    \end{aligned}}\right)} 
    & \text{commuting conversion}\\
  &\vdash {\left({\begin{aligned}
    &R_D \sim \ber^n 1/2 ~~*~~ R_{\neg D}\sim \ber^n 1/2 ~~*~~ D\sim\ber1/2~~*~~ \\
    &(\ite c{m_0}{m_1})\oplus F_{\neg c}\asequal R_{\neg D} \\
    \end{aligned}}\right)} 
    & \text{property of }\oplus\\
  &\vdash {\left({\begin{aligned}
    &R_D \sim \ber^n 1/2 ~~*~~ D\sim\ber1/2 ~~*~~ \\
    &((\ite c{m_0}{m_1})\oplus F_{\neg c})\sim \ber^n 1/2
    \end{aligned}}\right)} 
    & \text{substitute away }R_{\neg D}\\
  &\vdash {
    R_D \sim \ber^n 1/2 ~~*~~ D\sim\ber1/2  ~~*~~
    F_{\neg c}\sim \ber^n 1/2
    }
    & \text{Lemma~\ref{lem:unif-bij}} \\
  &\vdash D\sim\ber1/2 ~~*~~ (R_D,F_{\neg c})\sim \ber^{2n}1/2 & \text{Lemma~\ref{lem:unif-fuse}}\\
  &= R
\end{align*}}%
The entailment hidden by ellipses
establishes the mutual independence of $D$, $R_D$, and $R_{\neg D}$,
the key property that perfect secrecy hinges on. The proof goes by case analysis on $D$,
and is shown on the next page.
\clearpage
{\small\begin{align*}
  &{\left({\begin{aligned}
    &R_0 \sim \ber^n 1/2 ~~*~~ R_1\sim \ber^n 1/2 ~~*~~ D\sim\ber1/2 ~~*~~ \\
    &R_D \asequal (\ite D{R_1}{R_0}) ~~*~~ R_{\neg D}\asequal (\ite D{R_0}{R_1})
    \end{aligned}}\right)} \\
  &\vdash \underbrace{\left({\begin{aligned}
    &(R_0, R_1) \sim \ber^{2n} 1/2 ~~*~~ D\sim\ber1/2 ~~*~~ \\
    &R_D \asequal (\ite D{R_1}{R_0}) ~~*~~ R_{\neg D}\asequal (\ite D{R_0}{R_1})
    \end{aligned}}\right)}_{S[D\sim\ber 1/2]} & \text{Lemma~\ref{lem:unif-fuse}} \\
  &\vdash S[D\sim\ber1/2] \land S[D\sim\ber 1/2] & \text{property of }\land \\
  &\vdash S[D\sim\ber 1/2] \land (D\sim\ber 1/2) & \text{drop conjuncts}\\
  &\vdash S[\own D] \land (D\sim\ber 1/2) & X\sim\mu\vdash \own X\\
  &\vdash 
    {\Down_{d\gets D} \left(
    (R_0, R_1) \sim \ber^{2n} 1/2 ~~*~~ 
    R_D \asequal (\ite D{R_1}{R_0}) ~~*~~ \dots
    \right)}
     \land (D\sim\ber 1/2) & \text{Lemma~\ref{lem:dindep-variant}} \\
  &\vdash 
    {\Down_{d\gets D} \underbrace{\left({\begin{aligned}
    &(R_0,R_1) \sim \ber^{2n}1/2 ~~*~~\\
    &R_D \asequal (\ite d{R_1}{R_0}) ~~*~~ R_{\neg D}\asequal \dots
    \end{aligned}}\right)}_{Q(d)}} 
     \land (D\sim\ber 1/2) & \text{replace }D\textrm{ with }d \\
  &\vdash 
    \Down_{d\gets D} \left(Q(\vtrue) \land Q(\vfalse)\right) 
     \land (D\sim\ber 1/2) & \text{cases on }d \\
  &\vdash 
    \Down_{d\gets D} \bigg(\begin{aligned}
        &\overbrace{\left(
          (R_0,R_1) \sim \ber^{2n} ~~*~~
          R_D \asequal R_1 ~~*~~ R_{\neg D}\asequal R_0\right)
        }^{Q(\vtrue)} \land \\
        &\underbrace{\left(
          (R_0,R_1) \sim \ber^{2n} 1/2 ~~*~~
          R_D \asequal R_0 ~~*~~ R_{\neg D}\asequal R_1\right)
        }_{Q(\vfalse)}
    \end{aligned}\bigg) 
     \land(D\sim\ber 1/2) & \text{unfold }Q \\
  &\vdash 
    \Down_{d\gets D}\bigg(
        ((R_{\neg D},R_D) \sim \ber^{2n} 1/2) \land 
        ((R_D,R_{\neg D}) \sim \ber^{2n} 1/2)\bigg)
     \land (D\sim\ber 1/2) & \text{substitute} \\
  &\vdash 
    \Down_{d\gets D}\bigg(\begin{aligned}
        &(R_{\neg D}\sim\ber^n1/2 ~~*~~ R_D\sim\ber^n 1/2) \land  \\
        &(R_D \sim \ber^{n} 1/2 ~~*~~ R_{\neg D}\sim\ber^n 1/2)\end{aligned}\bigg)
     \land (D\sim\ber 1/2) & \text{Lemma~\ref{lem:unif-fuse}} \\
  &\vdash 
    \Down_{d\gets D} (R_{\neg D} \sim \ber^n 1/2 ~~*~~ R_D\sim \ber^n 1/2)
     \land (D\sim\ber 1/2) & \text{by }P\land P\vdash P \\
  &\vdash 
    \Down_{d\gets D} ((R_D,R_{\neg D}) \sim \ber^{2n} 1/2)
     \land (D\sim\ber 1/2) & \text{Lemma~\ref{lem:unif-fuse}} \\
  &\vdash 
    (\own D ~~*~~ R_{\neg D} \sim \ber^n 1/2 ~~*~~ R_D\sim \ber^n 1/2)
     \land (D\sim\ber 1/2) & \text{Lemma~\ref{lem:indep-unif}} \\
  &\vdash 
    D\sim\ber1/2 ~~*~~ R_{\neg D} \sim \ber^n 1/2 ~~*~~ R_D\sim \ber^n 1/2
     & \text{Lemma~\ref{lem:separate-own-dist}}
\end{align*}}

\subsubsection{Input independence}

Following \citet{barthe2019probabilistic}, the following Hoare triple specifies input
independence for \ref{prog:ot}:
\begin{align*}
  &\assertc{\own (C, M_0, M_1)} \\
  &\hspace{1em}{\ref{prog:ot}(M_0,M_1,C)} \\
  &\assertc{(R_0, R_1, D, R_D, R_{\neg D}, E, F_0, F_1, M_C, M_{1-C}, F_{\neg C}).~
    \begin{aligned}
      &(\own C ~~*~~ \own (R_0, R_1, E)) ~~\land \\
      &(\own M_{1-C} ~~*~~ \own (D, R_D, F_0, F_1)) \\
    \end{aligned}}
\end{align*}
In the postcondition, $M_{1-C}$ denotes the random variable $(\ite{C}{M_0}{M_1})$.

\citet{barthe2019probabilistic} observe that the proof gets stuck, mentioning
that even an informal proof sketch does not seem easy.
We show that this triple is in fact
impossible to establish by giving an explicit counterexample. This takes the form of 
three random variables $C$, $M_0$, and $M_1$ such that
the postcondition fails.
In particular, we choose $C = 0$ and $M_0 = M_1 = M$ for some uniformly-distributed $M$.
Now the triple reads
\begin{align*}
  &\assertc{\own (C, M, M)} \\
  &\hspace{1em}{\ref{prog:ot}(M, M, C)} \\
  &\assertc{(R_0, R_1, D, R_D, R_{\neg D}, E, F_0, F_1, M_C, M_{1-C}, F_{\neg C}).~
    \begin{aligned}
      &(\own C ~~*~~ \own (R_0, R_1, E)) ~~\land \\
      &(\own M ~~*~~ \own (D, R_D, F_0, F_1)) \\
    \end{aligned}}
\end{align*}
Suppose this triple holds.
Then so does the following triple, where we have dropped the first conjunct of the postcondition:
\begin{align*}
  &\assertc{\own (C, M, M)} \\
  &\hspace{1em}{\ref{prog:ot}(M, M, C)} \\
  &\assertc{(R_0, R_1, D, R_D, R_{\neg D}, E, F_0, F_1, M_C, M_{1-C}, F_{\neg C}).~
      (\own M ~~*~~ \own (D, R_D, F_0, F_1))} \\
\end{align*}
Lilac's semantic model validates the following derived rule:
\begin{lemma}\label{lem:own-compose}
  Let $X$ be a random variable and $f : \mathrm{cod}(X)\to B$ measurable.
  Then $\own X\vdash \own (f\circ X)$.
\end{lemma}
\begin{proof}
  The composition of measurable maps remains measurable.
\end{proof}
Thus we have that $\own (D,R_D, F_0,F_1) \vdash \own (R_D\oplus F_0)$,
so the following triple holds:
\begin{align*}
  &\assertc{\own (C, M, M)} \\
  &\hspace{1em}{\ref{prog:ot}(M, M, C)} \\
  &\assertc{(R_0, R_1, D, R_D, R_{\neg D}, E, F_0, F_1, M_C, M_{1-C}, F_{\neg C}).~
      (\own M ~~*~~ \own (R_D\oplus F_0))} \\
\end{align*}
At this point we transition from working in Lilac to reading off
the meanings of Lilac propositions in our semantic model.
We know that, by forward symbolic execution as in the previous section,
the OT protocol sets
\begin{align*}
  F_0 &\asequal (\ite E{M_0\oplus R_1}{M_1\oplus R_0}) \\
  E &\asequal C\oplus D \\
  R_D &\asequal (\ite D{R_1}{R_0})
\end{align*}
We have set $M_0 = M_1 = M$ and $C = 0$ for the sake of contradiction, so these equations become
\begin{align*}
  F_0 &\asequal M \oplus (\ite E{R_1}{R_0}) \\
  E &\asequal D \\
  R_D &\asequal (\ite D{R_1}{R_0})
\end{align*}
Further substitution gives
\[ F_0 \asequal M\oplus (\ite D{R_1}{R_0})\asequal M\oplus  R_D. \]
Reading the final Hoare triple obtained above in terms of the model, we have that
$\sembr{\ref{prog:ot}}$ is a Markov kernel
whose pushforward along the distribution on $(C,M,M)$
gives a distribution on \[(R_0, R_1, D, R_D, R_{\neg D}, E, F_0, F_1, M_C, M_{1-C}, F_{\neg C})\]
in which $M$ is independent of $(R_D\oplus F_0)$.
But we also have
\[ R_D \oplus F_0 = R_D\oplus (M\oplus R_D) = M \]
by the above deduction, so this Hoare triple asserts
the self-independence of $M$.
This is a contradiction: we have chosen $M$ to be a uniformly-distributed $n$-tuple
of boolean values, which cannot be self-independent.

Intuitively, the special case $C=0$ and $M_0=M_1=M$ that we have chosen
is the situation where both messages offered by the sender are exactly the same
and the receiver always opts to receive message $0$. The failure of input
independence corresponds to the fact that the receiver manages to learn what
message $1$ is. But it only learns message $1$ because in this particular situation
the two messages happen to be the same! Thus this counterexample appears
to be more an issue with this particular specification of perfect secrecy via
input independence than a vulnerability in the OT protocol.

\fi 

\end{document}